\RequirePackage{amsmath,mathrsfs,thmtools,amssymb,mathtools,stmaryrd,bm,float,array,tikz,verbatim,float}
\documentclass{llncs}
\usepackage{enumitem}
\usepackage{thm-restate}
\usepackage[a4paper]{geometry}
\usepackage{cite}
\allowdisplaybreaks
\newcommand{\llb}{\llbracket}
\newcommand{\rrb}{\rrbracket}
\newcommand{\ran}{\rangle}
\newcommand{\lan}{\langle}
\spnewtheorem{notat}{Notation}{\bfseries}{\rmfamily}

\makeatletter

\def\thmt@rst@storecounters#1{%
\vspace{-1.9ex}%
  \bgroup
  \def\@currentlabel{}%
  \@for\thmt@ctr:=\thmt@innercounters\do{%
    \thmt@sanitizethe{\thmt@ctr}%
    \protected@edef\@currentlabel{%
      \@currentlabel
      \protect\def\@xa\protect\csname the\thmt@ctr\endcsname{%
        \csname the\thmt@ctr\endcsname}%
      \ifcsname theH\thmt@ctr\endcsname
        \protect\def\@xa\protect\csname theH\thmt@ctr\endcsname{%
          (restate \protect\theHthmt@dummyctr)\csname theH\thmt@ctr\endcsname}%
      \fi
      \protect\setcounter{\thmt@ctr}{\number\csname c@\thmt@ctr\endcsname}%
    }%
  }%
  \label{thmt@@#1@data}%
  \egroup
}%

\makeatother

\usetikzlibrary{positioning,arrows}
\newcommand\Item[1][]{%
  \ifx\relax#1\relax  \item \else \item[#1] \fi
  \abovedisplayskip=0pt\abovedisplayshortskip=0pt~\vspace*{-\baselineskip}}

\begin{document}

\title{Evidence Logics with Relational Evidence }
\titlerunning{Ev. Log. with Rel. Ev.}  
%
\author{Alexandru Baltag \inst{1} \and Andr\'es Occhipinti Liberman\inst{2}}
\authorrunning{Alexandru Baltag and Andr\'es Occhipinti Liberman} 
%
\tocauthor{Alexandru Baltag and Andr\'es Occhipinti Liberman}
\institute{University of Amsterdam, Amsterdam, Netherlands,\\
\email{A.Baltag@.uva.nl}\\
\and
DTU Compute, Technical University of Denmark, Copenhagen, Denmark,\\
\email{aocc@dtu.dk}}

\maketitle              

\begin{abstract}
We introduce a family of logics for reasoning about \textit{relational evidence}: evidence that involves an orderings of states in terms of their relative plausibility. We provide sound and complete axiomatizations for the logics. We also present several evidential actions and prove soundness and completeness for the associated dynamic logics.
\end{abstract}
\textit{Dynamic evidence logics} \cite{vBDP-evnbdhd-APAL,vanBenthem2011,vBP-dynev,vBP-logdynev,Baltag2016} are logics for reasoning about the evidence and evidence-based beliefs of agents in a dynamic environment. Evidence logics are concerned with scenarios in which an agent collects several pieces of evidence about a situation of interest, from a number of sources, and uses this evidence to form and revise her beliefs about this situation. The agent is typically uncertain about the actual state of affairs, and as a result takes several alternative descriptions of this state as possible (these descriptions are typically called \textit{possible worlds} or \textit{possible states}). The existing evidence logics, i.e., \textit{neighborhood evidence logics} (\textsf{NEL}) \cite{vBDP-evnbdhd-APAL,vanBenthem2011,vBP-dynev,vBP-logdynev,Baltag2016}, have the following features:

\begin{enumerate}[leftmargin=*]
\item \textit{All evidence is `binary'}. Each piece of evidence is modeled as a set of possible states.  This set indicates which states are good candidates for the actual state, and which ones are not, according to the source. Hence the name binary; every state is either a good candidate (`in'), or a bad candidate (`out').
\item \textit{All evidence is equally reliable}. The agent treats all evidence pieces on a par. There is no explicit modeling of the relative \textit{reliability} of pieces of evidence.
\item \textit{One procedure to combine evidence}. The logics developed so far study the evidence and beliefs held by an agent relying on one specific procedure for combining evidence. 
\end{enumerate}

This work presents a family of dynamic evidence logics which we call \textit{relational evidence logics} (\textsf{REL}). Relational evidence logics aim to contribute to the existing work on evidence logic as follows. 

\begin{enumerate}[leftmargin=*]

\item \textit{Relax the assumption that all evidence is binary}. This is accomplished by modeling pieces of evidence by \textit{evidence relations}. Evidence relations are preorders over the set of possible states. The ordering is meant to represent the relative plausibility of states given an evidence item. While a special type of evidence relation -- \textit{dichotomous weak order} -- can be used to represent binary evidence, less `black-and-white' forms of evidence can also be encoded in \textsf{REL} models. 
\item \textit{Model levels of evidence reliability}. In general, not all evidence is equally reliable. Expert advice and gossip provide very different grounds for belief, and a rational agent should weight the evidence that it is exposed to accordingly. To model evidence reliability, we equipped our models with \textit{priority orders}, i.e., orderings of the family of evidence relations according to their relative reliability. Priority orders were introduced in \cite{Andr}, and have already been used in other \textsf{DEL} logics (see, e.g. \cite{Girard2011,Liu2011}). Here, we use them to model the relative reliability of pieces of evidence.
\item \textit{Explore alternative evidence aggregation rules}. Our evidence models come equipped with an aggregator, which merges the available evidence relations into a single relation representing the combined plausibility of the possible states. The beliefs of the agent are then defined on the basis of this combined plausibility order. By focusing on different classes of evidence models, given by their underlying aggregator, we can then compare the logics of belief arising from different approaches to combining evidence.
\end{enumerate}

\section{Relational Evidence Models}

\noindent \textbf{Relational evidence.} \label{rel ev def} 
We call \textit{relational evidence} any type of evidence that induces an ordering of states in terms of their relative plausibility. 
A suitable representation for relational evidence, which we adopt, is given by the class of \textit{preorders}. We call preorders representing relational evidence, \textit{evidence relations}, or \textit{evidence orders}. As is well-known, preorders can represent several meaningful types of orderings, including those that feature incomparable or tied alternatives.

\begin{definition}[Preorder] A \emph{preorder} is a binary relation that is reflexive and transitive. We denote the set of all preorders on a set $X$ by $Pre(X)$. For a preorder $R$ on $X$ and an element $x\in X$, we define the following associated notions: $R[x]\coloneqq \{y\in X \mid Rxy\}$; $R^< \coloneqq \{(x,y)\in X^2 \mid Rxy \text{ and } \neg Ryx\}$; $R^\sim \coloneqq \{(x,y) \in X^2 \mid Rxy \text{ and } Ryx\}$.
\end{definition}

\noindent \textbf{Evidence reliability.}  In general, not all sources are equally trustworthy, so an agent combining evidence may be justified in giving priority to some evidence items over others. Thus, as suggested in \cite{vBP-logdynev}, a next reasonable step in evidence logics is modeling levels of reliability of evidence. One general format for this is given by the \textit{priority graphs} of \cite{Andr}, which have already been used extensively in dynamic epistemic logic (see, e.g., \cite{Girard2011,Liu2011}). In this thesis, we will use the related, yet simpler format of a `priority order', as used in \cite{Baltag_talkingyour,Baltag_prot}, to represent hierarchy among pieces of evidence. Our definition of a priority order is as follows:

\begin{definition}[Priority order] Let $\mathscr{R}$ be a family of evidence orders over $W$. A \textit{priority order for} $\mathscr{R}$ is a preorder $\preceq$ on $\mathscr{R}$. For $R,R'\in \mathscr{R}$, $R\preceq R'$ reads as: ``the evidence order $R'$ has at least the same priority as evidence order $R$''. 
\end{definition}

Intuitively, priority orders tell us which pieces of evidence are more reliable according to the agent. They give the agent a natural way to break stalemates when faced with inconsistent evidence.\\ 

\noindent \textbf{Evidence aggregators.} \label{ev ag} We are interested in modeling a situation in which an agent integrates evidence obtained from multiple sources to obtain and update a combined plausibility ordering, and forms beliefs based on this ordering. When we consider relational evidence with varying levels of priority, a natural way model the process of evidence combination is to define a function that takes as input a family of evidence orders $\mathscr{R}$ together with a priority order $\preceq$ defined on them, and combines them into a plausibility order. The agent's beliefs can then be defined in terms of this output. 

\begin{definition}[Evidence aggregator] Let $W$ be a set of alternatives. Let $\mathcal{W}$ be the set of preorders on $W$. An \emph{evidence aggregator for} $W$ is a function $Ag$ mapping any preordered family $P=\lan \mathscr{R},\preceq\ran$ to a preorder $Ag(P)$ on $W$, where $\emptyset\not\in\mathscr{R}\subseteq \mathcal{W}$ and $\preceq$ is a preorder on $\mathscr{R}$. $\mathscr{R}$ is seen here as a family of evidence orders over $W$, $\preceq$ as a priority order for $\mathscr{R}$, and $Ag(P)$ as an evidence-based plausibility order on $W$.
\end{definition}

At first glance, our definition of an aggregator may seem to impose mild constraints that are met by most natural aggregation functions. However, as it is well-known, the output of some common rules, like the majority rule, may not be transitive (thus not a preorder), and hence they don't count as aggregators. A specific aggregator that \textit{does} satisfy the constraints is the \textit{lexicographic rule}. This aggregator was extensively studied in \cite{Andr}, where it was shown to satisfy several nice aggregative properties. The definition of the aggregator is the following:

\begin{definition} \label{lex def} The (anti-)lexicographic rule is the aggregator $\mathsf{lex}$ given by
\[(w,v)\in \mathsf{lex}(\lan \mathscr{R},\preceq\ran) \text{ iff } \forall R' \in \mathscr{R} \ (R' wv \ \lor \ \exists R \in \mathscr{R} (R'\prec R \wedge R^< wv))\]
\end{definition}

Intuitively, the lexicographic rule works as follows.  Given a particular hierarchy $\preceq$ over a family of evidence $\mathscr{R}$, aggregation is done by giving priority to the evidence orders further up the hierarchy in a compensating way: the agent follows what all evidence orders agree on, if it can, or follows more influential pieces of evidence, in case of disagreement. Other well-known aggregators that satisfy the constraints, but don't make use of the priority structure, are the intersection rule (defined below), or the Borda rule.

\begin{definition} \label{cap def} The \textit{intersection rule} is the aggregator $Ag_\cap$ given by $(w,v)\in Ag_\cap(\lan \mathscr{R},\preceq\ran) \text{ iff } (w,v)\in \bigcap\mathscr{R}$.
\end{definition}

\noindent \textbf{The models.} Having defined relational evidence and evidence aggregators, we are now ready to introduce relational evidence models.

\begin{definition}[Relational evidence model] Let $\mathsf{P}$ be a set of propositional variables. A \emph{relational evidence model} (\textsf{REL} model, for short) is a tuple $M=\langle W, \lan \mathscr{R},\preceq\ran, V, Ag\rangle$ where $W$ is a non-empty set of \emph{states}; $\lan \mathscr{R},\preceq\ran$ is an \textit{ordered family of evidence}, where: $\mathscr{R}$ is a set of evidence orders on $W$ with $W^2\in \mathscr{R}$ and $\preceq$ is a priority order for $\mathscr{R}$; $V:\mathsf{P}\to 2^W$ is a valuation function; $Ag$ is an evidence aggregator for $W$. $M=\langle W, \lan \mathscr{R},\preceq\ran, V, Ag\rangle$ is said to be an \textit{$f$-model} iff $Ag=f$.
\end{definition}

$W^2\in \mathscr{R}$ is called the \textit{trivial evidence order}. It represents the evidence stating that  ``the actual state is in $W$''. This evidence represents full uncertainty and is taken to be always available to the agent as a starting point. 

\subsubsection{Syntax and semantics.} We now introduce a \textit{static} language for reasoning about relational evidence, which we call $\mathscr{L}$. In \cite{Baltag2016}, this language is interpreted over \textsf{NEL} models (there, the language is called  $\mathscr{L}_{\forall\Box\Box_0}$).

\begin{definition}[$\mathscr{L}$] Let $\mathsf{P}$ be a countably infinite set of propositional variables. The language $\mathscr{L}$ is defined by:
\[ \varphi  \Coloneqq  \ p \mid \neg \varphi \mid \varphi\wedge\varphi \mid \Box_0\varphi \mid \Box\varphi \mid \forall \varphi \ \ \ (p\in \mathsf{P})\]
We define $\bot \coloneqq p \wedge \neg p$ and  $\top \coloneqq \neg\bot$. The Boolean connectives $\lor$ and $\to$ are defined in terms of $\neg$ and $\wedge$ in the usual manner. The duals of the modal operators are defined in the following way: $\Diamond_0\coloneqq \neg \Box_0 \neg$, $\Diamond\coloneqq \neg \Box \neg$, $\exists \coloneqq \neg \forall \neg$. 
\end{definition}

The intended interpretation of the modalities is as follows. $\Box_0\varphi$ reads as: `the agent has basic, factive evidence for $\varphi$'; $\Box\varphi$ reads as: `the agent has combined, factive evidence for $\varphi$'. The language $\mathscr{L}$ is interpreted over \textsf{REL} models as follows.

\begin{definition}[Satisfaction] Let $M=\langle W, \lan \mathscr{R},\preceq\ran, V, Ag\rangle$ be an \textsf{REL} model and $w\in W$. The satisfaction relation $\models$ between pairs $(M,w)$ and formulas $\varphi\in\mathscr{L}$ is defined as follows:
\begin{flalign*}
\begin{array}{@{}>{\displaystyle}l@{}>{\displaystyle{}}l@{}>{\displaystyle{{}}}l@{{}}}
M,w\models p & \text{ iff }  w\in V(p) \\
M,w\models \neg \varphi & \text{ iff }  M,w\not\models \varphi \\
M,w\models \varphi \wedge \psi & \text{ iff }  M,w\models \varphi \text{ and } M,w\models \psi \\
M,w\models \Box_0 \varphi & \text{ iff }  \text{there is } R\in \mathscr{R} \text{ such that, for all } v\in W, R wv \text{ implies } M,v\models \varphi  \\
M,w\models \Box \varphi & \text{ iff } \text{for all } v\in W, Ag(\lan \mathscr{R},\preceq\ran)wv \text{ implies } M,v\models \varphi\\ 
M,w\models \forall \varphi & \text{ iff }   \text{for all } v\in W, M,v\models \varphi
\end{array}
&&
\end{flalign*}
\end{definition}

\begin{definition}[Truth map] Let $M=\langle W, \lan \mathscr{R},\preceq\ran, V, Ag\rangle$ be a \textsf{REL} model. We define a \textit{truth map} $\llb \cdot \rrb_M: \mathscr{L}\to 2^W$ given by: $\llb \varphi \rrb_M  = \{w\in W \mid M,w\models \varphi \}$
\end{definition}

Next, we introduce some definable notions of evidence and belief over $\mathsf{REL}$ models. Fix a model $M=\langle W, \lan \mathscr{R},\preceq\ran, V, Ag\rangle$.\\

\subsubsection{Basic (factive) evidence.} We say that a piece of evidence $R\in\mathscr{R}$ \textit{supports} $\varphi$ at $w\in W$ iff $R[w]\subseteq \llb \varphi \rrb_M$. That is, every world that is at least as plausible as $w$ under $R$ satisfies $\varphi$. Using this notion of support, we say that the agent has basic, factive evidence for $\varphi$ at $w\in W$ if there is a piece of evidence $R\in\mathscr{R}$ that supports $\varphi$ at $w$. That is: `the agent has basic evidence for $\varphi$ at $w\in W$' iff $\exists R\in \mathscr{R}(R[w]\subseteq \llb\varphi\rrb_M)$ iff $M,w\models \Box_0\varphi$.
We also have a non-factive version of this notion, which says that the agent has basic evidence for $\varphi$ if there is a piece of evidence $R$ that supports $\varphi$ at \textit{some} state, i.e.: `the agent has basic evidence for $\varphi$ (at any state)' iff  $\exists w (\exists R\in \mathscr{R}(R[w]\subseteq\llb\varphi\rrb_M))$ iff $M,w\models \exists \Box_0\varphi$. We can also have a \textit{conditional} version of basic evidence: `the agent has basic, factive evidence for $\psi$ at $w$, conditional on $\varphi$ being true'. Putting $\Box_0^\varphi\psi \coloneqq \Box_0(\varphi\to\psi)$, we have: `the agent has basic, factive evidence for $\psi$ at $w$, conditional on $\varphi$ being true' iff $\exists R\in\mathscr{R}(\forall v(R wv \Rightarrow ( v\in \llb \varphi \rrb_M \Rightarrow v\in \llb \psi\rrb_M)))$ iff $M,w\models \Box_0^\varphi\psi$. The notion of conditional evidence reduces to that of plain evidence when $\varphi=\top$.

\subsubsection{Aggregated (factive) evidence.} We propose a notion of aggregated evidence based on the output of the aggregator: the agent has aggregated, factive evidence for $\varphi$  at $w\in W$ iff $Ag(\lan \mathscr{R},\preceq\ran)[w]\subseteq \llb\varphi\rrb_M$ iff $M,w\models \Box\varphi$. The non-factive version of the previous notion is as follows: the agent has aggregated evidence for $\varphi$ (at any state) iff  $\exists w (Ag(\lan \mathscr{R},\preceq\ran)[w]\subseteq \llb\varphi\rrb_M)$ iff $M,w\models \exists \Box \varphi$. As we did with basic evidence, we can define a conditional notion of aggregated evidence in $\varphi$ by putting $\Box^\varphi\psi \coloneqq \Box(\varphi\to\psi)$. The unconditional version is given by $\varphi=\top$.

\subsubsection{Evidence-based belief.} The notion of belief we will work with is based on the agent's plausibility order, which in \textsf{REL} models corresponds to the output of the aggregator. As we don't require the plausibility order to be converse-well founded, it may have no maximal elements, which means that Grove's definition of belief may yield inconsistent beliefs. For this reason, we adopt a usual generalization of Grove's definition, which defines beliefs in terms of truth in all `plausible enough' worlds (see, e.g., \cite{vBP-dynev,Baltag2014}). Putting $B\varphi \coloneqq \forall \Diamond \Box \varphi$, we have: the agent believes $\varphi$ (at any state) iff $\forall w( \exists v((w,v)\in Ag(\lan \mathscr{R},\preceq\ran) \text{ and } Ag(\lan \mathscr{R},\preceq\ran)[v]\subseteq \llb\varphi\rrb_M))$ iff $M,w\models \forall\Diamond\Box\varphi$. That is, the agent believes $\varphi$ iff for every state $w\in W$, we can always find a more plausible state $v\in \llb \varphi \rrb_M$, all whose successors are also in $\llb \varphi \rrb_M$. When the plausibility relation is indeed converse well-founded, this notion of belief coincides with Grove's one, while ensuring consistency of belief otherwise. We can also define a notion of conditional belief.  Putting $B^\varphi\psi \coloneqq \forall(\varphi\to\Diamond(\varphi\to (\Box \varphi\to\psi)))$, we have: `the agent believes $\psi$ conditional on $\varphi$ iff $\forall w (w\in \llb \varphi \rrb_M \Rightarrow \exists v (Ag(\lan \mathscr{R},\preceq\ran)wv \text { and } v\in \llb \varphi \rrb_M \text{ and } Ag(\lan \mathscr{R},\preceq\ran)[v]\cap \llb \varphi \rrb_M \subseteq \llb \psi \rrb_M))$ iff $M,w\models B^\varphi\psi$. As before, this conditional notion reduces to that of absolute belief when $\varphi=\top$.

\begin{example}[The diagnosis] Consider an agent seeking medical advice on an ongoing health issue. To keep thing simple, assume that there are four possible diseases: asthma ($a$), allergy ($al$), cold ($c$), and flu ($f$). This can be described by a set $W$ consisting of four possible worlds, $\{w_{a},w_{al},w_{c},w_{f}\}$ and a set of atomic formulas $\{a, al, c, f\}$ (each true at the corresponding world). The agent consults three sources, a medical intern ($IN$), a family doctor ($FD$) and an allergist ($AL$). The doctors inspect the patient, observing fairly non-specific symptoms: cough, no fever, and some inconclusive swelling at an allergen test spot. Given the non-specificity of the symptoms, the doctors can't single out a condition that best explains all they observed. Instead, comparing the diseases in terms of how well they explain the observed symptoms and drawing on their experience, each doctor arrives at a ranking of the possible diseases. Let us denote by $R_{IN}$, $R_{FD}$ and $R_{AL}$ the evidence orders representing the judgment of the intern, family doctor and allergist, respectively, which we assume to be as depicted below. If the agent has no information about how reliable each doctor is, she may just trust them all equally. We can model this by a priority order $\preceq$ over the evidence orders $R_{IN} \sim R_{FD} \sim R_{AL}$ that puts all evidence as equally likely. On the other hand, if the agent knows that the intern is the least experienced of the doctors, she may give consider his evidence as strictly less reliable than the one provided by the other doctors. Similarly, if allergist has a very strong reputation, the agent may wish to give the allergist's judgment strict priority over the rest. We can model this by a different priority order $\preceq'$ given by $R_{IN} \prec' R_{FD} \prec' R_{AL}$ (note that this is meant to be reflexive and transitive). If, e.g., the agent uses the lexicographic rule, we arrive at the following scenarios, with different aggregated evidence depending on the priority order used:

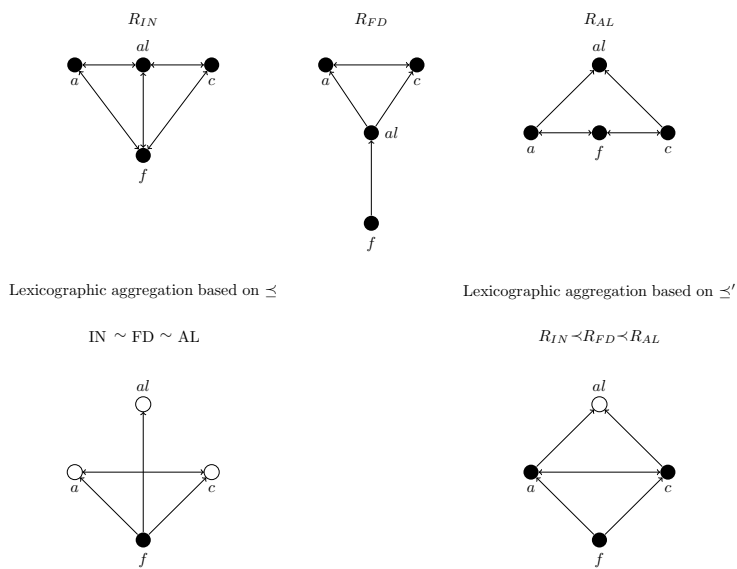
\begin{center}

\scalebox{0.6}{
\begin{tikzpicture}

\node[circle,fill=black][label=below:$a$] (v8) at (7,1) {};
\node[circle,fill=black][label=above:$al$] (v7) at (8.5,1) {};
\node[circle,fill=black][label=below:$c$] (v6) at (10,1) {};
\node[circle,fill=black][label=below:$f$] (v5) at (8.5,-1) {};

\draw [<->] (v5) edge (v8);
\draw [<->] (v8) edge (v7);
\draw [<->] (v7) edge (v6);
\draw [<->] (v6) edge (v5);
\draw [<->] (v7) edge (v5);

\node[circle,fill=black][label=right:$al$] (v4) at (13.5,-0.5) {};
\node[circle,fill=black][label=below:$a$] (v2) at (12.5,1) {};
\node[circle,fill=black][label=below:$c$] (v3) at (14.5,1) {};
\node[circle,fill=black][label=below:$f$] (v1) at (13.5,-2.5) {};

\draw [<->] (v2) edge (v3);
\draw [->] (v4) edge (v2);
\draw [->] (v4) edge (v3);
\draw [->] (v1) edge (v4);

\node[circle,fill=black][label=below:$a$] (v8) at (17,-0.5) {};
\node[circle,fill=black][label=above:$al$] (v7) at (18.5,1) {};
\node[circle,fill=black][label=below:$c$] (v6) at (20,-0.5) {};
\node[circle,fill=black][label=below:$f$] (v5) at (18.5,-0.5) {};

\draw [->] (v8) edge (v7);
\draw [->] (v6) edge (v7);
\draw [<->] (v8) edge (v5);
\draw [<->] (v5) edge (v6);
\node at (8.5,2) {$R_{IN}$};
\node at (13.5,2) {$R_{FD}$};
\node at (18.5,2) {$R_{AL}$};

\node[circle,draw][label=below:$a$] (v8) at (7,-8) {};
\node[circle,draw][label=above:$al$] (v7) at (8.5,-6.5) {};
\node[circle,draw][label=below:$c$] (v6) at (10,-8) {};
\node[circle,fill=black][label=below:$f$] (v5) at (8.5,-9.5) {};

\draw [<->] (v8) edge (v6);
\draw [->] (v5) edge (v8);
\draw [->] (v5) edge (v6);
\draw [->] (v5) edge (v7);

\node (v9) at (7.5,-5) {IN};
\node (v10) at (8.5,-5) {FD};
\node (v11) at (9.5,-5) {AL};

\node at (8.5,-4) {Lexicographic aggregation based on $\preceq$};
\node at (8,-5) {$\sim$};
\node at (9,-5) {$\sim$};

\node[circle,fill=black][label=below:$a$] (v8) at (17,-8) {};
\node[circle,draw][label=above:$al$] (v7) at (18.5,-6.5) {};
\node[circle,fill=black][label=below:$c$] (v6) at (20,-8) {};
\node[circle,fill=black][label=below:$f$] (v5) at (18.5,-9.5) {};

\draw [<->] (v8) edge (v6);
\draw [->] (v5) edge (v8);
\draw [->] (v5) edge (v6);

\node at (18.5,-4) {Lexicographic aggregation based on $\preceq'$};

\node (v9) at (17.5,-5) {$R_{IN}$};
\node (v10) at (18.5,-5) {$R_{FD}$};
\node (v11) at (19.5,-5) {$R_{AL}$};
\node at (18,-5) {$\prec$};
\node at (19,-5) {$\prec$};

\draw [->] (v8) edge (v7);
\draw [->] (v6) edge (v7);

\end{tikzpicture}
}
\end{center}

The best candidates for the actual disease, in each case, are depicted in white. Note that, e.g., the agent has basic evidence for $a\vee al\vee c$, but she doesn't have evidence for $f$. Moreover, in the scenario based on $\preceq'$, the agent believes that the allergy is the actual disease, but she doesn't in the scenario based on $\preceq$.

\end{example}

\subsubsection{A $\mathsf{PDL}$ language for relational evidence} \label{pdl lang rel} Later in this work, we will discuss \textit{evidential actions} by which the agent, upon receiving a new piece of relational evidence, revises its existing body of evidence. To encode syntactically the evidence pieces featured in evidential actions, we will enrich our basic language $\mathscr{L}$ with formulas that stand for specific evidence relations. A natural way to introduce relation-defining expressions, in a modal setting such as ours, is to employ suitable program expressions from Propositional Dynamic Logic ($\mathsf{PDL}$). We will follow this approach, augmenting $\mathscr{L}$ with $\mathsf{PDL}$-style \textit{evidence programs} that define pieces of relational evidence. As evidence orders are preorders, we will employ a set of program expressions whose terms are guaranteed to always define preorders. An natural fragment of $\mathsf{PDL}$ meeting this condition is the one provided by programs of the form $\pi^*$, which always define the reflexive transitive closure of some relation. 

\begin{definition}[Evidence programs] The set $\Pi$ has all program symbols $\pi$ defined as follows:
\[\pi \Coloneqq \ A \mid ?\varphi \mid \pi\cup\pi \mid \pi;\pi \mid \pi^*\]
where $\varphi\in \mathscr{L}$. Here $A$ denotes the universal program, while the rest of the programs have their usual $\mathsf{PDL}$ meanings (see, e.g., \cite{kozen}). We call $\Pi_*\coloneqq \{ \pi^* \mid \pi\in\Pi\}$ the set of \emph{evidence programs}. 
\end{definition}

To interpret evidence programs in \textsf{REL} models, we extend the truth map $\llb \cdot \rrb_M$ as follows:

\begin{definition}[Truth map] Let $M=\langle W, \lan \mathscr{R},\prec\ran, V, Ag\rangle$ be an \textsf{REL} model. We define an extended \textit{truth map} $\llb \cdot \rrb_M: \mathscr{L} \cup \Pi \to 2^W \cup 2^{W^2}$ given by: $ \llb \varphi \rrb_M  = \{w\in W \mid M,w\models \varphi \}$; $\llb A \rrb_M  = W^2$; $ \llb ?\varphi \rrb_M  = \{(w,w)\in W^2 \mid  w\in \llb \varphi \rrb^M\}$; $ \llb \pi \cup \pi' \rrb_M  = \llb \pi  \rrb_M  \cup \llb \pi' \rrb_M$;
$\llb \pi ; \pi' \rrb_M   = \llb \pi  \rrb_M \circ \llb \pi' \rrb_M$; $ \llb \pi^*\rrb_M = \llb \pi  \rrb_M^*$.
\end{definition}

\subsubsection{Some examples of definable evidence programs.} Here are some natural types of relational evidence that can be constructed with programs from $\Pi_*$. \newline

\noindent \textit{Dichotomous evidence}. For a formula $\varphi$, let $\pi_\varphi\coloneqq (A;?\varphi)\cup(?\neg\varphi ; A ; ?\neg \varphi)$. $\pi_\varphi$ puts the $\varphi$ worlds strictly above the $\neg\varphi$ worlds, and makes every world equally plausible within each of these two regions. It is easy to see that $\pi_\varphi$ always defines a preorder, and therefore $(\pi_\varphi)^*$ is an evidence program equivalent to $\pi_\varphi$.
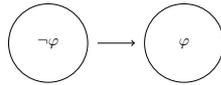
\begin{figure}[H]
\begin{center}
\scalebox{0.65}{
\begin{tikzpicture}

\draw  (-3.25,4) ellipse (0.8 and 0.8);
\draw  (-0.5,4) ellipse (0.8 and 0.8);

\node at (-3.25,4) {$\neg\varphi$};
\node at (-0.5,4) {$\varphi$};

\draw [->](-2.25,4) -- (-1.5,4);
\end{tikzpicture}
}
\end{center}
\caption{The dichotomous order defined by $\pi_\varphi$.}
\end{figure}

\noindent\textit{Totally ordered evidence}. Several programs can be used to define total orders. For example, for formulas $\varphi_1,\dots,\varphi_n$, we can define the program
\begin{flalign*}
\pi^t({\varphi_1,\dots,\varphi_n})\coloneqq(A;?\varphi_1)& \cup  (?\neg\varphi_1 ; A ; ?\neg \varphi_1;?\varphi_2) \\
& \cup (?\neg\varphi_1 ; \neg\varphi_2; A ; ?\neg \varphi_1;?\neg \varphi_2;?\varphi_3)\\ 
& \cup \dots\\
& \cup (?\neg \varphi_1;\dots;?\neg\varphi_n;A;?\neg\varphi_1;\dots;?\neg\varphi_{n-1};?\varphi_n) 
\end{flalign*}
This type of program, described in \cite{vanEijck}, puts the $\varphi_1$ worlds above everything else, the $\neg\varphi_1\wedge \varphi_2$ worlds above the $\neg\varphi_1\wedge\neg\varphi_2$ worlds, and so on, and the $\neg\varphi_1\wedge\neg\varphi_2 \wedge \dots \wedge \neg\varphi_{n-1}\wedge \varphi_n$ above the $\neg\varphi_1\wedge\neg\varphi_2 \wedge \dots \wedge \neg \varphi_n$ worlds. $\pi^t({\varphi_1,\dots,\varphi_n})$ always defines a preorder, so the evidence program $(\pi^t({\varphi_1,\dots,\varphi_n}))^*$ is equivalent to it.

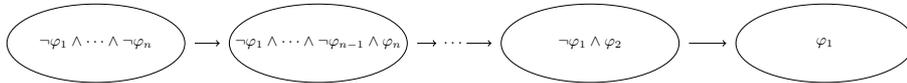
\begin{figure}[H]
\begin{center}
\scalebox{0.65}{
\begin{tikzpicture}

\draw  (-4,4) ellipse (1.8 and 0.8);
\draw  (0.5,4) ellipse (1.8 and 0.8);
\node (v1) at (3,4) {};
\draw  (6,4) ellipse (1.8 and 0.8);
\draw  (10.75,4) ellipse (1.8 and 0.8);
\node at (-4,4) {$\neg{\varphi_1}\wedge \dots \wedge \neg{\varphi_n}$};
\node at (0.5,4) {$\neg{\varphi_1}\wedge \dots \wedge \neg{\varphi_{n-1}}\wedge\varphi_n$};
\node at (6,4) {$\neg{\varphi_1}\wedge \varphi_2$};
\node at (10.75,4) {$\varphi_1$};

\draw [->](-2,4) -- (-1.5,4);
\draw [->](3.5,4) -- (4,4);
\draw [->](8,4) -- (8.75,4);
\draw [->](2.5,4) -- (v1);

\node at (3.25,4) {$\dots$};
\end{tikzpicture}
}
\end{center}
\caption{The total order defined by $\pi^t({\varphi_1,\dots,\varphi_n})$.}
\end{figure}

\noindent \textit{Partially ordered evidence}. Several programs can be used to define evidence orders featuring incomparabilities. To illustrate this, let us consider the program $\pi_{\varphi\wedge \psi} \coloneqq (A;?\varphi\wedge\psi) \cup (?\neg\varphi\wedge \neg\psi ; A ; ?\varphi \vee \psi) \cup (?\neg\varphi\wedge\psi ; A ; ?\neg\varphi\wedge\psi) \cup (?\varphi\wedge\neg\psi ; A ; ?\varphi\wedge\neg\psi)$. As depicted in Figure 3, this program puts the $\varphi\wedge\psi$ worlds above everything else, the $\neg\varphi \wedge\psi$ and $\varphi \wedge\neg \psi$ as incomparable `second-best' worlds, and the $\neg\varphi \wedge\neg\psi$ below everything else. As with the other programs $\pi_{\varphi\wedge \psi}$ always defines a preorder, so $(\pi_{\varphi\wedge \psi})^*$ is an equivalent evidence program.

\begin{figure}[!htb]
\begin{center}
\scalebox{0.7}{
\begin{tikzpicture}

\draw  (-4,4) ellipse (0.8 and 0.8);
\draw  (-1,5) ellipse (0.8 and 0.8);

\draw  (-1,3) ellipse (0.8 and 0.8);
\draw  (2,4) ellipse (0.8 and 0.8);
\node at (-4,4) {$\neg{\varphi}\wedge \neg{\psi}$};
\node at (-1,5) {$\neg{\varphi}\wedge \psi$};
\node at (-1,3) {${\varphi}\wedge \neg\psi$};
\node at (2,4) {$\varphi\wedge \psi$};

\draw [->](-3,4.5) -- (-2,5);

\draw [->](-3,3.5) -- (-2,3);

\draw [->](0,5) -- (1,4.5);

\draw [->](0,3) -- (1,3.5);

\end{tikzpicture}
}
\end{center}
\caption{A partial evidence order.}
\end{figure}
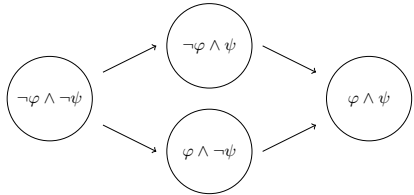

\subsubsection{Normal form programs} We now introduce a `normal form' lemma for the programs in $\Pi$. This lemma shows that, for any evidence program $\pi\in \Pi$, we can find another program $\pi'\in \Pi$, which is a union of certain programs, and which is equivalent to $\pi$. Of special interest for us is the normal form established for programs of the shape $\pi^*$. The fact that every evidence program $\pi^*$ is equivalent to a program with a specific syntactic shape is used extensively in the completeness proofs for dynamic extensions of $\mathscr{L}$ discussed later on. We first introduce some notation and necessary definitions:
\begin{notat} For $\pi_1,\dots,\pi_n\in \Pi$ we write $\bigcup^n_{i=1}\pi_i$ to denote the program $\pi_1\cup\dots\cup\pi_n$. We denote the set of all finite sequences of elements of a set $X$ by $S_0(X)$. Moreover, we denote by $\mathsf{len}(s)$ the length of a sequence $s=(s_1,s_2,\dots,s_{\mathsf{len}(s)})$ and by $s | s'$ the concatenation of sequences $s$ and $s'$. 
\end{notat} 

\begin{definition}[Program equivalence] Two programs $\pi,\pi'\in \Pi$ are \emph{equivalent} iff for every \textsf{REL} model $M$, $\llb \pi \rrb_M = \llb \pi' \rrb_M$. 
\end{definition}

\begin{definition}[Normal form] A \emph{normal form} for a program $\pi\in \Pi$ is a program $\pi'\in\Pi$ such that (1) $\pi'$ has the form $\bigcup_{i\in I}(?\varphi_i;A;?\psi_i)\cup(?\theta)$, where $\varphi_i,\psi_i,\theta\in \mathscr{L}$ and $I$ is a finite index set; and (2) $\pi$ and $\pi'$ are equivalent.
\end{definition}

The proofs for the following lemma, as well as the other results presented in this paper, can be found in the `Proofs Appendix' at the end of the paper.\\

\begin{restatable}[Normal Form Lemma]{lma}{unf}
\label{unf}  Given any program $\pi\in \Pi$, we can find a normal form $\pi'$ for it. In particular, the normal form for program $\pi^*$, where $\pi$ has a normal form $\bigcup_{i\in I}(?\varphi_i;A;?\psi_i)\cup(?\theta)$, is $\bigcup_{s\in S_0(I)}\big(?(\varphi_{s_1} \wedge \bigwedge^{\mathsf{len}(s)}_{k=2}(\exists(\psi_{s_{k-1}} \wedge \varphi_{s_k}))); A ; ?\psi_{s_{\mathsf{len}(s)}})\big) \cup (?\top)$.
 \end{restatable}
 
As the normal form  for $*$-programs is a rather long program, we will generally use the following abbreviation to ease reading. 
 
\begin{notat} 
 Let $I$ be an index set. Let $s\in S_0(I)$ be a sequence. We will use the following abbreviation: $s(\bm{\varphi},\bm{\psi}) \coloneqq (\varphi_{s_1} \wedge \bigwedge^{\mathsf{len}(s)}_{k=2}(\exists(\psi_{s_{k-1}} \wedge \varphi_{s_k})))$. With this abbreviation, the union normal form for $*$-programs will be written as $\bigcup_{s\in S_0(I)}\big(?s(\bm{\varphi},\bm{\psi}); A ; ?\psi_{s_{\mathsf{len}(s)}})\big) \cup (?\top)$.
\end{notat}

\section{The Logics of $Ag_\cap$-Models and $\mathsf{lex}$-Models }

\subsubsection{Static logics.} We initiate here our logical study of the statics of belief and evidence in the \textsf{REL} setting. We first zoom into two specific classes of \textsf{REL} models, the classes of $Ag_\cap$-models and $\mathsf{lex}$-models, and study the static logics for belief and evidence based on these models. In particular, we introduce systems $\mathsf{L}_\cap$ and $\mathsf{L}_\mathsf{lex}$ that axiomatize the class of $Ag_\cap$-models and the class of $\mathsf{lex}$-models, respectively. (To simplify notation, we write $\cap$-models instead of $Ag_\cap$-models hereafter). In later sections, after extending our basic static language appropriately, we will `zoom out' and study the class of all \textsf{REL} models. Our decision to study $\cap$ and $\mathsf{lex}$ models in some detail is as follows. The class of $\cap$-models is interesting because it links our relational evidence setting back to the \textsf{NEL} setting that inspired it. Indeed, as we show below, given any \textsf{NEL} model with finitely many pieces of evidence, we can always find a $\cap$-model that is modally equivalent to it (with respect to language $\mathscr{L}$). This $\cap$-model represents binary evidence in a relational way, thereby encoding the same information present in the \textsf{NEL} model. $\mathsf{lex}$-models, on the other hand, provide a good study case for the \textsf{REL} setting, as they exemplify its main novel features: non-binary evidence and reliability-sensitive aggregation.  We start by recalling the definition of a \textsf{NEL} model. The definition of these models follows the one in \cite{Baltag2016}. For a more general notion, see \cite{vBDP-evnbdhd-APAL}, where the models we consider are called \textit{uniform} models. 

\begin{definition}[Neighborhood evidence model] A \emph{neighborhood evidence model} is a tuple $M=\langle W, E_0, V \rangle$ where: $W$ is a non-empty set of states; $E_0\subseteq \mathscr{P}(W)$ is a family of evidence sets, such that $\emptyset\not\in E_0$ and $W\in E_0$; $V: \mathsf{P} \to \mathscr{P}(W)$ is a valuation function. A  model is called \emph{feasible} if $E_0$ is finite.
\end{definition}

\begin{definition}[Satisfaction] Let $M=\langle W, E_0, V\rangle$ be an \textsf{NEL} model and $w\in W$. The satisfaction relation $\models$ between pairs $(M,w)$ and formulas $\varphi\in\mathscr{L}$ is:
\begin{flalign*}
\begin{array}{@{}>{\displaystyle}l@{}>{\displaystyle{}}l@{}>{\displaystyle{{}}}l@{{}}}
M,w\models \Box_0 \varphi & \text{ iff }  \text{there is } e\in E_0 \text{ such that } w\in e\subseteq \llb \varphi \rrb_M  \\
M,w\models \Box \varphi & \text{ iff } \text{there is } e\in E \text{ such that } w\in e\subseteq \llb \varphi \rrb_M\\ 
M,w\models \forall \varphi & \text{ iff }  W=\llb \varphi \rrb_M
\end{array}
&&
\end{flalign*}
\end{definition}

We now present a way to `transform' a \textsf{NEL} model into a matching \textsf{REL} model. To do that, we first encode binary evidence, the type of evidence considered in \textsf{NEL} models, as relational evidence.

\begin{definition}\label{ev ord assoc} Let $W$ be a set. For each $e\subseteq W$, we denote by $R_e$ the relation given by: $(w,v)\in R_e \text{ iff } w\in e \Rightarrow v\in e$.
\end{definition}

That is, $R_e$ is a preorder with at most two indifference classes (i.e., a dichotomous weak order) of `good' and `bad' candidates for the actual state, which puts all the `bad' candidates strictly below the `good' ones. Having fixed this connection evidence sets and evidence orders, we can now consider a natural way to transform every \textsf{NEL} into a $\cap$-model in which each evidence order is dichotomous. To fix this connection, we define a mapping between \textsf{NEL} and \textsf{REL} models.

\begin{definition} Let $Rel$ be a map from \textsf{NEL} to \textsf{REL} models given by: \[\langle W, E_0, V \rangle \mapsto \langle Rel(W),\lan Rel(E_0),\preceq\ran, Rel(V), Ag_\cap \rangle \] where $Rel(W)\coloneqq W$, $Rel(V)\coloneqq V$ $Rel(E_0)\coloneqq\{R_e \mid e\in E_0\}$ and $\preceq=Rel(E_0)^2$. 
\end{definition}

We can then observe that feasible \textsf{NEL} models and their images under $Rel$ are modally equivalent, in the sense of having point-wise equivalent modal theories.\\

\begin{restatable}{prop}{NELtoREL}
Let $M = \langle W, E_0, V \rangle$ be a feasible \textsf{NEL}  model. For any $\varphi \in \mathscr{L}$ and any $w\in W$, we have: $M,w\models \varphi \emph{ iff } Rel(M),w\models \varphi$.
\end{restatable}

That is, feasible \textsf{NEL} models can be seen as `special cases' of \textsf{REL} models in which all evidence is dichotomous and equally reliable. As the following proposition shows, the modal equivalence result does not extend to non-feasible \textbf{NEL} models. This is because, in models with infinitely many pieces of evidence, the notion of combined evidence presented in \cite{Baltag2016} differs from the one proposed here for \textbf{REL} models. To clarify this, consider a \textbf{NEL} model $M = \langle W, E_0, V \rangle$. Recall that the agent has combined evidence for a proposition $\varphi$ at $w$ if there is a \textit{finite} body of evidence whose combination contains $w$ and supports $\varphi$, i.e., if there is some finite $F\subseteq E_0$ such that $w\in \bigcap F\subseteq \llb \varphi \rrb_M$. Suppose $M$ is a non-feasible model in which we have $w\in \bigcap E_0\subseteq \llb \varphi \rrb_M$, while no finite family $F\subseteq E_0$ is such that $w\in \bigcap F\subseteq \llb \varphi \rrb_M$. That is, the combination of \textit{all} the evidence supports $\varphi$ at $w$, but no combination of a finite subfamily of $E_0$ does. In a \textbf{NEL} model like this, the agent does \textit{not} have combined evidence for $\varphi$ at $w$. That is, $M,w\not\models \Box\varphi$. However, our proposed notion of aggregated evidence for \textbf{REL} models is based on combining \textit{all} the available evidence, and as a result in $Rel(M)$ the agent does have aggregated evidence for $\varphi$ (i.e., $Rel(M),w\models \Box\varphi$). A concrete example of such a model is $M = \langle W, E_0, V \rangle$ with $W=\mathbb{N}$, $E_0=\{\mathbb{N} \setminus\{2n+1\} \mid n\in\mathbb{N}\}$ and $V(p)=\{2n \mid n\in\mathbb{N}\}$. It is easy to verify that $M, 0\not\models p$, while  $Rel(M), 0\models p$.  \\

\begin{restatable}{prop}{NELtoRELbis}  Non-feasible \textbf{NEL} models need not be modally equivalent to their images under $Rel$. In particular, the left-to-right direction of Proposition 1 holds for non-feasible evidence models, but the right-to-left direction doesn't: there are non-feasible neighborhood models $M$ s.t. $Rel(M),w\models \Box \psi$ but $M,w\not\models \Box \psi$.
\end{restatable}

Having motivated our interest in $\cap$-models via their connection to neighborhood evidence logics, we now focus again on the static logics of $\cap$- and $\mathsf{lex}$-models. Table 1 lists the axioms and rules in $\mathsf{L}_\cap$ and $\mathsf{L}_\mathsf{lex}$. \\

\begin{table}[]
\label{systems}
\centering
\caption{The systems $\mathsf{L}_\cap$ and $\mathsf{L}_\mathsf{lex}$}
\begin{tabular}{lllll}
\multicolumn{1}{l|}{\begin{tabular}[c]{@{}l@{}}Axioms and  inference rules\end{tabular}}                                                                 & System(s)                 &  &  &  \\ \cline{1-2}
\multicolumn{1}{l|}{All tautologies of propositional logic}                                                                                                   & both                      &  &  &  \\
\multicolumn{1}{l|}{$\mathsf{S5}$ axioms for $\forall$, $\mathsf{S4}$ axioms for $\Box$, axiom $\mathsf{4}$ for $\Box_0$ }                                                                                                       & both                      &  &  &  \\
\multicolumn{1}{l|}{\begin{tabular}[c]{@{}l@{}}$\forall\varphi\to\Box_0\varphi$   (Universality for $\Box_0$)\end{tabular}}                                 & both                      &  &  &  \\
\multicolumn{1}{l|}{\begin{tabular}[c]{@{}l@{}}$(\Box_0\varphi\wedge\forall\psi)\leftrightarrow\Box_0(\varphi\wedge\forall\psi)$ (Pullout)\end{tabular}} & both                      &  &  &  \\
\multicolumn{1}{l|}{$\Box_0\varphi \to \Box\varphi$}                                                                                                          & $\mathsf{L}_\cap$         &  &  &  \\
\multicolumn{1}{l|}{\begin{tabular}[c]{@{}l@{}}Axioms $\mathsf{T}$ and $\mathsf{N}$ for   $\Box_0$\end{tabular}}                                            & $\mathsf{L}_\mathsf{lex}$ &  &  &  \\
\multicolumn{1}{l|}{\begin{tabular}[c]{@{}l@{}}$\forall\varphi\to\Box\varphi$   (Universality for $\Box$)\end{tabular}}                                     & $\mathsf{L}_\mathsf{lex}$ &  &  &  \\
\multicolumn{1}{l|}{Modus ponens}                                                                                                                             & both                      &  &  &  \\
\multicolumn{1}{l|}{\begin{tabular}[c]{@{}l@{}}Necessitation Rule for $\bullet\in\{\forall,\Box\}$ (from   $\varphi$ infer $\bullet\varphi$)\end{tabular}}                     & both                      &  &  &  \\
\multicolumn{1}{l|}{\begin{tabular}[c]{@{}l@{}}Monotonicity Rule for $\Box_0$ (from   $\varphi\to\psi$ infer $\Box_0\varphi \to \Box_0\psi$)\end{tabular}}  & both                      &  &  &  \\
\end{tabular}
\end{table}

\begin{restatable}{thm}{compcaplex}  $\mathsf{L}_\cap$  and $\mathsf{L}_\mathsf{lex}$ are sound and strongly complete with respect to $\cap$-models and $\mathsf{lex}$-models, respectively.
\end{restatable}

\subsubsection{Evidence dynamics for $\cap$-models.} Having established the soundness and completeness of the static logics, we now turn to evidence dynamics, starting with $\cap$-models. In line with the work on \textsf{NEL}, we consider update, evidence addition and evidence upgrade actions for $\cap$-models. As the intersection rule is insensitive to the priority order, when we consider $\cap$-models, it is convenient to treat the models as if they came with a family of evidence orders $\mathscr{R}$ only, instead of an ordered family $\lan \mathscr{R}, \preceq\ran$. Accordingly, hereafter we will write  $\cap$-models as follows: $M=\langle W, \mathscr{R}, V, Ag_\cap\rangle$. Throughout this section, we fix a $\cap$-model $M=\langle W, \mathscr{R}, V, Ag_\cap\rangle$, some proposition $P\subseteq W$ and some evidence order $R\in Pre(W)$.\\ 

\noindent\textbf{Update.} We first consider updates that involve learning a new fact $P$ with absolute certainty. Upon learning $P$, the agent rules out all possible states that are incompatible with it. For \textsf{REL} models, this means keeping only the worlds in $\llb P\rrb_M$ and restricting each evidence order accordingly.

\begin{definition}[Update] The model $M^{!P}=\langle W^{!P}, \mathscr{R}^{!P}, V^{!P}, Ag^{!P}_\cap \rangle$ has $ W^{!P}\coloneqq P$, $\mathscr{R}^{!P}\coloneqq \{ R\cap P^2 \mid R\in \mathscr{R}\}$, $Ag^{!P}_\cap\coloneqq Ag_\cap$ restricted to $P$, and for all $p\in \mathsf{P}$, $V^{!P}(p)\coloneqq V(p)\cap P$.
\end{definition}

\noindent\textbf{Evidence addition.} Unlike update, which is standardly defined in terms of an incoming proposition $P\subseteq W$, our proposed notion of evidence addition for $\cap$-models involves accepting a new piece of \textit{relational evidence} $R$ from a trusted source. That is, relational evidence addition consists of adding a new piece of relational evidence $R\subseteq Pre(W)$ to the family $\mathscr{R}$. 

\begin{definition}[Evidence addition] The model 
$M^{+R}= \lan W^{+R},\mathscr{R}^{+R}, V^{+R}, Ag^{+R}_\cap\rangle$ has $ W^{+R}\coloneqq W$, $\mathscr{R}^{+R}\coloneqq \mathscr{R}\cup \{R\}$, $V^{+R}\coloneqq V$ and $Ag^{+R}_\cap\coloneqq Ag_\cap$.
\end{definition}

\noindent\textbf{Evidence upgrade.} Finally, we consider an action of upgrade with a piece of relational evidence $R$. This upgrade action is based on the notion of \textit{binary lexicographic merge} from Andr\'eka et. al. \cite{Andr}. 

\begin{definition}[Evidence upgrade] The model $M^{\Uparrow R}=\langle W^{\Uparrow R}, \mathscr{R}^{\Uparrow R}, V^{\Uparrow R}, Ag^{\Uparrow R}_\cap\rangle$ has $ W^{\Uparrow R}\coloneqq W$, $\mathscr{R}^{\Uparrow R}\coloneqq \{R^<\cup (R\cap R') \mid R'\in \mathscr{R}\}$, $V^{\Uparrow R}\coloneqq V$ and $Ag^{\Uparrow R}_\cap\coloneqq Ag_\cap$.
\end{definition}

Intuitively, this operation modifies each existing piece of evidence $R'$ with $R$ following the rule: ``keep whatever $R$ and $R'$ agree on, and where they conflict, give priority to $R$''. To encode syntactically the evidential actions described above, we present extensions of $\mathscr{L}$, obtained by adding to $\mathscr{L}$ dynamic modalities for update, evidence addition and evidence upgrade. The modalities for update will be standard, i.e., modalities of the form $[!\varphi]\psi$. The new formulas of the form $[!\varphi]\psi$ are used to express the statement: ``$\psi$ is true after $\varphi$ is publicly announced''. 

\begin{definition}[$\mathscr{L}^!$] The language $\mathscr{L}^!$ is defined recursively by:
\begin{align*}
\begin{array}{@{}>{\displaystyle}l@{}>{\displaystyle{}}l@{}>{\displaystyle{{}}}l@{{}}}
\varphi & \Coloneqq & \ p \mid \neg \varphi \mid \varphi\wedge\varphi \mid \Box_0\varphi \mid \Box\varphi \mid \forall \varphi \mid [!\varphi]\varphi \ \ \ (p\in\mathsf{P})\end{array}
&&
\end{align*}
\end{definition}

\begin{definition}[Satisfaction for {$[!\varphi]\psi$}] Let $M=\langle W, \mathscr{R}, V, Ag_\cap\rangle$ be a $\cap$-model and $w\in W$. The satisfaction relation $\models$ between pairs $(M,w)$ and formulas $[!\varphi]\psi\in\mathscr{L}^!$ is defined by: $M,w\models[!\varphi]\psi \text{ iff } M,w\models\varphi \text{ implies }  M^{!\llb\varphi\rrb_M},w\models\psi$.
\end{definition}

For the remaining actions, we  extend $\mathscr{L}$ with dynamic modalities of the form $[+\pi]\psi$ for addition and $[\Uparrow\pi]\psi$ for upgrade, where the symbol $\pi$ occurring inside the modality is an evidence program.

\begin{definition}[$\mathscr{L}^\bullet$] Let $\bullet\in\{+,\Uparrow\}$. The language $\mathscr{L}^\bullet$ is defined by:
\begin{align*}
\begin{array}{@{}>{\displaystyle}l@{}>{\displaystyle{}}l@{}>{\displaystyle{{}}}l@{{}}}
\varphi & \Coloneqq & \ p \mid \neg \varphi \mid \varphi\wedge\varphi \mid \Box_0\varphi \mid \Box\varphi \mid \forall \varphi \mid [\bullet\pi^*]\varphi \ \ (p\in\mathsf{P}) \\
\pi & \Coloneqq & \ A \mid ?\varphi \mid \pi\cup\pi \mid \pi;\pi \mid \pi^*  
\end{array}
&&
\end{align*}
\end{definition}

The new formulas of the form $[+\pi]\varphi$ are used to express the statement:  ``$\varphi$ is true after the evidence order defined by $\pi$ is added as a piece of evidence'', while the $[\Uparrow\pi]\varphi$ are used to express: ``$\varphi$ is true after the existing evidence is upgraded with the relation defined by $\pi$''. 
We extend the satisfaction relation $\models$ to cover formulas of the form $[\bullet\pi]\varphi$ as follows:

\begin{definition}[Satisfaction for {$[\bullet\pi]\varphi$}] Let $M=\langle W, \mathscr{R}, V, Ag_\cap\rangle$ be a $\cap$-model, $w\in W$ and $\bullet\in\{+,\Uparrow\}$. The satisfaction relation $\models$ between pairs $(M,w)$ and formulas $[\bullet\pi]\varphi\in\mathscr{L}^\bullet$ is defined by: $M,w\models[\bullet\pi]\varphi \text{ iff } M^{\bullet\llb\pi\rrb_M},w\models\varphi$.
\end{definition}

We now introduce proof systems whose logics are sound and complete with respect to $\cap$-models. The soundness and completeness proofs work via a standard reductive analysis, appealing to \textit{reduction axioms}. We refer to \cite{kooi2004reduction} for an extensive explanation of this technique. 

\begin{definition}[$\mathsf{L}^!$, $\mathsf{L}^+$ and $\mathsf{L}^\Uparrow$] The proof system $\mathsf{L}^!$ extends $\mathsf{L}$ with the following  \emph{reduction axioms}:
\begin{itemize}[leftmargin=*]
\small
\item[] $\text{PA1}_\cap$: $[! \varphi]p \leftrightarrow (\varphi \to p)$ for all $p\in \mathsf{P}$
\item[] $\text{PA2}_\cap$: $[! \varphi]\neg\psi \leftrightarrow (\varphi \to \neg [! \varphi]\psi) $
\item[] $\text{PA3}_\cap$: $[! \varphi](\psi\wedge\psi') \leftrightarrow [! \varphi]\psi \wedge [! \varphi]\psi' $
\item[] $\text{PA4}_\cap$: $[! \varphi]\Box_0\psi \leftrightarrow (\varphi \to \Box_0(\varphi \to [! \varphi]\psi)) $ 
\item[] $\text{PA5}_\cap$: $[! \varphi]\Box\psi \leftrightarrow (\varphi \to \Box(\varphi \to [! \varphi]\psi)) $ 
\item[] $\text{PA6}_\cap$: $[! \varphi]\forall\psi \leftrightarrow (\varphi \to \forall[!\varphi]\psi)$
\end{itemize}
Let $\pi$ be an evidence program with normal form $\bigcup_{s\in S_0(I)}\big(?s(\bm{\varphi},\bm{\psi}); A ; ?\psi_{s_{\mathsf{len}(s)}}\big) \cup (?\top)$. The proof system $\mathsf{L}^+$ extends $\mathsf{L}$ with the following \emph{reduction axioms}:
\begin{itemize}[leftmargin=*]
\small
\item[] $\text{EA1}_\cap$: $[+ \pi]p \leftrightarrow p$ for all $p\in \mathsf{P}$
\item[] $\text{EA2}_\cap$: $[+ \pi]\neg\chi \leftrightarrow \neg [+ \pi]\chi $
\item[] $\text{EA3}_\cap$: $[+ \pi](\chi\wedge\chi') \leftrightarrow [+\pi]\chi \wedge [+ \pi]\chi' $
\item[] $\text{EA4}_\cap$: $[+ \pi]\Box_0\chi \leftrightarrow \Box_0[+ \pi]\varphi \lor \big([+ \pi]\chi \wedge \bigwedge_{s\in S_0(I)}( s(\bm{\varphi},\bm{\psi}) \to \forall (\psi_{s_{\mathsf{len}(s)}} \to [+ \pi]\chi))\big)$
\item[] $\text{EA5}_\cap$: $[+ \pi]\Box\chi \leftrightarrow \big([+ \pi]\chi \wedge \bigwedge_{s\in S_0(I)}( s(\bm{\varphi},\bm{\psi}) \to \Box (\psi_{s_{\mathsf{len}(s)}} \to [+ \pi]\chi))\big)$
\item[] $\text{EA6}_\cap$: $[+ \pi]\forall\chi \leftrightarrow \forall [+ \pi]\chi $
\end{itemize}
The proof system $\mathsf{L}^\Uparrow$ extends $\mathsf{L}$ with the following \emph{reduction axioms}:
\begin{itemize}[leftmargin=*]
\small
\item[] $\text{EU1}_\cap$: $[\Uparrow \pi]p \leftrightarrow p$ for all $p\in \mathsf{P}$
\item[] $\text{EU2}_\cap$: $[\Uparrow \pi]\neg\chi \leftrightarrow \neg [\Uparrow \pi]\chi $
\item[] $\text{EU3}_\cap$: $[\Uparrow \pi](\chi\wedge\chi') \leftrightarrow [\Uparrow\pi]\chi \wedge [\Uparrow \pi]\chi' $
\item[] $\text{EU4}_\cap$: $[\Uparrow \pi]\Box_0\chi \leftrightarrow [\Uparrow\pi]\chi \wedge  \pi^<(\chi) \wedge \pi^\cap(\chi)$
\item[] $\text{EU5}_\cap$ : $[\Uparrow \pi]\Box\chi \leftrightarrow [\Uparrow\pi]\chi \wedge   \pi^<(\chi) \wedge \bigwedge_{s\in S_0(I)}( s(\bm{\varphi},\bm{\psi}) \to \Box (\psi_{s_{\mathsf{len}(s)}} \to [\Uparrow \pi]\chi))$
\item[] $\text{EU6}_\cap$: $[\Uparrow \pi]\forall\chi \leftrightarrow \forall [\Uparrow \pi]\chi$
\end{itemize}
where the formulas $\pi^<(\chi)$ and $\pi^\cap(\chi)$ in ${EU4}_\cap$ and ${EU5}_\cap$ are given by:
\small
\begin{align*}
J(\bm{\varphi})\coloneqq & \bigwedge_{j\in J}\varphi_j \wedge \bigwedge_{j'\in I\setminus J}\neg\varphi_{j'},  J(\bm{\psi})\coloneqq \bigwedge_{j\in J}\psi_j \wedge \bigwedge_{j'\in I\setminus J}\neg\psi_{j'}&&\\
\pi^\cap(\chi)\coloneqq & [\Uparrow \pi]\chi \wedge \bigvee_{J\subseteq I}\big( J(\bm{\varphi}) \wedge \Box_0\big((\bigvee_{\substack{s\in S_0(I) \text{ s.t. } s_1\in J}}(\exists(s(\bm{\varphi},\bm{\psi})) \wedge \psi_{s_{\mathsf{len}(s)}})) \to [\Uparrow\pi]\chi\big)\big) &&\\
\pi^<(\chi)\coloneqq & \bigvee_{J\subseteq I}\big( J(\bm{\psi}) \wedge \mathsf{suc^<}(\chi)\big) &&\\
\mathsf{suc^<}(\chi)\coloneqq & \bigwedge_{s\in S_0(I)}( s(\bm{\varphi},\bm{\psi}) \to \forall\big( (\psi_{s_{\mathsf{len}(s)}} \wedge \bigwedge_{s'\in S_0(I)}( s'(\bm{\varphi},\bm{\psi}) \to \forall( \psi_{s'_{\mathsf{len}(s')}} \to \bigwedge_{j\in J}\neg\varphi_j))) \to [\Uparrow\pi]\chi)\big) &
\end{align*}
\normalsize
\end{definition}

\begin{restatable}{thm}{dynamicscap} $\mathsf{L}^!$, $\mathsf{L}^+$ and $\mathsf{L}^\Uparrow$ are sound and complete with respect to $\cap$-models.
\end{restatable}

\noindent\textbf{Evidence dynamics for $\mathsf{lex}$-models.} We now have a first look at the dynamics of evidence over $\mathsf{lex}$ models. In the \textsf{REL} setting, evidential actions can be seen as complex actions involving two possible transformations on the initial model: (i) modifying the stock of evidence, $\mathscr{R}$, perhaps by adding a new evidence relation $R$ to it, or modifying the existing evidence with $R$; and (ii) updating the priority order, $\preceq$, e.g. to `place' a new evidence item where it fits, according to its reliability. We may also have actions involving evidence, not about the world, but about evidence itself or its sources (sometimes called `higher-order evidence` \cite{CHRISTENSEN2010}), which trigger a reevaluation of the priority order without changing the stock of evidence (for instance, upon learning that a specific source is less reliable than we initially thought, we may want to lower the priority of the evidence provided by this source). To illustrate the type of actions that can be explored in this setting, here we study an action of \textit{prioritized addition} over $\mathsf{lex}$ models. For the sake of generality, we describe this action over \textsf{REL} models. \\

\noindent \textbf{Prioritized addition}. Let $M=\langle W, \lan \mathscr{R}, \preceq\ran, V, Ag \rangle$ be a \textsf{REL} model and $R\in Pre(W)$ a piece of relational evidence. The prioritized addition of $R$ adds $R$ to the set of available evidence $\mathscr{R}$, giving the highest priority to the new evidence. 

\begin{definition}[Prioritized addition] The model $M^{\oplus R}=\langle W^{\oplus R}, \lan \mathscr{R}^{\oplus R}, \preceq^{\oplus R} \ran, V^{\oplus R}, Ag^{\oplus R}\rangle$ has $ W^{\oplus R}\coloneqq W$, $\mathscr{R}^{\oplus R}\coloneqq \mathscr{R}\cup \{R\}$, $V^{\oplus R}\coloneqq V$, $Ag^{\oplus R}\coloneqq Ag$ and $\preceq^{\oplus R}\coloneqq \preceq \cup \{(R',R)\mid R'\in \mathscr{R} \}$.
\end{definition}

To encode prioritized addition, we add formulas of the form $[\oplus \pi]\varphi$, used to express the statement that $\varphi$ is true after the prioritized addition of the evidence order defined by $\pi$.

\begin{definition}[$\mathscr{L}^\oplus$] The language $\mathscr{L}^\oplus$ is given by:
\begin{align*}
\begin{array}{@{}>{\displaystyle}l@{}>{\displaystyle{}}l@{}>{\displaystyle{{}}}l@{{}}}
\varphi & \Coloneqq & \ p \mid \neg \varphi \mid \varphi\wedge\varphi \mid \Box_0\varphi \mid \Box\varphi \mid \forall \varphi \mid [\oplus\pi^*]\varphi \ \ (p\in\mathsf{P}) \\
\pi & \Coloneqq & \ A \mid ?\varphi \mid \pi\cup\pi \mid \pi;\pi \mid \pi^*  
\end{array}
&&
\end{align*}
\end{definition}

\begin{definition}[Satisfaction for {$[\oplus \pi]\varphi$}] Let $M=\langle W, \lan \mathscr{R},\preceq\ran, V, Ag\rangle$ be an \textsf{REL} model and $w\in W$. The satisfaction relation $\models$ between pairs $(M,w)$ and formulas $[\oplus\pi ]\varphi\in\mathscr{L}^\oplus$ is defined as follows: $M,w\models[\oplus \pi]\varphi \text{ iff } M^{\oplus \llb\pi\rrb_M},w\models\varphi$.
\end{definition}

As we did with the dynamic extensions presented for actions in $\cap$-models, we wish to obtain a matching proof system for our dynamic language $\mathscr{L}^\oplus$. We do this via reduction axioms. Before presenting the proof system $\mathsf{L}^\oplus$, we introduce some abbreviations that will be used in the definition of these axioms.

\begin{notat}Let $\pi$ be a normal form $\pi\coloneqq \bigcup_{s\in S_0(I)}\big(?s(\bm{\varphi},\bm{\psi}); A ; ?\psi_{s_{\mathsf{len}(s)}}\big) \cup (?\top)$. For a formula $[\oplus \pi]\chi$, we define the following abbreviations:
\small
\begin{align*}
\pi^<(\chi)\coloneqq & \bigvee_{J\subseteq I}\big( J(\bm{\psi}) \wedge \mathsf{suc^<}(\chi)\big) &&\\
\mathsf{suc^<}(\chi)\coloneqq & \bigwedge_{s\in S_0(I)}( s(\bm{\varphi},\bm{\psi}) \to \forall\big( (\psi_{s_{\mathsf{len}(s)}} \wedge \bigwedge_{s'\in S_0(I)}( s'(\bm{\varphi},\bm{\psi}) \to \forall( \psi_{s'_{\mathsf{len}(s')}} \to \bigwedge_{j\in J}\neg\varphi_j))) \to [\oplus\pi]\chi)\big) &
\end{align*}
\normalsize
\end{notat}

\begin{definition}[$\mathsf{L}^\oplus$] Let $\chi,\chi'\in \mathscr{L}^\oplus$ and let $\pi\in \Pi_*$ be an evidence program with normal form $ \bigcup_{s\in S_0(I)}(?s(\bm{\varphi},\bm{\psi}); A ; ?\psi_{s_{\mathsf{len}(s)}}) \cup (?\top)$. The proof system $\mathsf{L}^\oplus$ extends $\mathsf{L}_0$ with the following  \emph{reduction axioms}:

\begin{itemize}[leftmargin=*]
\small
\item[] $\oplus $EA1: $[\oplus \pi]p \leftrightarrow p$ for all $p\in \mathsf{P}$
\item[] $\oplus $EA2: $[\oplus \pi]\neg\chi \leftrightarrow \neg [\oplus\vec{\pi}]\chi $
\item[] $\oplus $EA3: $[\oplus \pi](\chi\wedge\chi') \leftrightarrow [\oplus \pi]\chi \wedge [\oplus\vec{\pi}]\chi' $
\item[] $\oplus $EA4: $[\oplus \pi]\Box_0\chi \leftrightarrow \Box_0[\oplus \pi]\varphi \lor \big([\oplus \pi]\chi \wedge \bigwedge_{s\in S_0(I)}( s(\bm{\varphi},\bm{\psi}) \to \forall (\psi_{s_{\mathsf{len}(s)}} \to [\oplus \pi]\chi))\big)$
\item[] $\oplus $EA5: $[\oplus \pi]\Box\chi \leftrightarrow [\oplus \pi]\chi \wedge   \pi^<(\chi) \wedge \bigwedge_{s\in S_0(I)}( s(\bm{\varphi},\bm{\psi}) \to \Box (\psi_{s_{\mathsf{len}(s)}} \to [\oplus \pi]\chi))$
\item[] $\oplus $EA6: $[\oplus \pi]\forall\chi \leftrightarrow \forall [\oplus \vec{\pi}]\chi $
\end{itemize}
\end{definition}

\begin{restatable}{thm}{dynamiclex} $\mathsf{L}^\oplus$ is sound and strongly complete with respect to $\mathsf{lex}$ models.
\end{restatable}

\section{The Logic of \textsf{REL} Models}

In this section, we study the logic of evidence and belief based on some abstract aggregator. That is, instead of fixing an aggregator, we are now interested in reasoning about the beliefs that an agent would form, based on her evidence, \textit{irrespective} of the aggregator used. With respect to dynamics, we will focus on the action of \textit{prioritized addition} introduced for $\mathsf{lex}$-models, considering an \textit{iterated} version of prioritized addition, defined with a (possibly empty) sequence of evidence orders $\vec{R}=\lan R_1,\dots, R_n\ran$ as input.

\begin{definition}[Iterated prioritized addition] \label{iterated} Let $M=\langle W, \lan \mathscr{R}, \preceq\ran, V, Ag \rangle$ be a \textsf{REL} model and $\vec{R}=\lan R_1,\dots, R_n\ran$ be a sequence of evidence orders.The model $M^{\oplus R}=\langle W^{\oplus \vec{R}}, \lan \mathscr{R}^{\oplus \vec{R}}, \preceq^{\oplus \vec{R}} \ran, V^{\oplus \vec{R}}, Ag^{\oplus \vec{R}}\rangle$ has $ W^{\oplus \vec{R}}\coloneqq W$, $\mathscr{R}^{\oplus \vec{R}}\coloneqq \mathscr{R}\cup \{R_i \mid i\in \{1,\dots n\}\}$, $V^{\oplus \vec{R}}\coloneqq V$, $Ag^{\oplus \vec{R}}\coloneqq Ag$ and
{\small
\begin{align*}
\preceq^{\oplus \vec{R}}\coloneqq \preceq & \cup \{(R,R_1)\mid R\in \mathscr{R} \} \cup \{(R,R_2)\mid R\in \mathscr{R} \cup \{R_1\}\} &\\
& \cup \dots &\\
& \cup \{(R,R_n)\mid R\in \mathscr{R}\cup \{R_j \mid j\in \{1,\dots, n-1\}\} \}
\end{align*}
}%
\end{definition}
That is, first $R_1$ is added as the highest priority evidence, then $R_2$ is added as the highest priority evidence, on top of every other evidence (including $R_1$), and so on, up to $R_n$. Naturally, when $\vec{R}$ has one element, we are back to the basic notion of prioritized addition.\\ 

\noindent\textbf{Syntax and semantics.} To pre-encode part of the dynamics of iterated prioritized addition, we will modify our basic language $\mathscr{L}$ with \textit{conditional aggregated evidence modalities} of the form $\Box^{\vec{\pi}}$, where $\vec{\pi}$ is a finite, possibly empty sequence of evidence programs $\pi_1,\dots,\pi_n$ (i.e., $\pi_i\in \Pi_*$, for $i\in \{1,\dots n\}$). The intended interpretation of $\Box^{\vec{\pi}}\varphi$ is ``the agent would have aggregated evidence for $\varphi$, if she performed the iterated prioritized addition of the evidence orders defined by $\vec{\pi}$''.

\begin{definition}[$\mathscr{L}_c$] The language $\mathscr{L}_c$ is defined as follows:
\[
\varphi  \Coloneqq  \ p \mid \neg \varphi \mid \varphi\wedge\varphi \mid \Box_0\varphi \mid \Box^{\vec{\pi}}\varphi \mid \forall \varphi \ \  (p\in \mathsf{P});   \ \ \pi  \Coloneqq  \ A \mid ?\varphi \mid \pi\cup\pi \mid \pi;\pi \mid \pi^*  
\]
where $\vec\pi$ is a (possibly empty) finite sequence of evidence programs (i.e. $*$-programs).
\end{definition}

As we allow $\vec{\pi}$ to be empty, $\Box^{\vec{\pi}}$  reduces to the $\Box\varphi$ from $\mathscr{L}$ when $\vec{\pi}$ is the empty sequence, giving us a fully \textit{static} sub-language.

\begin{notat} We abuse the notation for the truth map $\llb \cdot\rrb_M$ and write $\llb \vec{\pi}\rrb_M$ to denote $\lan \llb \pi_1 \rrb_M, \dots,\llb \pi_n \rrb_M \ran$, where $\vec{\pi}=\lan \pi_1,\dots,\pi_n \ran$. 
\end{notat}

Satisfaction for formulas $\Box^{\vec{\pi}}\varphi\in\mathscr{L}_c$ is given as follows.

\begin{definition} Let $M=\langle W, \lan \mathscr{R},\preceq\ran, V, Ag\rangle$ be an \textsf{REL} model and $w\in W$. The satisfaction relation $\models$ between pairs $(M,w)$ and formulas $\Box^{\vec{\pi}}\varphi\in\mathscr{L}_c$ is defined as follows: $M,w\models \Box^{\vec{\pi}} \varphi \text{ iff } Ag(\lan \mathscr{R}^{\oplus \llb \vec{\pi}\rrb_M}, \preceq^{\oplus \llb \vec{\pi}\rrb_M} \ran)[w]\subseteq \llb \varphi\rrb_M$.
\end{definition}

That is, $\Box^{\vec{\pi}} \varphi$ is true at a state $w$ if the agent would have aggregated evidence for $\varphi$, assuming that the current ordered body of evidence is transformed by the iterated prioritized addition of $\llb \vec{\pi}\rrb_M$. Note that, as we allow $\vec{\pi}$ to be empty, $\Box^{\vec{\pi}}$  reduces to the standard $\Box\varphi$ from $\mathscr{L}$ when $\vec{\pi}$ is the empty sequence.\\

\noindent\textbf{Static logic.} Next, we introduce a complete proof system for the language with conditional modalities.

\begin{definition}[$\mathsf{L}_c$] The system $\mathsf{L}_c$ includes the same axioms and inference rules as $\mathsf{L}_\mathsf{lex}$, with axioms and inference rules for $\Box$ in $\mathsf{L}_\mathsf{lex}$ applying to $\Box^ {\vec{\pi}}$ in $\mathsf{L}_c$. 
\end{definition}

\begin{restatable}{thm}{lc} $\mathsf{L}_c$ is sound and strongly complete with respect to \textsf{REL} models. 
\end{restatable}

\noindent\textbf{Evidence dynamics for \textsf{REL} models.} Having established the soundness and completeness of the static logic, we now turn to evidence dynamics, focusing on prioritized evidence addition. To encode prioritized addition, we add formulas of the form $[\oplus \vec{\pi}]\varphi$, used to express the statement that $\varphi$ is true after the prioritized addition of the sequence of evidence orders defined by $\vec{\pi}$.

\begin{definition}[$\mathscr{L}_c^\oplus$] The language $\mathscr{L}_c^\oplus$ is given by:
\[\varphi  \Coloneqq  \ p \mid \neg \varphi \mid \varphi\wedge\varphi \mid \Box_0\varphi \mid \Box^\pi\varphi \mid \forall \varphi \mid [\oplus\vec{\pi}]\varphi \ \ (p\in\mathsf{P}); \ \ \pi  \Coloneqq  \ A \mid ?\varphi \mid \pi\cup\pi \mid \pi;\pi \mid \pi^*\]
where ${\vec{\pi}}$ is a (possibly empty) finite sequence of evidence programs (i.e. $*$-programs).
\end{definition}

\begin{definition}[Satisfaction for {$[\oplus \pi]\varphi$}] Let $M=\langle W, \lan \mathscr{R},\preceq\ran, V, Ag\rangle$ be an \textsf{REL} model and $w\in W$. The satisfaction relation $\models$ between pairs $(M,w)$ and formulas $[\oplus\pi ]\varphi\in\mathscr{L}^\oplus$ is defined as follows: $M,w\models[\oplus {\vec{\pi}}]\varphi \text{ iff } M^{\oplus \llb{\vec{\pi}}\rrb_M},w\models\varphi$.
\end{definition}

Having fixed our dynamic language, we now present reduction axioms for it.

\begin{definition}[$\mathsf{L}_c^\oplus$] Let $\chi,\chi'\in \mathscr{L}_{\vec\pi}^\oplus$ and let $\vec{\pi}=\lan \pi_1,\dots,\pi_n\ran\in S_0(\Pi_*)$ be a sequence of evidence programs where each $\pi_i$ has a normal form $\pi_i \coloneqq \bigcup_{s\in S_0(I_i)}(?s(\bm{\varphi},\bm{\psi}); A ; ?\psi_{s_{\mathsf{len}(s)}}) \cup (?\top)$. The proof system $\mathsf{L}_c^\oplus$ includes all \textit{axioms schemas and inference rules} of $\mathsf{L}_c$, together with the following \emph{reduction axioms}:
\begin{itemize}[leftmargin=0cm]
\small
\item[] PEA1: $[\oplus \vec{\pi}]p \leftrightarrow p$ for all $p\in \mathsf{P}$
\item[] PEA2: $[\oplus \vec{\pi}]\neg\chi \leftrightarrow \neg [\oplus\vec{\pi}]\chi $
\item[] PEA3: $[\oplus \vec{\pi}](\chi\wedge\chi') \leftrightarrow [\oplus \vec{\pi}]\chi \wedge [\oplus\vec{\pi}]\chi' $
\item[] PEA4: $[\oplus \vec{\pi}]\Box_0\chi \leftrightarrow \Box_0[\oplus \vec{\pi}]\varphi \lor \big([+  \vec{\pi}]\chi \wedge \bigvee_{1\leq i\leq n}(\bigwedge_{s\in S_0(I_i)}( s(\bm{\varphi},\bm{\psi}) \to \forall (\psi_{s_{\mathsf{len}(s)}} \to [\oplus \vec{\pi}]\chi)))\big)$
\item[] PEA5: $[\oplus \vec{\pi}]\Box^{\vec{\rho}}\chi \leftrightarrow \Box^{\vec{\pi} \vec{\rho}}\chi$, for  $\vec{\rho}\in S_0(\Pi_*)$
\item[] PEA6: $[\oplus \vec{\pi}]\forall\chi \leftrightarrow \forall [\oplus\vec{\pi}]\chi $
\end{itemize}
where $\vec{\pi}\vec{\rho}$ denotes the concatenation of the sequences $\vec{\pi}$ and $\vec{\rho}$. 
\end{definition}

\begin{restatable}{thm}{lcplus} $\mathsf{L}_c^\oplus$ is sound and complete with respect to \textsf{REL} models.
\end{restatable}

\section{Conclusions and Future Work}

We have presented evidence logics that use a novel representation for evidence and incorporate reliability-sensitive forms of evidence aggregation.  Clearly, many open problems remain. Here are a few more specific avenues for future research:
\begin{itemize}
\item \textit{Additional aggregators}: We studied two natural aggregators. As we know from the social choice literature, many other aggregators have nice properties. An interesting extension to this work could involve developing logics based on other well-known aggregators.
\item \textit{Additional evidential actions}: As we saw, in a setting with ordered evidence, evidence actions are complex transformations, both of the stock of evidence and the priority order. For the lexicographic case, we studied a form of prioritized addition. It could be interesting to consider more general forms of addition, or actions that transform the priority order (re-evaluation of reliability) without affecting the stock of evidence.
\item \textit{Probabilistic evidence}: We moved from the binary evidence case to the relational evidence case. Another important form of evidence is probabilistic evidence, i.e., evidence that comes in the form of a probability distribution over the set of states. The aggregation of probability functions is studied in \textit{probabilistic opinion pooling} \cite{dietrich2013probabilistic} and \textit{pure inductive logic} \cite{paris2015pure}, but a dynamic-logic study has yet to be developed.
\end{itemize}

\pagebreak

\section*{PROOFS APPENDIX}

\subsection*{PROOF OF PROPOSITIONS 1 AND 2}

\NELtoREL*

\begin{proof}
By induction on the structure of $\varphi$. The base case for $\varphi = p$ ($p\in \mathsf{P}$) and the inductive step for $\varphi = \neg \psi$,$\varphi = \psi \wedge \chi$ and $\varphi = \forall \psi$ are shown by unfolding the definitions. We show now the cases involving $\Box_0$ and $\Box$ modalities.
\begin{itemize}
\item $\varphi = \Box_0\psi$. Note that:
\small
\begin{align*}
M,w\models \Box_0\psi \text{ iff } & \text{there is an } e\in E_0 \text{ such that } w\in e \subseteq \llb \psi \rrb_M  \\ 
{\text{ iff }} & \text{there is an } R_e\in Rel(E_0) \text{ such that } R_e[w]=e \text{ and } e\subseteq \llb \psi \rrb_M  \\
\stackrel{i.h.}{\text{ iff }} & \text{there is an } R_e\in Rel(E_0) \text{ such that } R_e[w]\subseteq \llb \psi \rrb_{Rel(M)} \\
 \text{ iff } & Rel(M), w\models \Box_0\psi 
\end{align*}
\normalsize
\item $\varphi = \Box\psi$. Note that:
\small
\begin{align*}
M,w\models \Box\psi \text{ iff }  & \text{ there is an } e\in E \text{ such that } w\in e \subseteq \llb \psi \rrb_M  \\ 
\text{ iff }  & \text{ there are } e_1,\dots,e_n\in E_0 \text{ such that } \bigcap^n_{i=1}e_i = e  \\
& \text{ and }  w\in e \subseteq \llb \psi \rrb_M \\
\text{ iff } & \text{ there are } R_{e_1},\dots,R_{e_n}\in Rel(E_0) \text{ such that } R_{e_i}[w]=e_i \\
& \text{ and }  w\in e \subseteq \llb \psi \rrb_M \\
\text{ iff }  & \text{ there are } R_{e_1},\dots,R_{e_n}\in Rel(E_0) \text{ such that } 
 (\bigcap^n_{i=1}R_{e_i})[w]\subseteq \llb \psi \rrb_M \\
\text{ iff }  &  (\bigcap\mathscr{R})[w] \subseteq (\bigcap^n_{i=1}R_{e_i})[w] \subseteq\llb \psi \rrb_M \\
\stackrel{i.h.}{\text{ iff }} &  (\bigcap\mathscr{R})[w] \subseteq \llb \psi \rrb_{Rel(M)}\\
\text{ iff }  & Rel(M), w\models \Box\psi 
\end{align*}
\normalsize
\end{itemize}
\end{proof}

\NELtoRELbis*

\begin{proof}
That the left-to-right direct holds is clear from the fact that the proofs for this direction don't depend on the cardinality of $E_0$. For the right-to-left direction, the following is a counterexample. Let $M = \langle W, E_0, V \rangle$, with $W=\mathbb{N}$, $E_0=\{\mathbb{N} \setminus\{2n+1\} \mid n\in\mathbb{N}\}$ and $V(p)=\{2n \mid n\in\mathbb{N}\}$. Note that for all $e\in E$, $e\not\subseteq \llb p \rrb_M$ and thus $M,0\not\models \Box p$. 
Moreover, we have:
\[(\bigcap_{R\in Rel(E_0)}R)[0]=  \bigcap_{e\in E_0} (R_e [0])\]
And as $0\in e$ for all $e\in E_0$, by the fact that $0\in e$ implies $R_e[0]=e$, we have $R_e[0]=e$ for each $e\in E_0$. Hence
\[\bigcap_{e\in E_0} (R_e [0])=\bigcap_{e\in E_0} (e)= \bigcap E_0\]
Note that $\bigcap E_0 = \llb p \rrb_M$, and thus, $(\bigcap_{R\in Rel(E_0)}R)[0]= \llb p \rrb_M$. We also have $\llb p \rrb_M= \llb p \rrb_{Rel(M)}$ and hence $(\bigcap_{R\in Rel(E_0)}R)[0]= \llb p \rrb_{Rel(M)}$, which implies $Rel(M), 0\models \Box p$. 
\end{proof}

\subsection*{PROOF OF LEMMA 1}

The following well-known results about relational composition will be used in the normal form lemma.

\begin{restatable}{prop}{} \label{rel dist} Relational composition distributes over arbitrary unions. That is, for any binary relation $R$ and any indexed family of binary relations $Q_i$: (1) $R\circ (\bigcup_i Q_i)= \bigcup_i (R\cup Q_i)$; (2) $(\bigcup_i Q_i) \circ R = \bigcup_i (Q_i\circ R)$.
\begin{proof}
See, e.g., \cite[8]{kozen}.
\end{proof}
\end{restatable}

The following standard \textsf{PDL} facts will also be used.

\begin{restatable}{prop}{} \label{x sees y} Let $M$ be a \textsf{REL} model. Then:
\begin{enumerate}
\item $(x,y)\in \llb ?\varphi;A;?\psi \rrb_M$ iff $x\in \llb \varphi\rrb_M$ and $x\in \llb \psi\rrb_M$.
\item $(x,y)\in \llb ?\varphi;?\psi \rrb_M$ iff $(x,y)\in \llb ?(\varphi\wedge\psi)\rrb_M$.
\item $(x,y)\in \llb ?\varphi_1;A;?(\psi_1\wedge\varphi_2);A;?\psi_2 \rrb_M$ iff $(x,y)\in \llb ?(\varphi_1\wedge \exists (\psi_1\wedge\varphi_2);A;?\psi_2 \rrb_M$
\end{enumerate}
\end{restatable}

In the step of the normal form lemma concerning $*$-programs, we will make use of the following definitions and results.

\begin{definition}[Walks and paths] Let $R\subseteq W\times W$. A \emph{walk} along $R$ is a sequence of (not necessarily distinct) vertices $w_1, w_2, \dots , w_k$, where $w_i\in W$ for $i=1,2,\dots,k$, 
such that $(v_i,v_{i+1}) \in R$ for $i = 1, 2, \dots, k - 1$. A \emph{path} is a walk in which all vertices are distinct (except possibly the first and last). A $wv$-walk is a walk with first vertex $w$ and last vertex $v$. A $wv$-path is defined similarly. The \emph{length} of a walk (path) is its number of edges. 
\end{definition}

\begin{restatable}{prop}{} \label{ray} Let $R\subseteq W\times W$. Every $wv$-walk along $R$ contains a $wv$-path along $R$. 
\end{restatable}
\begin{proof}
This is a standard result. For a proof, see, e.g., \cite[19]{SahaRay2013}. 
\end{proof}

\begin{restatable}{prop}{}  \label{pathshort} Let $M$ be a \textsf{REL} model. Every $wv$-path along $\llb \bigcup^n_{i=1}(?\varphi_i ; A ; ?\psi_i)\rrb_M$ of length $\ell>n$ contains a $wv$-path along $\llb \bigcup^n_{i=1}(?\varphi_i ; A ; ?\psi_i) \rrb_M$ of length at most $n$. \end{restatable}
\begin{proof} Straightforward proof by induction on the length $\ell$ of a $wv$-path $\mathcal{P}=wu_1u_2\dots u_\ell v$.
\end{proof}

\begin{restatable}{prop}{}  \label{star elim} Let $M$ be an \textsf{REL} model. Then $\llb \pi \cup ?\varphi \rrb^*_M= \llb\pi \rrb^*_M$.\end{restatable}
\begin{proof}
Straightforward.
\end{proof}

After fixing the auxiliary results, we get to the proof of Lemma \ref{unf}.

\unf*

\begin{proof}
The proof is by induction on the structure of $\pi$. Let $M$ be any \textsf{REL} model. 
\begin{itemize}
\item $\pi = A$. Let $\pi'$ be the union form $\pi' \coloneqq (?\top;A;?\top) \cup (?\bot)$. It is straightforward to check that $\pi'$ and $\pi$ are equivalent.
\item$\pi = ?\varphi$. Let $\pi'$ be the union form $\pi' \coloneqq  (?\bot;A;?\bot) \cup (?\varphi)$. Again, it is straightforward to check that $\pi'$ and $\pi$ are equivalent.
\item  $\pi = \pi_1\cup \pi_2$. By induction hypothesis, we can find normal forms for $\pi_1$ and $\pi_2$. Let the forms be $\pi'_1 \coloneqq \bigcup_{i\in I}(?\varphi_i;A;?\psi_i)\cup ?\theta$ and $\pi'_2 \coloneqq  \bigcup_{j\in J}(?\varphi_j;A;?\psi_j)\cup ?\theta'$ respectively. Let $\pi'$ be the union form $\pi' \coloneqq  (\bigcup_{k\in I\cup J}?\varphi_k;A;?\psi_k)\cup (?\theta\lor\theta')$. Again, it is straightforward to check that $\pi'$ and $\pi$ are equivalent.
\item $\pi = \pi_1; \pi_2$. By induction hypothesis, we can find normal forms for $\pi_1$ and $\pi_2$. Let the forms be $\pi'_1 \coloneqq \bigcup_{i\in I}?(\varphi_i;A;?\psi_i)\cup ?\theta$ and $\pi'_2 \coloneqq \bigcup_{j\in J}(?\varphi_j;A;?\psi_j)\cup ?\theta'$ respectively. It is not difficult to transform $\pi$ into a normal form shape, essentially using (several times) the fact that composition distributes over union (proposition \ref{rel dist}) to pull the big unions to the left side of the program, and re-indexing the formulas appropriately. 
\item  $\pi = \pi_1^*$. By induction hypothesis, we can find a normal form for $\pi_1$. Let this normal form be $\pi'_1 \coloneqq \bigcup_{i\in I}(?\varphi_i;A;?\psi_i)\cup ?\theta$. We recall here that $S_0(I)$ denotes the set of all finite sequences of elements from $I$, and $\mathsf{len}(s)$ denotes the length of sequence $s=(s_1,s_2,\dots, s_{\mathsf{len}(s)})$. Let $\pi'$ be the union form: 
\[\pi' \coloneqq \bigcup_{s\in S_0(I)}\Big(?(\varphi_{s_1} \wedge \bigwedge^{\mathsf{len}(s)}_{k=2}(\exists(\psi_{s_{k-1}} \wedge \varphi_{s_k}))); A ; ?\psi_{s_{\mathsf{len}(s)}})\Big) \cup (?\top)\]
We will show that $\pi'$ is a normal form for $\pi$. Observe that:
\small
\begin{flalign*}
& (x,y) \in \llb\pi_1^*\rrb_M&&\\
 \text{ iff } & (x,y) \in \llb\pi_1\rrb_M^*\\
\text{ iff } & (x,y) \in \llb\bigcup_{i\in I}(?\varphi_i;A;?\psi_i)\cup ?\theta\rrb_M^*\\
\text{ iff } & (x,y) \in \llb\bigcup_{i\in I}(?\varphi_i;A;?\psi_i)\rrb_M^* \ \ \text{(by proposition. } \ref{star elim})\\
 \text{ iff } & \text{ there is an } xy\text{-walk along }  \llb\bigcup_{i\in I}(?\varphi_i;A;?\psi_i)\rrb_M \text{ of length } \ell, \text{or } x=y\\
\text{ iff } & \text{ there is an } xy\text{-path along }  \llb\bigcup_{i\in I}(?\varphi_i;A;?\psi_i)\rrb_M \text{ of length } \ell', \text{or } x=y \ \ \text{(by proposition. } \ref{ray})\\
\text{ iff } & \text{ there is an } xy\text{-path along }  \llb\bigcup_{i\in I}(?\varphi_i;A;?\psi_i)\rrb_M \text{ of length at most }|I|,  \text{or } x=y \ \ \text{(by proposition. } \ref{pathshort})\\
 \text{ iff } & \text{ for some } s\in S_0(I),\text{there are } z_1,z_2,\dots,z_{\mathsf{len}(s)},z_{\mathsf{len}(s)+1} \text{ such that } z_1=x \text{ and } z_{\mathsf{len}(s)+1}=y \\
 &  \text{ and } \text{for each } k\in \{1,\dots,\mathsf{len}(s)\}, (z_k,z_{k+1})\in \llb?\varphi_{s_k} ; A ; ?\psi_{s_k}\rrb_M, \text{ or } x=y\\
  \text{ iff } & \text{ for some } s\in S_0(I),\text{there are } z_1,z_2,\dots,z_{\mathsf{len}(s)},z_{\mathsf{len}(s)+1} \text{ such that } z_1=x \text{ and } z_{\mathsf{len}(s)+1}=y\\
 & \text{and for each } k\in \{1,\dots,\mathsf{len}(s)\}, z_k\in \llb\varphi_{s_k}\rrb_M \text{ and } z_{k+1}\in \llb\psi_{s_k}\rrb_M \text{ and } y\in \llb \psi_{s_{\mathsf{len}(s)}}\rrb_M \\ 
 & \text{or } x=y \ \text{(by proposition. } \ref{x sees y})\\
\text{ iff } &\text{ for some } s\in S_0(I), x\in \llb\varphi_{s_1} \wedge \bigwedge^{\mathsf{len}(s)}_{k=2}(\exists(\psi_{s_{k-1}} \wedge \varphi_{s_k}))\rrb_M \text{ and } y\in \llb \psi_{s_{\mathsf{len}(s)}}\rrb_M, \text{ or } x=y  \\
\text{ iff } &\text{ for some } s\in S_0(I), (x,y)\in \llb?(\varphi_{s_1} \wedge \bigwedge^{\mathsf{len}(s)}_{k=2}(\exists(\psi_{s_{k-1}} \wedge \varphi_{s_k}))); A ; ?\psi_{s_{\mathsf{len}(s)}})\rrb_M \\ 
& \text{ or } x=y \ \ \ \text{(by proposition. } \ref{x sees y})\\
\text{ iff } & (x,y)\in \bigcup_{s\in S_0(I)} \llb?(\varphi_{s_1} \wedge \bigwedge^{\mathsf{len}(s)}_{k=2}(\exists(\psi_{s_{k-1}} \wedge \varphi_{s_k}))); A ; ?\psi_{s_{\mathsf{len}(s)}})\rrb_M,\text{ or } x=y\\
\text{ iff } & (x,y)\in  \llb\bigcup_{s\in S_0(I)}?\Big((\varphi_{s_1} \wedge \bigwedge^{\mathsf{len}(s)}_{k=2}(\exists(\psi_{s_{k-1}} \wedge \varphi_{s_k}))); A ; ?\psi_{s_{\mathsf{len}(s)}}\Big)\rrb_M,\text{ or } x=y\\
\text{ iff } & (x,y)\in  \llb\bigcup_{s\in S_0(I)}?\Big((\varphi_{s_1} \wedge \bigwedge^{\mathsf{len}(s)}_{k=2}(\exists(\psi_{s_{k-1}} \wedge \varphi_{s_k}))); A ; ?\psi_{s_{\mathsf{len}(s)}}\Big)\rrb_M,\text{ or } (x,y)\in\llb?\top\rrb_M\\
\text{ iff } & (x,y)\in \llb\bigcup_{s\in S_0(I)}\Big(?(\varphi_{s_1} \wedge \bigwedge^{\mathsf{len}(s)}_{k=2}(\exists(\psi_{s_{k-1}} \wedge \varphi_{s_k}))); A ; ?\psi_{s_{\mathsf{len}(s)}})\Big) \cup (?\top)\rrb_M\\
\text{ iff } & (x,y) \in \llb\pi'\rrb_M\\
\end{flalign*}
\normalsize
\end{itemize}
\end{proof}

\subsection*{PROOF OF THEOREM 1}

We recall the theorem:

\compcaplex*

\noindent We prove the soundness and completeness of $\mathsf{L}_\cap$ and $\mathsf{L}_\mathsf{lex}$ separately.

\subsubsection*{Soundness and completeness of $\mathsf{L}_\cap$.} The soundness proof is straightforward; it suffices to check that each axiom is valid and that the inference rules preserve truth. We focus on the completeness proof. The completeness of $\mathsf{L}$ w.r.t $\cap$-models follows from Proposition 1 and the fact that $\mathsf{L}$ is complete, and has the finite model property, w.r.t \textsf{NEL} models (a result from \cite{Baltag2016}). We recall Proposition 1.

\NELtoREL*

Next, we recall a result proven in \cite{Baltag2016}:\\ 

\noindent\textbf{Theorem 6, \cite{Baltag2016}}. \label{comp nel} $\mathsf{L}$ is sound, strongly complete and has the finite model property with respect to the class of \textsf{NEL} models.\\

The theorem above, together with the fact that feasible \textsf{NEL} models are modally equivalent to their images under $Rel$, gives us the completeness of $\mathsf{L}$ w.r.t $\cap$-models. 

\begin{claim}$\mathsf{L}$ is complete w.r.t $\cap$-models.
\begin{proof}
As indicated, e.g., in \cite[194-195]{blackburn2002modal}, a logic $\Lambda$ is strongly complete with respect to a class of models iff every $\Lambda$-consistent set of formulas is satisfiable on some model in this class. Hence, it suffices to show that every $\mathsf{L}$-consistent set of formulas is satisfiable on some $\cap$-model. Let $\Gamma$ be an $\mathsf{L}$-consistent set of formulas. As $\mathsf{L}$ is complete and has the finite model property with respect to \textsf{NEL} models, there is a finite (and hence feasible) \textsf{NEL} model $M$ and a state $w$ in $M$ such that $M,w\models \varphi$ for all $\varphi\in \Gamma$. By proposition 1, we have $Rel(M),w\models \varphi$ for all $\varphi\in \Gamma$. Thus, $\Gamma$ is satisfiable on a $\cap$-model. 
\end{proof}
\end{claim}

\subsubsection*{Soundness and completeness of $\mathsf{L}_\mathsf{lex}$.} The soundness proof is straightforward; we focus on completeness. The approach to the proof is similar to the one used by Fagin et. al. \cite{faginet} to prove completeness for the logic of distributed knowledge. Before going into the details of the proof, we give an outline of the main steps in it. 
\begin{enumerate}[leftmargin=*]
\item Step 1: \textit{Completeness of $\mathsf{L}_\mathsf{lex}$ with respect to pre-models.} First, we define a specific type of canonical \textsf{REL} model for each $\mathsf{L}_\mathsf{lex}$-consistent theory $T_0$, which we call a \textit{pre-model} for $T_0$. Then we prove completeness of $\mathsf{L}_\mathsf{lex}$ via canonical pre-models. 
\item Step 2: \textit{Unraveling}. In the second step, we unravel the canonical pre-model for $T_0$ (see Chapter 4.5 in \cite{blackburn2002modal} for details about this technique). This involves creating all possible histories in the pre-model rooted at $T_0$. The histories are the paths of the canonical pre-model that start at $T_0$. These histories are related in such a way that they form a tree.
\item Step 3: \textit{Completeness of $\mathsf{L}_\mathsf{lex}$ with respect to $\mathsf{lex}$ models}. In the third step, we take the tree we just
constructed, and from we define a $\mathsf{lex}$ model for $T_0$. Then we define a variant of a \textit{bounded morphism} between the canonical pre-model and the $\mathsf{lex}$ model generated from the tree, which makes completeness with respect to those models immediate.
\end{enumerate}

\noindent\textbf{Step 1: Completeness with respect to pre-models.} We build a canonical pre-model for each $\mathsf{L}_\mathsf{lex}$-consistent set of formulas $T_0$. We first fix a standard lemma.

\begin{lemma}[Lindenbaum's Lemma] Every consistent set of formulas of $\mathscr{L}$ can be extended to a maximally consistent one.
\begin{proof}
The proof is a special case of \cite[197]{blackburn2002modal}.
\end{proof}
\end{lemma}

We recall some properties of maximally consistent sets:

\begin{proposition}\label{mcs} Let $T_0$ be a maximally consistent set. The following hold:
\begin{enumerate}
\item For any formula $\varphi$: $\varphi\in T_0$ or $\neg\varphi\in T_0$.
\item $\varphi\in T_0$ iff $T_0 \vdash_{\mathsf{L}_{\mathsf{lex}}} \varphi$.
\item $T_0$ is closed under modus ponens: $\varphi,\varphi\to\psi\in T_0$ implies $\psi\in T_0$.
\item $\neg\varphi\in T_0$ iff $\varphi\not\in T_0$.
\item $\varphi\wedge\psi\in T_0$ iff $\varphi\in T_0$ and $\psi\in T_0$.
\end{enumerate}
\begin{proof}
The proofs are all standard. See, e.g., \cite[53]{chellas}. 
\end{proof}
\end{proposition}

We now define the notion of a canonical pre-model that we will use in the completeness proof of Step 1.

\begin{definition}[Canonical pre-model for $T_0$] \label{canmod} Let $T_0$ be a $\mathsf{L}_\mathsf{lex}$-\emph{consistent} set of formulas. A canonical pre-model for $T_0$ is a structure $M^c=\langle W^c, \lan \mathscr{R}^c, \preceq^c\ran, V^c, Ag^c\rangle$ with:
\small
\begin{itemize}
\item $W^c\coloneqq \{ T \mid T \text{ is a maximally consistent theory and } R^\forall T_0 T \}$.
\item $\mathscr{R}^c\coloneqq \{ R^{\Box_0\varphi} \mid \varphi\in \mathscr{L} \text{ and } (\exists\Box_0\varphi)\in T_0 \}\cup \{R'\}$.
\item $\preceq^c$ is the reflexive closure of $\{(R,R')\mid R\in \mathscr{R}\}$.
\item $V^c$ is a valuation function given by $V^c(p)\coloneqq \| p \|$.
\item $Ag^c$ is an aggregator for $W^c$ given by 
$$
Ag^c(\lan\mathscr{R},\preceq\ran)= 
\begin{cases}
R^{\Box} & \text{ if } \lan\mathscr{R},\preceq\ran=\lan \mathscr{R}^{c}, \preceq^{c}\ran\\
W^c\times W^c & \text{ otherwise }
\end{cases}
$$
\end{itemize}
Where:
\begin{itemize}
\item $R^{\forall}\subseteq W^c \times W^c$ is given by: $R^{\forall}TS$ iff for all $\varphi\in \mathscr{L}$: $(\forall\varphi)\in T \Rightarrow (\forall\varphi)\in S$.
\item for each $\varphi\in \mathscr{L}$, $R^{\Box_0\varphi}\subseteq W^c \times W^c$ is given by: $R^{\Box_0\varphi}TS$ iff $\Box_0\varphi\in T \Rightarrow \Box_0\varphi\in S$.
\item $R^\Box\subseteq W^c \times W^c$ is given by: $R^\Box TS$ iff for all $\varphi\in \mathscr{L}$: $\Box\varphi\in T \Rightarrow \varphi\in S$.
\item $F:W^c\times \mathscr{L} \to W^c$ is a function given by cases:
\begin{itemize}
\item[(a)] for every pair $(T,\varphi)$ such that $(\Box_0 \varphi)\not\in T$, choose some theory $S \in W^c$ such that $\varphi\not\in S$, and put $F(T,\varphi)\coloneqq S$;
\item[(b)] for every pair $(T,\varphi)$ not satisfying the condition of case (a), put $F(T,\varphi)\coloneqq T$.
\end{itemize}
\item $R'\coloneqq (R^\Box \cup \{(T,F(T,\varphi)) \mid T\in W^c \text{ and } \varphi\in \mathscr{L} \})^*$
\item for each $\varphi\in \mathscr{L}$, $\| \varphi \|_{c}\coloneqq \{T\in W^c \mid \varphi\in T \}$
\end{itemize}
\normalsize
\end{definition}

We first need to show that this canonical pre-model is indeed a \textsf{REL} model.

\begin{proposition} \label{is rel mod pre} $M^c$ is a \textsf{REL} model.
\begin{proof} In order to show that $M^c$ is an \textsf{REL} model, we have to show that: 
\begin{enumerate}
\item $\mathscr{R}^c$ is a family of evidence, i.e., every $R\in \mathscr{R}$ is a preorder.
\item $W^c\times W^c\in \mathscr{R}^c$.
\item $R^{\Box}$ is a preorder, and thus $Ag^c$ is well-defined.
\end{enumerate}
For item 1, let $\varphi\in\mathscr{L}$ be arbitrary. Let $R\in \mathscr{R}$ be arbitrary. Then either $R= R'$ or $R=R^{\Box_0\varphi}$ for some $\varphi$. As $R'$ is the reflexive transitive closure of the relation $R^\Box \cup \{(T,F(T,\varphi)) \mid T\in W^c \text{ and } \varphi\in \mathscr{L} \}$, it is a preorder, as required. Now consider $R=R^{\Box_0\varphi}$ for some $\varphi$. The reflexivity of $R$ is immediate from the definition of $R^{\Box_0\varphi}$. For the transitivity, let $T,S,U\in M^c$ and suppose that $R^{\Box_0\varphi}TS$ and $R^{\Box_0\varphi}SU$. Either $\Box_0\varphi\not\in T$ or $\Box_0\varphi\in T$. Note that, by definition of $R^{\Box_0\varphi}$, if $\Box_0\varphi\not\in T$, then $R^{\Box_0\varphi}[T]=W^c$ and thus $R^{\Box_0\varphi}TU$. Suppose now that $\Box_0\varphi\in T$. Then by definition of $R^{\Box_0\varphi}$, given $R^{\Box_0\varphi}TS$ we have $\Box_0\varphi\in S$, and thus as $R^{\Box_0\varphi}SU$ we get $\Box_0\varphi\in U$, which implies $R^{\Box_0\varphi}TU$. For item 2, observe that $\mathsf{N}_\Box$, i.e., $\Box_0\top$, is an axiom of our system. Thus it is a member of any maximal consistent set, which implies that $R^{\Box_0\top}=W^c\times W^c$. Now we consider item 3. For reflexivity, suppose that $(\Box\varphi)\in T$ for some  $T\in M^c$. As $\mathsf{T}_{\Box}$ is an axiom and $T$ is maximal consistent, $(\Box\varphi\to\varphi)\in T$. As $(\Box\varphi)\in T$ and $T$ is closed under modus ponens, we have $\varphi\in T$. Thus $R^\Box T T$. For transitivity, let $T,S,U\in M^c$ and suppose that $R^\Box TS$ and $R^\Box SU$. Suppose $(\Box\varphi)\in T$. As $\mathsf{4}_\Box$ is an axiom and $T$ is maximally consistent, $(\Box\varphi\to\Box\Box\varphi)\in T$. As $(\Box\varphi)\in T$ and $T$ is closed under modus ponens, we have $\Box\Box\varphi\in T$. As $R^\Box T S$, we then have $\Box\varphi\in S$. Hence, as $R^\Box SU$, we have $\varphi\in U$. As $\varphi$ was arbitrary, this holds for each $\varphi$ and hence we have $R^\Box TU$.
\end{proof}
\end{proposition}

Having established that $M^c$ is a \textsf{REL} model, we prove now the standard lemmas to show that the canonical pre-model works as expected.

\begin{lemma}[Existence Lemma for $\forall$] \label{ex forall lex} $\|\exists\varphi\|\neq\emptyset$ iff $\|\varphi\|\neq\emptyset$.
\begin{proof} ($\Rightarrow$). Assume $T\in \|\exists\varphi\|$, i.e., $(\exists\varphi)\in T\in W^c$. We first prove the following:

\begin{claim} The set $\Gamma \coloneqq \{\forall\psi \mid (\forall\psi)\in T\}\cup\{\varphi\}$ is consistent.
\begin{proof}
Suppose that $\Gamma$ is inconsistent, i.e., $\Gamma\vdash_{\mathsf{L}_{\mathsf{lex}}} \bot$. Then there are finitely many sentences $\forall\psi_1,\dots,\forall\psi_n\in T$ such that $\vdash_{\mathsf{L}_{\mathsf{lex}}}\forall\psi_1\wedge\dots\wedge\forall\psi_n\to \neg\varphi$.  By Necessitation for $\forall$ we have $\vdash_{\mathsf{L}_{\mathsf{lex}}} \forall (\forall\psi\wedge\dots\wedge\forall\psi_n\to \neg\varphi)$ and from this, by $\mathsf{K}_\forall$ and modus ponens we get  $\vdash_{\mathsf{L}_{\mathsf{lex}}} \forall(\forall\psi_1\wedge\dots\wedge\forall\psi_n)\to \forall \neg\varphi$.
The system $\mathsf{S5}$ has the theorem $\vdash_{\mathsf{L}_{\mathsf{lex}}} (\forall\forall\psi_1\wedge\dots\wedge\forall\forall\psi_n) \to \forall(\forall\psi_1\wedge\dots\wedge \forall\psi_n)$ (see, e.g., \cite[20]{chellas}). Hence by propositional logic we have $ \vdash_{\mathsf{L}_{\mathsf{lex}}} (\forall\forall\psi_1\wedge\dots\wedge\forall\forall\psi_n)\to \forall \neg\varphi$. Given $\mathsf{4}_\forall$ we have $\vdash_{\mathsf{L}_{\mathsf{lex}}}\forall\psi_1\to\forall\forall\psi_1, \dots, \vdash_{\mathsf{L}_{\mathsf{lex}}}\forall\psi_n\to\forall\forall\psi_n$, which by propositional logic implies $\vdash_{\mathsf{L}_{\mathsf{lex}}}(\forall\psi_1\wedge\dots\wedge\forall\psi_n)\to\forall\forall\psi_1\wedge\dots\wedge\forall\forall\psi_n$. Thus we have $\vdash_{\mathsf{L}_{\mathsf{lex}}}(\forall\psi_1\wedge\dots\wedge\forall\psi_n)\to\forall\neg\varphi$. Hence as $T$ is maximal consistent and closed under modus ponens, we get $(\forall\neg\varphi)\in T$. But we also have $(\exists\varphi)\in T$, i.e., $(\neg \forall \neg \varphi)\in T$, and since $T$ is maximal consistent, this means that $(\forall \neg \varphi)\not\in T$. Contradiction.
\end{proof}
\end{claim}
Given the Claim, by Lindenbaum's Lemma, there is some maximally consistent theory $S$ such that $\Gamma\subseteq S$. As $\varphi\in \Gamma$ we have $\varphi\in S$. Moreover, as $\{\forall\psi \mid (\forall\psi)\in T\}\subseteq \{\forall\chi \mid (\forall\chi)\in S\}$ we have $R^\forall TS$. As $T\in W^c$, we also have $R^\forall T_0 T$. That is,  $\{\forall\theta \mid (\forall\theta)\in T_0\}\subseteq \{\forall\psi \mid (\forall\psi)\in T\}$. Thus $\{\forall\theta \mid (\forall\theta)\in T_0\}\subseteq \{\forall\chi \mid (\forall\chi)\in S\}$ and thus $R^\forall T_0 S$. Hence $S\in W^c$, which together with $\varphi\in S$ gives us $S\in \|\varphi\|$.\\

($\Leftarrow$) Assume $T\in \|\varphi\|$, i.e., $\varphi\in T$. Given $\mathsf{T}_\forall$ we have $\vdash_{\mathsf{L}_{\mathsf{lex}}} \forall\neg\varphi \to \neg\varphi$, and by contraposition we get $\vdash_{\mathsf{L}_{\mathsf{lex}}} \neg\neg \varphi \to\neg\forall\neg\varphi$, i.e.,  $\vdash_{\mathsf{L}_{\mathsf{lex}}} \varphi \to\exists\varphi$. Hence $ (\varphi \to\exists\varphi)\in T$ and as $T$ is closed under modus ponens, given also $\varphi\in T$ we get  $(\exists\varphi)\in T$, i.e., $T\in \|\exists\varphi\|$.
\end{proof}
\end{lemma}

\begin{lemma}[Existence Lemma for $\Box$] \label{ex box lex} $T\in \|\Diamond\varphi\|$ iff there is an $S\in \|\varphi\|$ such that $R^{\Box}TS$.
\begin{proof}($\Rightarrow$). Assume $T\in \|\Diamond\varphi\|$, i.e., $\Diamond\varphi\in T\in W^c$. We first prove the following:

\begin{claim} The set $\Gamma \coloneqq \{\psi \mid (\Box\psi)\in T\}\cup\{\forall\theta \mid \forall\theta\in T_0\}\cup\{\varphi\}$ is consistent.
\begin{proof}
Suppose that $\Gamma$ is inconsistent. Then there is a finite $\Gamma_0\subseteq \Gamma$ such that $\Gamma_0\vdash_{\mathsf{L}_{\mathsf{lex}}} \bot$. By the theorems $\vdash_{\mathsf{L}_{\mathsf{lex}}} \Box(\psi_{i_1}\wedge\dots\wedge\psi_{i_n})\leftrightarrow (\Box\psi_{i_1}\wedge\dots\wedge\Box\psi_{i_n})$ and $\vdash_{\mathsf{L}_{\mathsf{lex}}}\forall(\theta_{j_1}\wedge\dots\wedge\theta_{j_n})\leftrightarrow (\forall\theta_{j_1}\wedge\dots\wedge\forall\theta_{j_n})$ we can assume that $\Gamma_0=\{\Box\psi,\forall\theta,\neg\varphi\}$ for some $\Box\psi,\forall\theta\in T$. That is, we have $\vdash_{\mathsf{L}_{\mathsf{lex}}} \Box\psi \wedge \forall \theta \to \neg\varphi$. By Necessitation for $\Box$ we obtain $\vdash_{\mathsf{L}_{\mathsf{lex}}} \Box(\Box\psi \wedge \forall \theta \to \neg\varphi)$. From this, by $\mathsf{K}_{\Box}$ we get $\vdash_{\mathsf{L}_{\mathsf{lex}}} \Box(\Box\psi \wedge \forall \theta)\to \Box \neg\varphi$. By the theorem   $\vdash_{\mathsf{L}_{\mathsf{lex}}} \Box(\Box\psi \wedge \forall \theta)\leftrightarrow (\Box\Box\psi \wedge\Box\forall \theta)$, from propositional logic we get  $\vdash_{\mathsf{L}_{\mathsf{lex}}} (\Box\Box\psi \wedge\Box\forall \theta)\to \Box \neg\varphi$. Given the axioms in our system we have $\vdash_{\mathsf{L}_{\mathsf{lex}}} \Box\psi\to \Box\Box\psi$ and $\vdash_{\mathsf{L}_{\mathsf{lex}}} \forall(\forall\theta)\to \Box(\forall\theta)$. Using these, by propositional logic we obtain $\vdash_{\mathsf{L}_{\mathsf{lex}}} (\Box\psi \wedge\forall\forall \theta)\to \Box \neg\varphi$. Given our axioms, we also have $\vdash_{\mathsf{L}_{\mathsf{lex}}} \forall\theta\to \forall\forall\theta$. Hence by propositional logic we get $\vdash_{\mathsf{L}_{\mathsf{lex}}} (\Box\psi \wedge\forall \theta)\to \Box \neg\varphi$. As $\Box\psi,\forall\theta\in T$ and $T$ is closed under modus ponens, we get $(\Box\neg\varphi)\in T$. But we also have $(\Diamond\varphi)\in T$, i.e., $(\neg \Box \neg \varphi)\in T$, and since $T$ is maximal consistent, this means that $(\Box \neg \varphi)\not\in T$. Contradiction.
\end{proof}
\end{claim}
Given the Claim, by Lindenbaum's Lemma, there is some maximally consistent theory $S$ such that $\Gamma\subseteq S$. As $\varphi\in \Gamma$ we have $\varphi\in S$. Moreover, as $\{\psi \mid (\Box\psi)\in T\}\subseteq S$, we have $R^{\Box}TS$. Additionally, we have
$\{\forall\theta \mid (\forall\theta)\in T_0\}\subseteq  S$ and thus $R^\forall T_0 S$. Hence $S\in W^c$, which together with $\varphi\in S$ gives us $S\in \|\varphi\|$.\\

($\Leftarrow$) Assume $T\in \|\varphi\|$, i.e., $\varphi\in T$. Given $\mathsf{T}_{\Box}$ we have $\vdash_{\mathsf{L}_{\mathsf{lex}}} \Box\neg\varphi \to \neg\varphi$, and by contraposition we get $\vdash_{\mathsf{L}_{\mathsf{lex}}} \neg\neg \varphi \to\neg\Box\neg\varphi$, i.e.,  $\vdash_{\mathsf{L}_{\mathsf{lex}}} \varphi \to\Diamond\varphi$. Hence, $(\varphi \to\Diamond\varphi)\in T$ and as $T$ is closed under modus ponens, given also $\varphi\in T$ we get $(\Diamond\varphi)\in T$, i.e., $T\in \|\Diamond\varphi\|$.
\end{proof}
\end{lemma}

\begin{lemma}[Existence Lemma for $\Box_0$] \label{ex box0 lex} $T\in \|\Box_0\varphi\|$ iff there is an $R\in\mathscr{R}^c$ such that $R[T]\subseteq \|\varphi\|$.
\begin{proof} ($\Rightarrow$). Assume $T\in \|\Box_0\varphi\|$, i.e., $(\Box_0\varphi)\in T\in W^c$. We first prove the following:

\begin{claim} $\exists\Box_0\varphi\in T_0$.
\begin{proof} Suppose not. As $T_0$ is maximal consistent, we have $\neg\exists\Box_0\varphi\in T_0$, i.e., $\forall\neg\Box_0\varphi\in T_0$. As $T\in W^c$, we have $R^\forall T_0 T$. So given  $\forall\neg\Box_0\varphi\in T_0$ we have  $\forall\neg\Box_0\varphi\in T$. By $\mathsf{T}_\forall$  we have $\vdash_{\mathsf{L}_{\mathsf{lex}}} \forall\neg\Box_0\varphi \to \neg\Box_0\varphi$, i.e., $(\forall\neg\Box_0\varphi \to \neg\Box_0\varphi)\in T$. As $T$ is closed under modus ponens, given $(\forall\neg\Box_0\varphi)\in T$ we get $(\neg\Box_0\varphi)\in T$. But we also have $(\Box_0\varphi)\in T\in W^c$ and thus $T$ is inconsistent. Contradiction.
\end{proof}
\end{claim}
Hence $R^{\Box_0\varphi}\in\mathscr{R}^c$. We will show that $R^{\Box_0\varphi}[T]\subseteq \|\varphi\|$. Let $S\in W^c$ be arbitrary and suppose that $R^{\Box_0\varphi}TS$. By definition of $R^{\Box_0\varphi}$, we have $(\Box_0\varphi)\in T$ implies $(\Box_0\varphi)\in S$. As $(\Box_0\varphi)\in T$ we get $(\Box_0\varphi)\in S$. Given $\mathsf{T}_{\Box_0}$ we have $\vdash_{\mathsf{L}_{\mathsf{lex}}} \Box_0\varphi \to \varphi$ and thus $(\Box_0\varphi \to \varphi)\in S$. Since $S$ is closed under modus ponens we thus get $\varphi\in S$, i.e., $S\in \|\varphi\|$. As $S$ was picked arbitrarily, we have $R^{\Box_0\varphi}[T]\subseteq \|\varphi\|$.\\

($\Leftarrow$) Let $T\in W^c$ and suppose there is an $R\in\mathscr{R}^c$ such that $R[T]\subseteq \|\varphi\|$. By definition of ${R}^c$, (i) $R=R'$ or (ii) $R=R^{\Box_0\theta}$ for some $\theta\in\mathscr{L}$ such that $(\exists\Box_0\theta)\in T_0$.\\ 

We first consider the case (i), i.e., $R=R'$. We need to show that $T\in \|\Box_0\varphi\|$. Suppose not, i.e., $\Box_0\varphi\not\in T$. Recall that, given the definition of $F:W^c\times \mathscr{L} \to W^c$, given that  $(\Box_0 \varphi)\not\in T$, we have $F(T,\varphi)= S$ for some theory $S \in W^c$ such that $\varphi\not\in S$. Moreover, $R'$ is the reflexive transitive closure of the relation $R^\Box \cup \{(T,F(T,\varphi)) \mid T\in W^c \text{ and } \varphi\in \mathscr{L} \}$. Hence $ (T,S)\in \{(T,F(T,\varphi)) \mid T\in W^c \text{ and } \varphi\in \mathscr{L} \}$ and thus $(T,S)\in R'$. But then, as $S\in W^c$ and $\varphi\not\in S$, we have $S\not\in \| \varphi\|$. This implies $R'[T]\not\subseteq \|\varphi\|$, contradicting our assumption to the contrary.\\ 

We now consider case (ii), i.e., $R=R^{\Box_0\theta}$ for some $\theta\in\mathscr{L}$ such that $(\exists\Box_0\theta)\in T_0$. Either $\Box_0\theta\in T$ or $\Box_0\theta\not\in T$. We consider both cases.\\

\noindent\textit{Case 1:} Suppose that $\Box_0\theta\in T\in W^c$. We first prove the following:
\begin{claim} The set $\Gamma \coloneqq \{\Box_0\theta\}\cup\{\forall\psi \mid (\forall\psi)\in T\}\cup\{\neg\varphi\}$ is inconsistent.
\begin{proof}
Suppose that $\Gamma$ is consistent. By Lindenbaum's Lemma there is some maximal consistent theory $S$ such that $\Gamma\subseteq S$. Moreover, as $\{\forall\psi \mid (\forall\psi)\in T_0\}\subseteq\{\forall\psi \mid (\forall\psi)\in T\}\subseteq S$, we have $R^\forall T_0 S$ and thus $S\in W^c$. As $\neg\varphi\in \Gamma$ we have  $\neg\varphi\in S$. Since $S$ is consistent we have $\varphi\not\in S$, i.e., $S\not\in \|\varphi\|$. From $\Box_0\theta\in \Gamma$ we have $\Box_0\theta\in S$. By definition of $R^{\Box_0\theta}$, we get $R^{\Box_0\theta}TS$. But then, given $S\not\in \|\varphi\|$, we have $R^{\Box_0\theta}[T]\not\subseteq \|\varphi\|$. Contradiction.
\end{proof}
\end{claim}
Given the Claim, there is a finite $\Gamma_0\subseteq \Gamma$ such that $\Gamma_0\vdash_{\mathsf{L}_{\mathsf{lex}}} \bot$. By the theorem $\vdash_{\mathsf{L}_{\mathsf{lex}}} \forall(\psi_1\wedge\dots\wedge\psi_n)\leftrightarrow (\forall\psi_1\wedge\dots\wedge\forall\psi_n)$ we can assume that $\Gamma_0=\{\Box_0\theta,\forall\psi,\neg\varphi\}$ for some $\psi\in T$. Since $\Gamma_0\vdash_{\mathsf{L}_{\mathsf{lex}}} \bot$ we have $\vdash_{\mathsf{L}_{\mathsf{lex}}} (\Box_0\theta\wedge\forall\psi\wedge\neg\varphi) \to \bot$, so by propositional logic $\vdash_{\mathsf{L}_{\mathsf{lex}}} (\Box_0\theta\wedge\forall\psi)\to(\neg\varphi \to \bot)$, i.e., $\vdash_{\mathsf{L}_{\mathsf{lex}}} (\Box_0\theta\wedge\forall\psi)\to(\neg\neg\varphi)$, i.e., $\vdash_{\mathsf{L}_{\mathsf{lex}}} (\Box_0\theta\wedge\forall\psi)\to\varphi$. Given the Pullout axiom, we have $\vdash_{\mathsf{L}_{\mathsf{lex}}} \Box_0(\theta\wedge \forall\psi)\to(\Box_0\theta\wedge \forall\psi)$ and thus $\vdash_{\mathsf{L}_{\mathsf{lex}}} \Box_0(\theta\wedge\forall\psi)\to\varphi$. By the Monotonicity Rule for $\Box_0$, we get $\vdash_{\mathsf{L}_{\mathsf{lex}}} \Box_0 \Box_0(\theta\wedge\forall\psi)\to\Box_0\varphi$. By $\mathsf{4}_{\Box_0}$, we have $\vdash_{\mathsf{L}_{\mathsf{lex}}} \Box_0(\theta\wedge\forall\psi)\to \Box_0\Box_0(\theta\wedge\forall\psi)$ and thus $\vdash_{\mathsf{L}_{\mathsf{lex}}} \Box_0 (\theta\wedge\forall\psi)\to\Box_0\varphi$. By the Pullout axiom, we have $\vdash_{\mathsf{L}_{\mathsf{lex}}} (\Box_0\theta\wedge \forall\psi) \to \Box_0(\theta\wedge \forall\psi)$. Hence $\vdash_{\mathsf{L}_{\mathsf{lex}}} (\Box_0 \theta\wedge\forall\psi)\to\Box_0\varphi$. Therefore $((\Box_0 \theta\wedge\forall\psi)\to\Box_0\varphi)\in T$. As $(\Box_0\theta)\in T$ and $(\forall\psi)\in T$, by closure under modus ponens, we have $\Box_0\varphi\in T$. That is, $T\in \| \Box_0\varphi \|$.\\ 

\noindent\textit{Case 2:} Suppose that $\Box_0\theta\not\in T$. 
Note that  $\Box_0\theta\not\in T$ implies that $R^{\Box_0\theta} [T] =W^c$, and since we have $R = R^{\Box_0\theta}$ and $R [T]\subseteq \|\varphi\|$, all this gives us that $W^c \subseteq \|\varphi\|_c$, i.e. all theories in the canonical model contain $\varphi$. We now prove the following:

\begin{claim} The set $\Gamma \coloneqq \{\forall\psi \mid (\forall\psi)\in T\}\cup\{\neg\varphi\}$ is inconsistent.
\begin{proof}
Suppose that $\Gamma$ is consistent. By Lindenbaum's Lemma there is some maximal consistent theory $S$ such that $\Gamma\subseteq S$. Moreover, as $\{\forall\psi \mid (\forall\psi)\in T_0\}\subseteq\{\forall\psi \mid (\forall\psi)\in T\}\subseteq S$, we have $R^\forall T_0 S$ and thus $S\in W^c$. As $\neg\varphi\in \Gamma$ we have  $\neg\varphi\in S$ and thus $S\in \|\neg \varphi\|$.  Therefore $W^c\not\subseteq \|\varphi\|$ (contradiction).
\end{proof}
\end{claim}
Given the Claim, there is a finite $\Gamma_0\subseteq \Gamma$ such that $\Gamma_0\vdash_{\mathsf{L}_{\mathsf{lex}}} \bot$. By the theorem $\vdash_{\mathsf{L}_{\mathsf{lex}}} \forall(\psi_1\wedge\dots\wedge\psi_n)\leftrightarrow (\forall\psi_1\wedge\dots\wedge\forall\psi_n)$ we can assume that $\Gamma_0=\{\forall\psi,\neg\varphi\}$ for some $\psi\in T$. Since $\Gamma_0\vdash_{\mathsf{L}_{\mathsf{lex}}} \bot$ we have $\vdash_{\mathsf{L}_{\mathsf{lex}}} (\forall\psi\wedge\neg\varphi) \to \bot$, so by propositional logic $\vdash_{\mathsf{L}_{\mathsf{lex}}} (\forall\psi)\to(\neg\varphi \to \bot)$, i.e., $\vdash_{\mathsf{L}_{\mathsf{lex}}} (\forall\psi)\to(\neg\neg\varphi)$, i.e., $\vdash_{\mathsf{L}_{\mathsf{lex}}} (\forall\psi)\to\varphi$. By propositional logic, given  $\vdash_{\mathsf{L}_{\mathsf{lex}}} (\forall\psi)\to\varphi$ we can strengthen the antecedent getting $\vdash_{\mathsf{L}_{\mathsf{lex}}} (\Box_0\top\wedge \forall\psi)\to\varphi$. 
Given the $\text{Pullout axiom}^\leftrightarrow$, we have $\vdash_{\mathsf{L}_{\mathsf{lex}}} \Box_0(\top\wedge \forall\psi)\leftrightarrow(\Box_0\top\wedge \forall\psi)$ and thus $\vdash_{\mathsf{L}_{\mathsf{lex}}} \Box_0(\top\wedge\forall\psi)\to\varphi$. By the Monotonicity Rule for $\Box_0$, we get $\vdash_{\mathsf{L}_{\mathsf{lex}}} \Box_0 \Box_0(\top\wedge\forall\psi)\to\Box_0\varphi$. By $\mathsf{4}_{\Box_0}$, we have $\vdash_{\mathsf{L}_{\mathsf{lex}}} \Box_0(\top\wedge\forall\psi)\to \Box_0\Box_0(\top\wedge\forall\psi)$ and thus $\vdash_{\mathsf{L}_{\mathsf{lex}}} \Box_0 (\top\wedge\forall\psi)\to\Box_0\varphi$. By the $\text{Pullout axiom}^\leftrightarrow$, we have $\vdash_{\mathsf{L}_{\mathsf{lex}}} (\Box_0\top\wedge \forall\psi) \leftrightarrow \Box_0(\top\wedge \forall\psi)$. Hence $\vdash_{\mathsf{L}_{\mathsf{lex}}} (\Box_0 \top\wedge\forall\psi)\to\Box_0\varphi$. Therefore $((\Box_0 \top\wedge\forall\psi)\to\Box_0\varphi)\in T$. As $\Box_0\top$ is an axiom of our system, we have $(\Box_0\top)\in T$ and $(\forall\psi)\in T$. Hence by closure under modus ponens, we have $\Box_0\varphi\in T$. That is, $T\in \| \Box_0\varphi \|$. 
\end{proof}
\end{lemma}

\begin{lemma}[Truth Lemma] \label{truth lemma lex} For every $\varphi\in\mathscr{L}$, we have: $\llb \varphi\rrb_{M^c}=\|\varphi \|$.
\begin{proof} The proof is by induction on the complexity of $\varphi$. The base case follows  from the definition of $V^c$. For the inductive case, suppose that for all $T \in W^c$ and all formulas $\psi$ of lower complexity than $\varphi$, we have $\llb \psi\rrb_{M^c}=\|\psi \|$. The Boolean cases where $\varphi = \neg\psi$ and $\varphi = \psi_1 \wedge \psi_2$ follow from the induction hypothesis together with the standard facts about maximal consistent theories included in proposition \ref{mcs}. Only the modalities remain. Let $\varphi= \exists\psi$ and consider any $T\in M^c$. We have $T\in \|\exists\psi\|$ iff (proposition \ref{ex forall lex}) $\|\varphi\|\neq\emptyset$ iff (induction hypothesis) $\llb \psi\rrb_{M^c}$ iff $\llb \exists\psi \rrb_{M^c}=W^c$ iff $T\in \llb \exists\psi \rrb_{M^c}$. Now let $\varphi= \Box_0\psi$ and consider any $T\in M^c$. We have $T\in \|\Box_0\psi\|$ iff (proposition \ref{ex box0 lex}) there is an $R\in\mathscr{R}^c$ such that $R[T]\subseteq \|\psi\|$ iff (induction hypothesis) there is an $R\in\mathscr{R}^c$ such that $R[T]\subseteq \llb \psi \rrb_{M^c}$ iff $T\in \llb \Box_0\psi \rrb_{M^c}$. Let $\varphi= \Diamond\psi$. We have $T\in \|\Diamond\psi\|$ iff (proposition \ref{ex box lex}) there is an $S\in\|\psi\|$ such that $R^{\Box}TS$ iff (induction hypothesis) there is an $S\in\llb \psi \rrb_{M^c}$ such that $R^{\Box}TS$ iff there is an $S\in\llb \psi \rrb_{M^c}$ such that $(T,S)\in Ag^c(\lan \mathscr{R}^{c}, \preceq^{c}\ran)$ iff  $T\in \llb \Diamond\psi \rrb_{M^c}$.
\end{proof}
\end{lemma}

\begin{lemma} \label{lemma comp lex} $\mathsf{L}_\mathsf{lex}$ is strongly complete with respect to the class of pre-models (and hence it is also complete with respect to \textsf{REL} models).
\begin{proof}
It suffices to show that every $\mathsf{L}_\mathsf{lex}$-consistent set of formulas is satisfiable on some $\mathsf{lex}$ model. Let $\Gamma$ be an $\mathsf{L}_\mathsf{lex}$-consistent set of formulas. By Lindenbaum’s Lemma, there is a maximally consistent set $T_0$ such that $\Gamma\subseteq T_0$. Choose any canonical pre-model $M^c$ for $T_0$. By Lemma \ref{truth lemma lex}: $M^c,T_0\models \varphi$ for all $\varphi\in T_0$.
\end{proof}
\end{lemma}

\subsection{Step 2: Unravelling the canonical pre-model}

Next, we will unravel the canonical pre-model. We first fix some preliminary notions. First, define a set of ``evidential indices''
$$\emph{I}\coloneqq\{\Box_0\varphi \mid (\exists\Box_0\varphi)\in T_0\}\cup \{\Box\}\cup \{(\varphi, j)\mid\varphi\in \mathscr{L}, j\in \{l,r\}\},$$
where $l$ is a symbol for ``left'' copy and $r$ is a symbol for ``right'' copy. We use $\epsilon, \epsilon'$ as meta-variables ranging over evidential indices in $\emph{I}$. To each $\epsilon\in \emph{I}$, we associate a corresponding relation $R^{\epsilon}$ on the canonical model $W^c$, as follows: $R^{\Box_0\varphi}$ and $R^{\Box}$ are as before (the relations in the canonical pre-model), and $R^{(\varphi, l)} = R^{(\varphi, r)} \coloneqq \{(T, S)\mid S=F(T, \varphi)\}$. 

\begin{definition}[Histories] Let $M^c=\langle W^c, \lan \mathscr{R}^c, \preceq^c\ran, V^c, Ag^c\rangle$ be a canonical pre-model for $T_0$. The set of histories rooted at $T_0$ is the following set of finite sequences:
$$\tilde{W}\coloneqq\{(T_0, \epsilon_1, T_1, \epsilon_2, \ldots, \epsilon_n, T_n)\mid n\geq 0, \epsilon_i \in \emph{I} \mbox{ and } T_{i-1} R^{\epsilon_i} T_i \mbox{ for all } i\leq n\}$$
The set $\tilde{W}$ forms the set of worlds of the unravelled tree.
\end{definition}

Basically, histories record all finite sequences of worlds in $M^c$ starting with $T_0$ and passing to $R^\epsilon$-successors at each step, where $\epsilon\in I$.

\begin{definition}[$\beta$] We denote by $\beta: \tilde{W}\to W^c$ the map returning the last theory in each history, i.e. $\beta (T_0, \epsilon_1, T_1, \epsilon_2, \ldots, \epsilon_n, T_n)\coloneqq T_n$ for all histories in $\tilde{W}$.
\end{definition}
 
We now define the relations that will feature in the unravelling of $M^c$ around $T_0$. 
 
\begin{definition}[$\to^{\epsilon}$ relations] For a history $w= (T_0, \epsilon_1, T_1, \epsilon_2, \ldots, \epsilon_n, T_n)\in \tilde{W}$, we denote by
$$(w, \epsilon, T)\coloneqq (T_0, \epsilon_1, T_1, \epsilon_2, \ldots, \epsilon_n, T_n, \epsilon, T)$$
the history obtained by extending the history $w$ with the sequence $(\epsilon, T)$ (where $T\in W^c$). Using this notation, we define the following relations $\to^{\epsilon}$ over $\tilde{W}$, labelled by indices in $\emph{I}$:
$$w \to^{\epsilon} w' \,\, \mbox{ iff } \,\, w'=(w, \epsilon, T) \mbox{ for some } T\in W^c$$
\end{definition}

We now define the unravelled tree for $T_0$.

\begin{definition}[Unravelled tree] Let $M^c=\langle W^c, \lan \mathscr{R}^c, \preceq^c\ran, V^c, Ag^c\rangle$ be a canonical pre-model for $T_0$. The unravelling of $M^c$ around $T_0$ is the structure $\tilde{K}=\lan \tilde{W}, \{\to^{\epsilon} \mid \epsilon\in I\}, \tilde{V}\ran$ with
$$\tilde{V}(p)\coloneqq \{w\in \tilde{W}\mid \beta(w)\in V^c(p)\}$$
\end{definition}

In the tree unravelling, one history has another history accessible if the second is one step longer than the first. The valuation on histories is copied from that on their last nodes. We now define paths on this tree of histories.

\begin{definition}[$\mathcal{R}$-path] Let $w,w'\in \tilde{W}$ and let $\mathcal{R}\subseteq \{\to^{\epsilon} \mid \epsilon\in I\}$. An $\mathcal{R}$-path from $w$ to $w'$ is a finite sequence 
\[p=(w_0,\epsilon_1, w_1, \epsilon_2, \dots,\epsilon_{n-1}, w_n) \]
where $w_0=w$, $w_n=w'$, $w_k\in \tilde{W}$ for $k=1,2,\dots,n$, $\epsilon_k\in I$ for $k = 1, 2, \dots, n - 1$ and $w_k\to^{\epsilon_k}w_{k+1}$ for $k = 1, 2, \dots, n - 1$.
For an $\mathcal{R}$-path $p= (w_0,\epsilon_1, w_1, \epsilon_2, \dots,\epsilon_{n-1}, w_n)$ from $w$ to $w'$, we denote by
$$(p, \epsilon, w'')\coloneqq (w_0, \epsilon_1, w_1, \epsilon_2, \ldots, \epsilon_{n-1}, w_n, \epsilon, w'')$$
the path obtained by extending the path $p$ with $(\epsilon, w'')$.
If $\mathcal{R}$ is not specified, we speak of a \textit{path}. For any path $p=(w_0,\epsilon_1, w_1, \epsilon_2, \dots,\epsilon_{n-1}, w_n)$ we define $first(p) = w_0$ and $last(p) = w_n$.
\end{definition}

The following is a standard results about (unravelled) trees, which we will refer to later on.

\begin{lemma}[Uniqueness of paths] \label{path unique}  Let $\tilde{K}$ be the unravelling of $M^c$ around $T_0$. Let $w,w'\in\tilde{W}$ and $\mathcal{R}\subseteq\{\to^{\epsilon} \mid \epsilon\in I\}$. Then, there is at most one $\mathcal{R}$-path $p$ from $w$ to $w'$.
\end{lemma}

\subsection{Step 3: Completeness with respect to $\mathsf{lex}$-models}

Step 2 unravelled the canonical pre-model from step 1. Using the structure from the unravelled tree, we now define a \textsf{REL} model $\tilde{M}$ from it. We then show that this model is in fact a $\mathsf{lex}$-model. Finally, we define a variant of a bounded morphism defined for \textsf{REL} models, which we call \textit{bounded aggregation-morphism}. Bounded aggregation-morphisms work on \textsf{REL} models in the same way as standard bounded-morphisms do on Kripke models: for \textsf{REL} models, modal satisfaction is invariant under bounded aggregation-morphisms. We then show that $M^c$ is a bounded-morphic image of $\tilde{M}$, which gives us completeness. We first define the model $\tilde{M}$. 

\begin{definition}[$\tilde{M}$] Let $\tilde{K}$ be the unravelling of $M^c$ around $T_0$. The structure $\tilde{M}=\lan \tilde{W}, \lan \tilde{\mathscr{R}},\tilde{\preceq}\ran, \tilde{V}, \tilde{Ag}\ran$ has:
$$\tilde{\mathscr{R}}\coloneqq \{\tilde{R}^{\Box_0\varphi}\mid (\exists\Box_0\varphi)\in T_0\}\cup\{R'_l, R'_r\} \cup \{\tilde{W}\times \tilde{W}\}$$
where:
$$\tilde{R}^{\Box_0\varphi}\coloneqq (\to^{\Box_0\varphi})^*$$
$$R'_l = (\to^{\Box}\cup \bigcup\{\to^{(\varphi,l)}\mid\varphi\in \mathscr{L}\})^*$$
$$R'_r = (\to^{\Box}\cup \bigcup\{\to^{(\varphi,r)}\mid\varphi\in \mathscr{L}\})^*$$
Moreover, $\tilde{\preceq}$ is the reflexive closure of $\{(R,Q) \mid R\in\tilde{\mathscr{R}}, Q\in \{R'_l, R'_r\} \}$ 
Finally, the aggregator $\tilde{Ag}$ is given by: 
$$
\tilde{Ag}(\lan\mathscr{R},\preceq\ran)= 
\begin{cases}
\tilde{R}^{\Box} \coloneqq (\to^{\Box})^* & \text{ if } \lan\mathscr{R},\preceq\ran=\lan \tilde{\mathscr{R}}, \tilde{\preceq}\ran\\
\mathsf{lex}(\lan\mathscr{R},\preceq\ran) & \text{ otherwise }
\end{cases}
$$
\end{definition}

\begin{proposition} \label{all preorders} All the evidence relations in $\tilde{\mathscr{R}}\setminus \{\tilde{W}\times \tilde{W}\}$ are reflexive, transitive and anti-symmetric.
\begin{proof}
Reflexivity and transitivity follow immediately from the fact that each $R\in\tilde{\mathscr{R}}\setminus \{\tilde{W}\times \tilde{W}\}$ is the reflexive transitive closure of some other relation, and $\tilde{W}\times \tilde{W}$ is reflexive and transitive. Hence we just need to show the anti-symmetry of the relations. Let $R\in\tilde{\mathscr{R}}\setminus \{\tilde{W}\times \tilde{W}\}$ and suppose $Rwv$ and $Rvw$. First, we consider the case $R= \tilde{R}^{\Box_0\varphi}$ for some $\varphi$ such that $\exists\Box_0\varphi\in T_0$. I.e. $R= (\to^{\Box_0\varphi})^*$. Given $Rwv$ there is some $n\geq 0$ such that:
\[w=w_0 \to^{\Box_0\varphi} w_1 \to^{\Box_0\varphi}\dots \to^{\Box_0\varphi} w_n=v\]
Similarly, given $Rvw$ there is some $m\geq 0$ such that:
\[v=w'_0 \to^{\Box_0\varphi} w'_1 \to^{\Box_0\varphi}\dots \to^{\Box_0\varphi} w'_m=w\]
By definition of $\to^{\Box_0\varphi}$, we have $w_1=(w,\Box_0\varphi,T_1)$ for some $T_1\in W^c$, $w_2=(w,\Box_0\varphi,T_1,\Box_0\varphi,T_2)$ for some $T_2\in W^c$, and proceeding in this way we get
\begin{equation} \label{first eq}
w_n=v=(w,\Box_0\varphi,T_1,\Box_0\varphi,T_2,\dots, \Box_0\varphi, T_n) \text{ where } T_i\in W^c, \text{for }  i\leq n
\end{equation}

Similarly, we have $w'_1=(v,\Box_0\varphi,T'_1)$ for some $T'_1\in W^c$, $w'_2=(v,\Box_0\varphi,T'_1,\Box_0\varphi,T'_2)$ for some $T'_2\in W^c$, and proceeding in this way we get
\begin{equation} \label{sec eq}
w'_m=w=(v,\Box_0\varphi,T'_1,\Box_0\varphi,T'_2,\dots, \Box_0\varphi, T'_n) \text{ where } T'_i\in W^c, \text{for }  i\leq m
\end{equation}
Hence we must have $n=m=0$. For otherwise, substituting $v$ in \ref{sec eq} with the expression in \ref{first eq} we get
$$w=(w,\Box_0\varphi,T_1,\Box_0\varphi,T_2,\dots, \Box_0\varphi, T_n, v,\Box_0\varphi,T'_1,\Box_0\varphi,T'_2,\dots, \Box_0\varphi, T'_m) \text{ for } n>0 \text{ or } m>0$$
which is impossible. Therefore $w=w_0=w_n=v$, as required.\\

Now we consider the case $R= R'_l$, i.e. $R= (\to^{\Box}\cup \bigcup\{\to^{(\varphi,l)}\mid\varphi\in \mathscr{L}\})^*$. Given $Rwv$ there is some $n\geq 0$ such that:
\[w=w_0 \to^{i_0} w_1 \to^{i_2}\dots \to^{i_{n-1}} w_n=v \text{ where } i_k\in \{\Box\}\cup\{(\varphi,l)\mid \varphi\in \mathscr{L}\}, \text{ for } k=1,\dots, n-1\] 
Similarly, given $Rvw$ there is some $m\geq 0$ such that:
\[v=w'_0 \to^{j_0} w'_1 \to^{j_2}\dots \to^{j_{m-1}} w'_n=v \text{ where } j_k\in \{\Box\}\cup\{(\varphi,l)\mid \varphi\in \mathscr{L}\}, \text{ for } k=1,\dots, m-1\] 
Reasoning as we did in the case of $R= \tilde{R}^{\Box_0\varphi}$, we conclude that $m=n=0$ and hence $w=v$. The case of $R=R'_r$ is analogous to the one just discussed, and we are done.

\end{proof}
\end{proposition}

\begin{proposition} \label{intersect l r} In $\tilde{M}$ we have:
$$R'_l \cap R'_d=\tilde{R}^{\Box}$$
\begin{proof}

\end{proof} ($\subseteq$) Let $(w,v)\in R'_l \cap R'_d$. Then we have $(w,v)\in R'_l$, i.e., $$(w,v)\in (\to^{\Box}\cup \bigcup\{\to^{(\varphi,l)}\mid\varphi\in \mathscr{L}\})^*$$
Hence, there is some $n\geq 0$ such that:
\[w=w_0 \to^{i_0} w_1 \to^{i_2}\dots \to^{i_{n-1}} w_n=v \text{ where } i_k\in \{\Box\}\cup\{(\varphi,l)\mid \varphi\in \mathscr{L}\}, \text{ for } k=1,\dots, n-1\] 
Similarly, we have $(w,v)\in R'_r$, i.e., $$(w,v)\in (\to^{\Box}\cup \bigcup\{\to^{(\varphi,r)}\mid\varphi\in \mathscr{L}\})^*$$
Hence, there is some $m\geq 0$ such that:
\[w=w'_0 \to^{j_0} w'_1 \to^{j_2}\dots \to^{j_{m-1}} w'_m=v \text{ where } j_k\in \{\Box\}\cup\{(\varphi,r)\mid \varphi\in \mathscr{L}\}, \text{ for } k=1,\dots, m-1\] 
By definition of $\to^{i_k}$, we have $w_1=(w,i_0,T_1)$ for some $T_1\in W^c$, $w_2=(w,i_0,T_1,i_1,T_2)$ for some $T_2\in W^c$, and proceeding in this way we get
\begin{equation} \label{first eq new}
w_n=v=(w,i_0,T_1,i_1,T_2,\dots, i_{n-1}, T_n) 
\end{equation}
where $T_i\in W^c$ and $i_k\in \{\Box\}\cup\{(\varphi,l)\mid \varphi\in \mathscr{L}\}$, for $k=1,\dots, n-1$. Reasoning in a similar way, we get 
\begin{equation} \label{sec eq new}
w'_m=v=(w,j_0,T'_1,j_1,T'_2,\dots, j_{m-1}, T'_m) 
\end{equation}
where $T'_i\in W^c$ and $j_k\in \{\Box\}\cup\{(\varphi,r)\mid \varphi\in \mathscr{L}\}$, for $k=1,\dots, m-1$.\\

Given the expressions \ref{first eq new} and \ref{sec eq new}, we have $w_n=v=w'_m$. Hence $n=m$ and for all $k<n$, $i_k=j_k$. Hence we must have $i_k=\Box=j_k$ for all $k<n$. This means that the path
\[w=w_0 \to^{i_0} w_1 \to^{i_2}\dots \to^{i_{n-1}} w_n=v\] 
can be rewritten as 
\[w=w_0 \to^{\Box} w_1 \to^{\Box}\dots \to^{\Box} w_n=v\]
which is an $\{\Box\}$-path from $w$ to $v$. Hence $(w,v)\in\tilde{R}^{\Box}=(\to^{\Box})^*$.\\

$(\supseteq)$ Let $(w,v)\in\tilde{R}^{\Box}=(\to^{\Box})^*$. Then there is some $n\geq 0$ such that:
\[w=w_0 \to^{\Box} w_1 \to^{\Box}\dots \to^{\Box} w_n=v\]
The $\{\Box\}$-path described above is also an $\{\Box\}\cup \bigcup\{(\varphi,l)\mid \varphi\in\mathscr{L}\}$-path and an $\{\Box\}\cup \bigcup\{(\varphi,r)\mid \varphi\in\mathscr{L}\}$-path. Hence we have $(w,v)\in R'_l$ and $(w,v)\in R'_r$. Thus $(w,v)\in R'_l \cap R'_r$.
\end{proposition}

\begin{proposition} \label{tilde is lex} $\tilde{M}$ is a $\mathsf{lex}$ model.
\begin{proof} To establish that $\tilde{M}$ is a $\mathsf{lex}$ model, we need to show that it meets the condition of a \textsf{REL} model and that $\tilde{Ag}=\mathsf{lex}$. That is, we have to show:
\begin{enumerate}[leftmargin=*]
\item $\tilde{\mathscr{R}}$ is a family of evidence, i.e., every $R\in \tilde{\mathscr{R}}$ is a preorder.
\item $\tilde{W}\times\tilde{W}\in \tilde{\mathscr{R}}$, i.e., the trivial evidence order is a piece of available evidence.
\item $\tilde{R}^{\Box} = \mathsf{lex}(\lan \tilde{\mathscr{R}}, \tilde{\preceq}\ran)$ (which given the definition of $\tilde{Ag}$, gives $\tilde{Ag}=\mathsf{lex}$ as required)
\end{enumerate}
Item 1 follows from \ref{all preorders}, and Item 2 follows from the definition of $\tilde{M}$. Hence Item 3 remains to be shown. Note first that by \ref{intersect l r}, we have $R'_l \cap R'_d=\tilde{R}^{\Box}$. \\ 

($\subseteq$) Suppose that $(w,v)\in \mathsf{lex}(\lan \tilde{\mathscr{R}}, \tilde{\preceq})$. Note that $\mathsf{lex}$ is given here by
\begin{equation} \label{lex ag can}
(w,v)\in \mathsf{lex}(\lan \tilde{\mathscr{R}},\tilde{\preceq}\ran) \text{ iff } \forall R' \in \mathscr{R} \ (R' wv \ \lor \ \exists R \in \tilde{\mathscr{R}} (R'\tilde{\prec} R \wedge R^< wv))
\end{equation}
Suppose for reductio that $(w,v)\not \in R'_l \cap R'_d$. Then $(w,v)\not \in R'_l$ or $(w,v)\not \in R'_r$. Without loss of generality, suppose $(w,v)\not \in R'_l$. Given \ref{lex ag can}, we have in particular:
\begin{equation} \label{lex ag can2}
(R'_l wv \ \lor \ \exists R \in \tilde{\mathscr{R}} (R'_l\tilde{\prec} R \wedge R^< wv))
\end{equation}
Note that the definition of $\tilde{\preceq}$ is such that $R'_l$ has no relation strictly above it other than $\tilde{W}\times \tilde{W}$. And $\tilde{W}\times \tilde{W}$ is symmetric and thus it is not the case that $(\tilde{W}\times \tilde{W})^<wv$. Hence the right disjunct in \ref{lex ag can2} is false. Therefore we must have $R'_l wv$, contradicting our assumption to the contrary.\\ 

Suppose that $(w,v)\in R'_l \cap R'_d$. Then $R'_l wv$ and $R'_r wv$. Suppose first that $w=v$. As $\mathsf{lex}(\lan \tilde{\mathscr{R}},\tilde{\preceq}\ran)$ is a preorder, we have $(w,v)\in \mathsf{lex}(\lan \tilde{\mathscr{R}},\tilde{\preceq}\ran)$ and we are done. Suppose now that $w\neq v$. By proposition \ref{all preorders}, $R'_l$ and $R'_r$ are antisymmetric. Thus from $w\neq v$, $R'_l wv$ and $R'_r wv$, we get $(R'_l)^< wv$ and $(R'_r)^< wv$. Hence, as we have 
$$R \tilde{\prec} R'_l, R'_r \text{ for all } R\in\tilde{\mathscr{R}}\setminus \{R'_l, R'_r \}$$ 
from the definition of $\mathsf{lex}$ we get $(w,v)\in \mathsf{lex}(\lan \tilde{\mathscr{R}}, \tilde{\preceq})$ as required. 
\end{proof}
\end{proposition}

We now introduce the notion of a bounded aggregation-morphism. This is, as we will show, a truth-preserving map between \textsf{REL} models, which works similarly to standard bounded morphisms for Kripke models.

\begin{definition}[Bounded aggregation-morphism] Let $M=\langle W, \lan \mathscr{R}, \preceq\ran, V, Ag \rangle$ and $M'=\langle W', \lan \mathscr{R}', \preceq'\ran, V', Ag' \rangle$ be two \textsf{REL} models. A mapping $f: W \to W'$ is a bounded aggregation-morphism if the following hold:
\begin{enumerate}
\item Valuation condition: for all $w\in W$, $w\in V(p)$ iff $f(w)\in V'(p)$
\item Forth conditions:
\begin{enumerate}
\item[(a)] for all $R\in {\mathscr R}$, for all $w\in W$, there exists some $R'\in \mathscr{R}'$ such that $R'[ f(w) ] \subseteq \{ f(v)\mid Rwv\}$
\item[(b)] for all $w, v\in W$, if $Ag(\lan {\mathscr R},\preceq\ran)wv$ then $Ag'(\lan {\mathscr R}',\preceq'\ran) f(w) f(v)$
\end{enumerate}
\item Back conditions:
\begin{enumerate}
\item[(a)] for all $R' \in {\mathscr R}'$ and all $w\in W$ there exists some $R\in {\mathscr R}$ such that $\{f(v)\mid Rwv\} \subseteq R'[f(w)]$.
\item[(b)] for all $w\in W$, $v' \in W'$, if $Ag'(\lan {\mathscr R}',\preceq'\ran) f(w) v'$ then there exists some world $v\in W$ such that $Ag(\lan {\mathscr R},\preceq\ran)wv$ and $f(v)=v'$.
\end{enumerate}
\end{enumerate}
\end{definition}

\begin{proposition} \label{bounded invariant} Let $M=\langle W, \lan \mathscr{R}, \preceq\ran, V, Ag \rangle$ and $M'=\langle W', \lan \mathscr{R}', \preceq'\ran, V', Ag' \rangle$ be two \textsf{REL} models. Let $f: W \to W'$ be a surjective bounded aggregation-morphism. Then for all $w\in W$ and $\varphi\in\mathscr{L}$: $M,w\models \varphi$ iff $M',f(w)\models \varphi$. That is: modal satisfaction is invariant under surjective bounded aggregation-morphisms.
\begin{proof}
By induction on the structure of $\varphi$. The base case holds by the valuation condition. The boolean cases are shown by unfolding the definitions, so we consider the cases involving modalities.\\

Suppose $M,w\models \Diamond_0\psi$. Then for all $R\in \mathscr{R}$ there is some $v\in W$ such that $Rwv$ and $M,v\models \psi$. Now we want to show: $M,f(w)\models \Diamond_0\psi$. That is, for all $R'\in \mathscr{R}'$ there is some $v'\in W'$ such that $R'f(w)v'$ and $M',v'\models \psi$. Let $R'\in \mathscr{R}'$ be arbitrary. By the back condition 3(a), there is some $R\in {\mathscr R}$ such that $\{f(v)\mid Rwv\} \subseteq R'[f(w)]$. Hence given $Rwv$, we have $f(v)\in R'[f(w)]$. That is, $R'f(w)f(v)$. By induction hypothesis, given $M,v\models \psi$ we have $M',f(v)\models \psi$. As $R'$ was arbitrarily picked, this holds for all relations in $\mathscr{R}'$. Hence $M',f(w)\models \Diamond_0\psi$.\\
 
 Suppose now that $M',f(w)\models \Diamond_0\psi$. Then for all $R'\in \mathscr{R}'$ there is some $v'\in W'$ such that $R'f(w)v'$ and $M',v'\models \psi$. Now we want to show: $M,w\models \Diamond_0\psi$. That is, for all $R\in \mathscr{R}$ there is some $v\in W$ such that $Rwv$ and $M,v\models \psi$. Let $R\in \mathscr{R}$ be arbitrary. By the forth condition $2(a)$, there exists some $R'\in \mathscr{R}'$ such that $R'[ f(w) ] \subseteq \{ f(v)\mid Rwv\}$. We have $R'f(w)v'$ and $M',v'\models \psi$ for some $v'\in W'$. As $f$ is surjective, we have $v'=f(u)$ for some $u\in W$. Hence given $R'[ f(w) ] \subseteq \{ f(v)\mid Rwv\}$ and  $f(u)\in R'[ f(w) ]$, we get $f(u)\in \{ f(v)\mid Rwv\}$. Hence $Rwu$. By induction hypothesis, given $M',f(u)\models \psi$ we get $M,u\models \psi$. As $R$ was arbitrarily picked, this holds for all relations in $\mathscr{R}$. Hence we have $M,w\models \diamond_0\psi$.\\
 
Now suppose $M,w\models \Diamond\psi$. Then there is some $v\in W$ such that $ Ag(\ran\mathscr{R},\preceq\lan)wv$ and $M,v\models \psi$. By the forth condition $2(b)$, we have $Ag'(\ran\mathscr{R}',\preceq'\lan)f(w)f(v)$. By induction hypothesis, $M',f(v)\models \psi$. Hence $M',f(w)\models \psi$.\\
 
Lastly, suppose $M,f(w)\models \Diamond\psi$. Then there is some $v'\in W'$ such that $ Ag'(\ran\mathscr{R}',\preceq'\lan)wv$ and $M',v'\models \psi$. Hence by the back condition $3(b)$, there is some world $v\in W$ such that $ Ag(\ran\mathscr{R},\preceq\lan)wv$ and $f(v)=v'$. By induction hypothesis, we get $M,v\models \psi$. Hence $M,w\models \Diamond\psi$.

\end{proof}
\end{proposition}

\begin{proposition} \label{is bounded m} The map $\beta: \tilde{W}\to W^c$ is a surjective bounded aggregation-morphism.
\begin{proof} We need to check that $\beta$ satisfies the conditions of a surjective bounded aggregation-morphism.
\begin{enumerate}[leftmargin=*]
\item Surjectivity: Let $T\in W^c$ be arbitrary. We need to show that there is some $h\in\tilde{W}$ such that $\beta(h)=T$. Recall that we showed in \ref{is rel mod pre}.2. that $W^c\times W^c= R^{\Box_0\top}\in \mathscr{R^c}$. Hence $R^{\Box_0\top} T_0 T$. Thus the history $h=(T_0,\Box_0\top,T)$ is an element of $\tilde{W}$ with $\beta(h)=T$, as required.
\item Valuation condition. This follows from the definition of $\tilde{V}$.
\item Forth conditions:
\begin{enumerate}
\item[(a)] We need to show that for all $R\in \tilde{{\mathscr R}}$, for all $w\in\tilde{W}$, there exists some $R''\in \mathscr{R}^c$ such that $R''[ \beta(w) ] \subseteq \{ \beta(v)\mid Rwv\}$. Let $R\in \tilde{{\mathscr R}}$ and $w\in\tilde{W}$ be arbitrary. 
Suppose first that $R= \tilde{R}^{\Box_0\varphi}$ for some $\varphi$ with $\exists\Box_0\varphi\in T_0$. Consider $R^{\Box_0\varphi}\in \mathscr{R}^c$. Take any $T\in R^{\Box_0\varphi}[\beta(w)]$, i.e., $R^{\Box_0\varphi}\beta(w) T$. We will show that $T\in \{ \beta(v)\mid \tilde{R}^{\Box_0\varphi}wv\}$. Note that, given $R^{\Box_0\varphi}\beta(w) T$, the history $w'=(w,\Box_0\varphi, T)$ is in $\tilde{W}$. This means that $w\to^{\Box_0\varphi}w'$. Hence $(\to^{\Box_0\varphi})^* ww'$, i.e., $\tilde{R}^{\Box_0\varphi}w w'$. Given $\beta(w')= T$, we get $T\in \{ \beta(v)\mid \tilde{R}^{\Box_0\varphi}wv\}$, as required.\\

Suppose now that $R= R'_l=(\to^\Box \cup \bigcup\{\to^{(\varphi,l)} \mid \varphi\in \mathscr{L}\})^*$. Consider $R'=(R^\Box \cup \bigcup\{R^{(\varphi,l)}\})^*\in \mathscr{R}^c$. Take any $T\in R'[\beta(w)]$, i.e., $R'\beta(w) T$. We will show that $T\in \{ \beta(v)\mid R'_l wv\}$. Given $R'\beta(w) T$, for some $n\geq 0$, there is a path:
\[\beta(w)=S_0 R^{\epsilon_0} S_1 R^{\epsilon_2},\dots ,R^{\epsilon_{n-1}} S_n=T\] 
where $S_i\in W^c$, $\epsilon_k\in \{\Box\}\cup\{(\varphi,l)\mid \varphi\in \mathscr{L}\}$, for $k<n$. Hence there are histories  $w_1=(w,\epsilon_0,S_1)$, $w_2=(w,\epsilon_0,S_1,\epsilon_1,S_2)$, up to $w_n=(w,\epsilon_0,S_1,\epsilon_1,S_2,\dots, \epsilon_{n-1}, T)$. Hence, by definition of $\to^{\epsilon_k}$, for each $k<n$, we have $w_k\to^{\epsilon_k}w_{k+1}$. Hence there is a path
\[w=w_0 \to^{\epsilon_0} w_1 \to^{\epsilon_2},\dots ,\to^{\epsilon_{n-1}} w_n\] 
And hence $R'_l w w_n$. Given $\beta(w_n)=T$, we get $T\in \{ \beta(v)\mid R'_l wv\}$, as required.\\ 

The case of $R=R'_r$ is analogous to the one above. Hence we have left the case $R=\tilde{W}\times \tilde{W}$. Consider $R^{\Box_0\top}=W^c\times W^c\in \mathscr{R}^c$. Take any $T\in R^{\Box_0\top}[\beta(w)]$, i.e., $R^{\Box_0\top}\beta(w) T$. This just means that $T\in W^c$. We will show that $T\in \{ \beta(v)\mid (w,v)\in \tilde{W}\times \tilde{W}\}=\{ \beta(v)\mid v\in \tilde{W}\}$. As $\beta$ is surjective, we know that there is some $u\in \tilde{W}$ such that $\beta(u)=T$, and we are done.

\item[(b)] We need to show that for all $w, v\in \tilde{W}$, if $\tilde{Ag}(\lan \tilde{{\mathscr R}},\tilde{\preceq}\ran)wv$ then $Ag^c(\lan {\mathscr R}^c,\preceq^c\ran) \beta(w) \beta(v)$. Let $w, v\in \tilde{W}$ be arbitrary and suppose that $\tilde{Ag}(\lan \tilde{{\mathscr R}},\tilde{\preceq}\ran)wv$. By proposition \ref{tilde is lex}.3, given $\tilde{Ag}(\lan \tilde{{\mathscr R}},\tilde{\preceq}\ran)wv$ we have $\tilde{R}^\Box wv$, i.e., $(\to^\Box)^* wv$. Hence, for some $n\geq 0$, there is a path:
\[w=w_0 \to^{\Box} w_1 \to^{\Box}\dots \to^{\Box} w_n=v\]
Hence there are histories $w_1=(w,\Box,S_1)$, $w_2=(w,\Box,S_1,\Box,S_2)$, up to $w_n=v=(w,\Box,S_1,\Box,S_2,\dots, \Box, S_n)$. Hence by definition of $w_n$ we have
$$\beta(w) R^\Box S_1 R^\Box S_2,\dots, R^\Box S_n$$
And since $R^\Box$ is transitive, we get $R^\Box \beta(w) S_n$, i.e.,  $R^\Box \beta(w) \beta(v)$, as required.
\end{enumerate}
\item Back conditions:
\begin{enumerate}
\item[(a)] We need to show that for all $R''\in {\mathscr R}^c$ and all $w\in \tilde{W}$ there exists some $R\in \tilde{{\mathscr R}}$ such that $\{\beta(v)\mid Rwv\} \subseteq R''[\beta(w)]$. Let $R'' \in {\mathscr R}^c$ and $w\in \tilde{W}$ be arbitrary. We reason by cases. First, suppose that $R''= R^{\Box_0\varphi}$. Consider $\tilde{R}^{\Box_0\varphi}\in \tilde{{\mathscr R}}$. We will show that $\{\beta(v)\mid \tilde{R}^{\Box_0\varphi}wv\} \subseteq R^{\Box_0\varphi}[\beta(w)]$. Take any $\beta(u)\in \{\beta(v)\mid  \tilde{R}^{\Box_0\varphi}wv\}$. We have $\tilde{R}^{\Box_0\varphi}wu$, i.e., $(\to^{\Box_0\varphi})^*wu$. Hence for some $n\geq 0$, there is a path:
\[w=w_0 \to^{\Box_0\varphi} w_1 \to^{\Box_0\varphi},\dots ,\to^{\Box_0\varphi} w_n=u\] 
where $w_i\in \tilde{W}$, for $k\leq n$. Hence there are histories  $w_1=(w,\Box_0\varphi,S_1)$, $w_2=(w,\Box_0\varphi,S_1,\Box_0\varphi,S_2)$, up to $w_n=u=(w,\Box_0\varphi,S_1,\Box_0\varphi,S_2,\dots, \Box_0\varphi, S_n)$. Hence, by definition of $w_n$, we have
\[\beta(w) R^{\Box_0} S_1 R^{\Box_0},\dots ,R^{\Box_0} \beta(u)\] 
And since $R^{\Box_0}$ is transitive, we get $R^{\Box_0}\beta(w)\beta(u)$, as required.\\ 

Suppose now that $R''=R'\in W^c$. Consider $R'_l\in \tilde{{\mathscr R}}$. We will show that $\{\beta(v)\mid R'_l wv\} \subseteq R'[\beta(w)]$. Take any $\beta(u)\in \{\beta(v)\mid  R'_l wv\}$. We have $R'_l wu$, i.e., $(\to^{\Box}\cup\bigcup\{\to^{(\varphi,l)}\mid \varphi\in \mathscr{L}\})^*wu$. Hence for some $n\geq 0$, there is a path:
\[w=w_0 \to^{\epsilon_0} w_1 \to^{\epsilon_1},\dots ,\to^{\epsilon_{n-1}} w_n=u\] 
where $w_i\in \tilde{W}$, $\epsilon_k\in \{ \Box \} \cup \bigcup\{(\varphi,l)\mid \varphi\in \mathscr{L}\}$ for $k< n$. By definition of $\to^{\epsilon_k}$, $w_1=(w,\epsilon_0,S_1)$ for some $S_1\in W^c$, $w_2=(w,\epsilon_0,S_1,\epsilon_1\varphi,S_2)$ for some $S_2\in W^c$, up to $w_n=u=(w,\epsilon_0,S_1,\epsilon_1,S_2,\dots, \epsilon_{n-1}, S_n)$ where $S_i\in W^c$ for $i\leq n$. Hence, by definition of $w_n$, we have a path 
\[\beta(w) R^{\epsilon_0} S_1 R^{\epsilon_1} S_2,\dots ,R^{\epsilon_{n-1}} \beta(u)\] 
Since $R'=(R^\Box \cup \bigcup\{R^{(\varphi,l)}\})^*$, the path above is a path from $\beta(w)$ to $\beta(u)$ along $R'$, i.e., we have $R'\beta(w)\beta(u)$, as required.
\item[(b)] We need to show that for all $w\in \tilde{W}$, $T \in W^c$, if $Ag^c(\lan {\mathscr R}^c,\preceq^c\ran) \beta(w) T$ then there exists some history $v\in \tilde{W}$ such that $\tilde{Ag}(\lan \tilde{{\mathscr R}},\tilde{\preceq}\ran)wv$ and $\beta(v)=T$. Let $w\in \tilde{W}$ and $T \in W^c$ be arbitrary, and suppose that $Ag^c(\lan {\mathscr R}^c,\preceq^c\ran) \beta(w) T$, i.e., $R^\Box \beta(w)T$. Then the history $w'=(w,\Box, T)$ is in $\tilde{W}$ and $\beta(w')=T$, as required. 
\end{enumerate}
\end{enumerate}
\end{proof}
\end{proposition}

With the results established above, we can now show the main result. 

\begin{claim}$\mathsf{L}_\mathsf{lex}$ is complete w.r.t. $\mathsf{lex}$ models.

\begin{proof}
Let $\Gamma$ be a $\mathsf{L}_\mathsf{lex}$-consistent set of formulas. By Lindenbaum's Lemma, $\Gamma$ can be extended to a maximal consistent set $T_0$. Choose any canonical pre-model $M^c$
for $T_0$. By Lemma 14, $M^c,T_0\models \varphi$ for all $\varphi\in T_0$. Let $\tilde{K}$ be the unraveling of $M^c$ around $T_0$ and let $\tilde{M}$ be the $\mathsf{lex}$ model generated from $\tilde{K}$. Note that the history $(T_0)\in \tilde{W}$. Let $\beta: \tilde{W}\to W^c$ be the map defined above. By proposition \ref{is bounded m}, $\beta$ is a surjective bounded aggregation-morphism. By  proposition \ref{bounded invariant}, we have $M^c,T_0\models \psi$ iff $\tilde{M}, \beta(T_0)\models \psi$. Hence, in particular, $\tilde{M}, \beta(T_0)\models \varphi$ for all $\varphi\in T_0$.
\end{proof}

\end{claim}
\subsection*{PROOF OF THEOREM 2}

We recall the theorem:

\dynamicscap*

\noindent We prove the soundness and completeness of $\mathsf{L}^!$, $\mathsf{L}^+$ and $\mathsf{L}^\Uparrow$ separately.

\subsubsection*{Soundness and completeness of $\mathsf{L}^!$.} 
\begin{claim} $\mathsf{L}^!$ is sound w.r.t. $\cap$-models.

\begin{proof} A straighforward validity check. 

\end{proof}
\end{claim}

\noindent \textit{Claim}. $\mathsf{L}^!$ is complete w.r.t $\cap$-models.
\begin{proof} The proof is standard, following the approach presented, e.g., in \cite{vanDitmarsch2007}, Chapter 7.
\end{proof}

\subsubsection*{Soundness and completeness of $\mathsf{L}^+$.} As before, we first consider soundness.

\begin{claim}  $\mathsf{L}^+$ is sound w.r.t $\cap$-models.
\begin{proof} Let $M=\langle W, \mathscr{R}, V, Ag_\cap\rangle$ be a $\cap$-model, $w$ a world in $M$, $\pi$ be an evidence program with normal form $\bigcup_{s\in S_0(I)}(?s(\bm{\varphi},\bm{\psi}); A ; ?\psi_{s_{\mathsf{len}(s)}}) \cup (?\top)$. 
\begin{enumerate}
\item Axiom $\text{EA4}_\cap$: We first prove the following:
\begin{claim} $\llb \pi\rrb_M[w]\subseteq \llb [+ \pi]\chi \rrb_{M}$ iff $M,w\models [+  \pi]\chi \wedge \bigwedge_{s\in S_0(I)}( s(\bm{\varphi},\bm{\psi}) \to \forall (\psi_{s_{\mathsf{len}(s)}} \to [+ \pi]\chi))$.
\begin{proof}
($\Rightarrow$) Suppose $\llb \pi\rrb_M[w]\subseteq \llb [+ \pi]\chi \rrb_{M}$. As $\pi$ is a  $*$-program, $\llb \pi\rrb_M$ is reflexive and thus $M,w\models [+ \pi]\chi$. It remains to be shown that 
\[M,w\models \bigwedge_{s\in S_0(I)}( s(\bm{\varphi},\bm{\psi}) \to \forall (\psi_{s_{\mathsf{len}(s)}} \to [+ \pi]\chi))\]
Take any $s\in S_0(I)$ and suppose that $M,w\models s(\bm{\varphi},\bm{\psi})$. We need to show that $M,w\models \forall (\psi_{s_{\mathsf{len}(s)}} \to [+ \pi]\chi)$. Take any $v\in W$ and suppose $M,v\models \psi_{s_{\mathsf{len}(s)}}$. If we show that $M,v\models [+ \pi]\chi$, we are done. Given $M,w\models s(\bm{\varphi},\bm{\psi})$ and $M,v\models \psi_{s_{\mathsf{len}(s)}}$, by Proposition \ref{x sees y}, we have $(w,v)\in \llb ?s(\bm{\varphi},\bm{\psi}); A ; ?\psi_{s_{\mathsf{len}(s)}}\rrb_M$. Thus 
\[(w,v)\in \bigcup_{s\in S_0(I)}\llb?s(\bm{\varphi},\bm{\psi}); A ; ?\psi_{s_{\mathsf{len}(s)}}\rrb_M\]
Hence as 
\small
\begin{flalign*}
\llb\pi\rrb_M & =\llb\bigcup_{s\in S_0(I)}\big(?s(\bm{\varphi},\bm{\psi}); A ; ?\psi_{s_{\mathsf{len}(s)}}\big) \cup (?\top)\rrb_M \\
 & = \llb\bigcup_{s\in S_0(I)}\big(?s(\bm{\varphi},\bm{\psi}); A ; ?\psi_{s_{\mathsf{len}(s)}}\big)\rrb_M \cup \llb ?\top\rrb_M \\
 & = \bigcup_{s\in S_0(I)}\llb?s(\bm{\varphi},\bm{\psi}); A ; ?\psi_{s_{\mathsf{len}(s)}}\rrb_M \cup \llb ?\top\rrb_M 
\end{flalign*}
\normalsize
we have $(w,v)\in \llb\pi\rrb_M$. Hence, given $\llb \pi\rrb_M[w]\subseteq \llb [+ \pi]\chi \rrb_{M}$ we have $M,v\models [+ \pi]\chi$, as required.

$(\Leftarrow$) Suppose that $M,w\models [+ \pi]\chi \wedge \bigwedge_{s\in S_0(I)}( s(\bm{\varphi},\bm{\psi}) \to \forall (\psi_{s_{\mathsf{len}(s)}} \to [+ \pi]\chi))$. We will show that $\llb \pi\rrb_M[w]\subseteq \llb [+ \pi]\chi \rrb_{M}$. Take any $v$ and suppose $(w,v)\in \llb \pi\rrb_M$. We need to show that $v\in \llb [+ \pi]\chi \rrb_{M}$. If $v=w$, given $M,w\models [+ \pi]\chi$ we are done. So suppose $v\neq w$. Note that 
\small
\begin{flalign*}
& (w,v) \in \llb\pi\rrb_M \\
\text{ iff } & (w,v)\in \llb\bigcup_{s\in S_0(I)}?\big(s(\bm{\varphi},\bm{\psi}); A ; ?\psi_{s_{\mathsf{len}(s)}}\big) \cup (?\top)\rrb_M\\
\text{ iff } & (w,v)\in  \llb\bigcup_{s\in S_0(I)}\big(?s(\bm{\varphi},\bm{\psi}); A ; ?\psi_{s_{\mathsf{len}(s)}}\big)\rrb_M\text{ or } (w,v)\in\llb?\top\rrb_M\\
\text{ iff } & (w,v)\in  \llb\bigcup_{s\in S_0(I)}\big(?s(\bm{\varphi},\bm{\psi}); A ; ?\psi_{s_{\mathsf{len}(s)}}\big)\rrb_M \text{ or } w=v\\
\text{ iff } & (w,v)\in  \llb\bigcup_{s\in S_0(I)}\big(?s(\bm{\varphi},\bm{\psi}); A ; ?\psi_{s_{\mathsf{len}(s)}}\big)\rrb_M \ \ \ (\text{as } w\neq v \text{ by assumption })\\
\text{ iff } & (w,v)\in \bigcup_{s\in S_0(I)} \llb?s(\bm{\varphi},\bm{\psi}); A ; ?\psi_{s_{\mathsf{len}(s)}}\rrb_M\\
\text{ iff } &\text{ for some } s'\in S_0(I), (w,v)\in \llb?s'(\bm{\varphi},\bm{\psi}); A ; ?\psi_{s'_{\mathsf{len}(s')}}\rrb_M\\
\text{ iff } &\text{ for some } s'\in S_0(I), w\in \llb s'(\bm{\varphi},\bm{\psi})\rrb_M \text{ and } v\in\llb\psi_{s'_{\mathsf{len}(s')}}\rrb_M \ \ \ \text{(by Proposition } \ref{x sees y})
\end{flalign*}
\normalsize
Since we have $M,w\models \bigwedge_{s\in S_0(I)}( s(\bm{\varphi},\bm{\psi}) \to \forall (\psi_{s_{\mathsf{len}(s)}} \to [+ \pi]\chi))$, we get in particular
\[M,w\models s'(\bm{\varphi},\bm{\psi}) \to \forall (\psi_{s'_{\mathsf{len}(s')}} \to [+ \pi]\chi)\]
Thus from $w\in \llb s'(\bm{\varphi},\bm{\psi})\rrb_M$ we get $M,w\models \forall (\psi_{s'_{\mathsf{len}(s')}} \to [+ \pi]\chi)$. And given $v\in\llb\psi_{s'_{\mathsf{len}(s')}}\rrb_M$ we get $M,v\models [+ \pi]\chi$, as required.
\end{proof}
\end{claim}
Given the Claim, we have:
\small
\begin{flalign*}
& M,w\models [+ \pi]\Box_0\chi &\\
\text{ iff } & M^{+ \pi},w\models \Box_0 \chi \\
\text{ iff } & \text{there is an } R\in \mathscr{R}\cup\{\llb\pi\rrb_M\} \text{ such that } R[w]\subseteq \llb \chi \rrb_{M^{+ \pi}}\\
\text{ iff } & \text{there is an } R\in \mathscr{R}\cup\{\llb\pi\rrb_M\} \text{ such that } R[w]\subseteq \llb [+ \pi]\chi \rrb_{M}\\
\text{ iff } & \text{there is an } R\in \mathscr{R} \text{ such that } R[w]\subseteq \llb  [+ \pi]\chi \rrb_{M}\\ 
& \text{ or } \llb \pi\rrb_M[w]\subseteq \llb [+ \pi]\chi \rrb_{M}\\
\text{ iff } & M,w\models\Box_0  [+ \pi]\chi\\ 
& \text{ or } M,w\models  [+ \pi]\chi \wedge \bigwedge_{s\in S_0(I)}( s(\bm{\varphi},\bm{\psi}) \to \forall (\psi_{s_{\mathsf{len}(s)}} \to  [+ \pi]\chi)) \ \ \ (\text{ by the Claim above)})\\
\text{ iff } & M,w\models\Box_0  [+ \pi]\chi \lor \big( [+ \pi]\chi \wedge \bigwedge_{s\in S_0(I_i)}( s(\bm{\varphi},\bm{\psi}) \to \forall (\psi_{s_{\mathsf{len}(s)}} \to [+ \lan \pi_i\ran_{i< n}]\chi))\big)
\end{flalign*}
\normalsize
\item Axiom $\text{EA5}_\cap$:
We first prove the following:
\begin{claim} $\bigcap(\mathscr{R}\cup\{\llb\pi\rrb_M\})[w]\subseteq \llb [+ \pi]\chi \rrb_{M}$ iff $M,w\models [+  \pi]\chi \wedge \bigwedge_{s\in S_0(I)}( s(\bm{\varphi},\bm{\psi}) \to \Box (\psi_{s_{\mathsf{len}(s)}} \to [+ \pi]\chi))$.
\begin{proof}
($\Rightarrow$) Suppose $\bigcap(\mathscr{R}\cup\{\llb\pi\rrb_M\})[w]\subseteq \llb [+ \pi]\chi \rrb_{M}$. As $\bigcap(\mathscr{R}\cup\{\llb\pi\rrb_M\})$ is reflexive, we have $M,w\models [+ \pi]\chi$. It remains to be shown that 
\[M,w\models \bigwedge_{s\in S_0(I)}( s(\bm{\varphi},\bm{\psi}) \to \Box (\psi_{s_{\mathsf{len}(s)}} \to [+ \pi]\chi))\]
Take any $s\in S_0(I)$ and suppose that $M,w\models s(\bm{\varphi},\bm{\psi})$. We need to show that $M,w\models \Box (\psi_{s_{\mathsf{len}(s)}} \to [+ \pi]\chi)$. Take any $v\in \bigcap(\mathscr{R})[w]$ and suppose $M,v\models \psi_{s_{\mathsf{len}(s)}}$. If we show that $M,v\models [+ \pi]\chi$, we are done. Given $M,w\models s(\bm{\varphi},\bm{\psi})$ and $M,v\models \psi_{s_{\mathsf{len}(s)}}$, by Proposition \ref{x sees y}, we have $(w,v)\in \llb ?s(\bm{\varphi},\bm{\psi}); A ; ?\psi_{s_{\mathsf{len}(s)}}\rrb_M$. Thus 
\[(w,v)\in \bigcup_{s\in S_0(I_k)}\llb?s(\bm{\varphi},\bm{\psi}); A ; ?\psi_{s_{\mathsf{len}(s)}}\rrb_M\]
Hence as 
\small
\begin{flalign*}
\llb\pi\rrb_M & = \llb\bigcup_{s\in S_0(I)}\big(?s(\bm{\varphi},\bm{\psi}); A ; ?\psi_{s_{\mathsf{len}(s)}}\big) \cup (?\top)\rrb_M \\
 & = \llb\bigcup_{s\in S_0(I)}\big(?s(\bm{\varphi},\bm{\psi}); A ; ?\psi_{s_{\mathsf{len}(s)}}\big)\rrb_M \cup \llb ?\top\rrb_M \\
 & = \bigcup_{s\in S_0(I)}\llb?s(\bm{\varphi},\bm{\psi}); A ; ?\psi_{s_{\mathsf{len}(s)}}\rrb_M \cup \llb ?\top\rrb_M 
\end{flalign*}
\normalsize
we have $(w,v)\in \llb\pi\rrb_M$. Hence, given $\bigcap(\mathscr{R}\cup\{\llb\pi\rrb_M\})[w]\subseteq \llb \pi\rrb_M[w]\subseteq \llb [+ \pi]\chi \rrb_{M}$ we have $M,v\models [+ \pi]\chi$, as required.

$(\Leftarrow$) Suppose that $M,w\models [+ \pi]\chi \wedge \bigwedge_{s\in S_0(I)}( s(\bm{\varphi},\bm{\psi}) \to \Box (\psi_{s_{\mathsf{len}(s)}} \to [+ \pi]\chi))$. We will show that $\bigcap(\mathscr{R}\cup\{\llb\pi\rrb_M\})[w]\subseteq \llb [+ \pi]\chi \rrb_{M}$. Take any $v$ and suppose $(w,v)\in \bigcap(\mathscr{R}\cup\{\llb\pi\rrb_M\})$. We need to show that $v\in \llb [+ \pi]\chi \rrb_{M}$. If $v=w$, given $M,w\models [+ \pi]\chi$ we are done. So suppose $v\neq w$. Since $(w,v)\in \bigcap(\mathscr{R}\cup\{\llb\pi\rrb_M\})= \bigcap(\mathscr{R})\cap\llb\pi\rrb_M$, we have $(w,v) \in \llb\pi\rrb_M$. Reasoning as we did in the proof of $\text{EA4}_\cap$, we get 
\begin{flalign*}
& (w,v) \in \llb\pi\rrb_M \text{ iff }\text{ for some } s'\in S_0(I), w\in \llb s'(\bm{\varphi},\bm{\psi})\rrb_M \text{ and } v\in\llb\psi_{s'_{\mathsf{len}(s')}}\rrb_M
\end{flalign*}
Given that $M,w\models \bigwedge_{s\in S_0(I)}( s(\bm{\varphi},\bm{\psi}) \to \Box (\psi_{s_{\mathsf{len}(s)}} \to [+ \pi]\chi))$, we have in particular
\[M,w\models s'(\bm{\varphi},\bm{\psi}) \to \Box (\psi_{s'_{\mathsf{len}(s)}} \to [+ \pi]\chi)\]
Thus from $w\in \llb s'(\bm{\varphi},\bm{\psi})\rrb_M$ we get $M,w\models \Box (\psi_{s'_{\mathsf{len}(s)}} \to [+ \pi]\chi)$. And given $(w,v)\in \bigcap\mathscr{R}$ and $v\in\llb\psi_{s'_{\mathsf{len}(s')}}\rrb_M$, we get $M,v\models [+ \pi]\chi$, as required.
\end{proof}
\end{claim}
Given the Claim, we have 
\small
\begin{flalign*}
& M,w\models [+ \pi]\Box\chi &\\
\text{ iff } & M^{+ \pi},w\models \Box \chi \\
\text{ iff } & \bigcap(\mathscr{R}\cup\{\llb\pi\rrb_M\})[w]\subseteq \llb \chi \rrb_{M^{+ \pi}}\\
\text{ iff } & \bigcap(\mathscr{R}\cup\{\llb\pi\rrb_M\})[w]\subseteq \llb [+ \pi]\chi \rrb_{M}\\
\text{ iff } & M,w\models [+ \pi]\chi \wedge \bigwedge_{s\in S_0(I_i)}( s(\bm{\varphi},\bm{\psi}) \to \Box (\psi_{s_{\mathsf{len}(s)}} \to [+ \pi]\chi)) \ \ \ (\text{ by the Claim above)})
\end{flalign*}
\normalsize
\item Axiom $\text{EA6}_\cap$: 
\[M,w\models  [+ \pi]\forall\chi \text{ iff }  M^{+ \pi},w\models \forall \chi \text{ iff } \llb \chi \rrb_{M^{+ \pi}}=W^{+ \pi} \\
\text{ iff } \llb  [+ \pi] \chi \rrb_{M}=W \text{ iff } M,w \models \forall [+ \pi]\chi\]
\end{enumerate}
\end{proof}
\end{claim}

\begin{claim} $\mathsf{L}^+$ is complete with respect to the class of $\cap$-models.
\begin{proof} The proof follows the approach presented above for $\mathsf{L}^!$.\end{proof}
\end{claim}

\subsubsection*{Soundness and completeness of $\mathsf{L}^\Uparrow$} As before, we first consider soundness. Before proving the main claim, we introduce a lemma about the formulas occurring in the reduction axioms.

\begin{lemma} \label{lemma upg}  
Let $M=\langle W, \mathscr{R}, V, Ag_\cap\rangle$ be a $\cap$-model, $w$ a world in $M$, $\pi\in \Pi_*$ be a program with normal form $\bigcup_{s\in S_0(I)}(?s(\bm{\varphi},\bm{\psi}); A ; ?\psi_{s_{\mathsf{len}(s)}}) \cup (?\top)$. Then 
\begin{enumerate}
\item $M,w\models \pi^\cap(\chi)$ iff there is an  $R\in\mathscr{R}: (R\cap\llb \pi\rrb_M)[w]\subseteq \llb [\Uparrow \pi]\chi \rrb_M$
\item $M,w\models [\Uparrow \pi]\chi
\wedge \pi^<(\chi)$ iff $w\in\llb [\Uparrow \pi]\chi\rrb_M$ and $\llb \pi\rrb^<_M[w]\subseteq \llb [\Uparrow \pi]\chi\rrb_M$
\end{enumerate}

\begin{proof} $ $

Item 1: ($\Rightarrow$) Let $M,w\models \pi^\cap(\chi)$, i.e., 
\[M,w\models[\Uparrow \pi]\chi \wedge \bigvee_{J\subseteq I}\big( J(\bm{\varphi}) \wedge \Box_0\big((\bigvee_{\substack{s\in S_0(I):\\ s_1\in J}}(\exists(s(\bm{\varphi},\bm{\psi})) \wedge \psi_{s_{\mathsf{len}(s)}})) \to [\Uparrow\pi]\chi\big)\big)\]

We need to show that there is an  $R\in\mathscr{R}: (R\cap\llb \pi\rrb_M)[w]\subseteq \llb [\Uparrow \pi]\chi \rrb_M$. Note first that we have 
\[M,w\models\bigvee_{J\subseteq I}\big( J(\bm{\varphi}) \wedge \Box_0\big((\bigvee_{\substack{s\in S_0(I):\\ s_1\in J}}(\exists(s(\bm{\varphi},\bm{\psi})) \wedge \psi_{s_{\mathsf{len}(s)}})) \to [\Uparrow\pi]\chi\big)\big)\]
Then, there is a $J\subseteq I$ such that $ M,w\models J(\bm{\varphi})$ and 
\[M,w\models \Box_0\big((\bigvee_{\substack{s\in S_0(I):\\ s_1\in J}}(\exists(s(\bm{\varphi},\bm{\psi})) \wedge \psi_{s_{\mathsf{len}(s)}})) \to [\Uparrow\pi]\chi\big)\]
Hence, there is some $R\in\mathscr{R}$ such that, for all $v$ with $Rwv$
\begin{equation}\label{eq cap up}
M,v\models (\bigvee_{\substack{s\in S_0(I):\\ s_1\in J}}(\exists(s(\bm{\varphi},\bm{\psi})) \wedge \psi_{s_{\mathsf{len}(s)}})) \to [\Uparrow\pi]\chi
\end{equation}
Now take any $u$ such that $(w,u)\in R\cap \llb \pi\rrb_M$. If we show that $u\in \llb [\Uparrow \pi]\chi \rrb_M$, we are done. Note first that given $M,w\models \pi^\cap(\chi)$, we have $M,w\models[\Uparrow \pi]\chi$, so if $w=u$ we are done. Suppose $w\neq u$. As $(w,u)\in R\cap \llb \pi\rrb_M$, we have $(w,u)\in \llb \pi\rrb_M$. Note that
\small
\begin{flalign*}
& (w,u)\in\llb\pi\rrb_M\\
\text{ iff } & (w,u)\in \llb\bigcup_{s\in S_0(I)}?\big(s(\bm{\varphi},\bm{\psi}); A ; ?\psi_{s_{\mathsf{len}(s)}}\big) \cup (?\top)\rrb_M  \\ 
\text{ iff } & (w,u)\in  \llb\bigcup_{s\in S_0(I)}\big(?s(\bm{\varphi},\bm{\psi}); A ; ?\psi_{s_{\mathsf{len}(s)}}\big)\rrb_M\text{ or } (w,u)\in\llb?\top\rrb_M\\
\text{ iff } & (w,u)\in  \llb\bigcup_{s\in S_0(I)}\big(?s(\bm{\varphi},\bm{\psi}); A ; ?\psi_{s_{\mathsf{len}(s)}}\big)\rrb_M \text{ or } w=u\\
\text{ iff } & (w,u)\in  \llb\bigcup_{s\in S_0(I)}\big(?s(\bm{\varphi},\bm{\psi}); A ; ?\psi_{s_{\mathsf{len}(s)}}\big)\rrb_M \text{ (since by assumption  } w\neq u)\\
\text{ iff } & (w,u)\in  \bigcup_{s\in S_0(I)} \llb?s(\bm{\varphi},\bm{\psi}); A ; ?\psi_{s_{\mathsf{len}(s)}}\rrb_M\\
\text{ iff } & \exists s^\star\in S_0(I)( (w,u)\in  \llb?s^\star(\bm{\varphi},\bm{\psi}); A ; ?\psi_{s^\star_{\mathsf{len}(s^\star)}}\rrb_M)\\
\text{ iff } & \exists s^\star\in S_0(I)( w\in  \llb s^\star(\bm{\varphi},\bm{\psi})\rrb_M \text{ and } u\in\llb\psi_{s^\star_{\mathsf{len}(s^\star)}}\rrb_M) \ \ \text{(by Prop. } \ref{x sees y})
\end{flalign*}
\normalsize
Thus, we have $M,w\models s^\star(\bm{\varphi},\bm{\psi})$ and hence $M,u\models \exists( s^\star(\bm{\varphi},\bm{\psi}))$. Recall that 
\[s^\star(\bm{\varphi},\bm{\psi}) = \varphi_{s^\star_1} \wedge \bigwedge^{\mathsf{len}(s^\star)}_{k=2}(\exists(\psi_{s^\star_{k-1}} \wedge \varphi_{s^\star_k}))) \]
Hence $M,w\models \varphi_{s^\star_1}$ and as $M,w\models J(\bm{\varphi})$, we must have $\varphi_{s^\star_1}=\varphi_j$ for some $j\in J$. This, together with $M,u\models \psi_{s^\star_{\mathsf{len}(s^\star)}}$ gives us
\[M,u\models \bigvee_{\substack{s\in S_0(I):\\ s_1\in J}}(\exists(s(\bm{\varphi},\bm{\psi})) \wedge \psi_{s_{\mathsf{len}(s)}})\]
which, given the statement \ref{eq cap up} above, implies $M,u\models [\Uparrow\pi]\chi$, as required.

$(\Leftarrow)$ Suppose that there is an  $R\in\mathscr{R}$ such that  $(R\cap\llb \pi\rrb_M)[w]\subseteq \llb [\Uparrow \pi]\chi \rrb_M$. First, note that $R\cap\llb \pi\rrb_M$ is reflexive, and hence $M,w\models [\Uparrow\pi]\chi$. Hence it remains to be shown that 
\[M,w\models\bigvee_{J\subseteq I}\big( J(\bm{\varphi}) \wedge \Box_0\big((\bigvee_{\substack{s\in S_0(I):\\ s_1\in J}}(\exists(s(\bm{\varphi},\bm{\psi})) \wedge \psi_{s_{\mathsf{len}(s)}})) \to [\Uparrow\pi]\chi\big)\big)\]
It is clear that there is some $J\subseteq I$ such that $M,w\models J(\bm{\varphi})$, so we must show for this $J$ that 
\[M,w\models \Box_0\big((\bigvee_{\substack{s\in S_0(I):\\ s_1\in J}}(\exists(s(\bm{\varphi},\bm{\psi})) \wedge \psi_{s_{\mathsf{len}(s)}})) \to [\Uparrow\pi]\chi\big)\]
Consider $R$ and take any $v$ such that $Rwv$. We need to show that 
\[M,v\models (\bigvee_{\substack{s\in S_0(I):\\ s_1\in J}}(\exists(s(\bm{\varphi},\bm{\psi})) \wedge \psi_{s_{\mathsf{len}(s)}})) \to [\Uparrow\pi]\chi\]
Suppose that 
\[M,v\models \bigvee_{\substack{s\in S_0(I):\\ s_1\in J}}(\exists(s(\bm{\varphi},\bm{\psi})) \wedge \psi_{s_{\mathsf{len}(s)}})\]
Then there is some $s\in S_0(I)$ with $s_1\in J$ such that $M,v\models \exists(s(\bm{\varphi},\bm{\psi})) \wedge \psi_{s_{\mathsf{len}(s)}}$. Hence there is some $u$ such that $M,u\models s(\bm{\varphi},\bm{\psi})$. Recall that 
\[s(\bm{\varphi},\bm{\psi}) = \varphi_{s_1} \wedge \bigwedge^{\mathsf{len}(s)}_{k=2}(\exists(\psi_{s_{k-1}} \wedge \varphi_{s_k}))) \]
Given $s_1\in J$ and $M,w\models J(\bm{\varphi})$, we have $M,w\models \varphi_{s_1}$ and thus $M,w\models s(\bm{\varphi},\bm{\psi})$. This, together with $M,v\psi_{s_{\mathsf{len}(s)}}$, implies $(w,v)\in \llb?s(\bm{\varphi},\bm{\psi}); A ; ?\psi_{s_{\mathsf{len}(s)}}\rrb_M$ and hence $(w,v)\in \llb \textsf{nf}(\pi)\rrb_M$ (where $\textsf{nf}(\pi)$ is the normal form for $\pi$), which means $(w,v)\in \llb \pi\rrb_M$. As $(w,v)\in R$ and $(w,v)\in \llb \pi\rrb_M$ we have $(w,v)\in R\cap\llb \pi\rrb_M$. Thus given $(R\cap\llb \pi\rrb_M)[w]\subseteq \llb [\Uparrow \pi]\chi \rrb_M$ we have $M,v\models [\Uparrow \pi]\chi$, as required.\\

Item 2: ($\Rightarrow$) Let $M,w\models [\Uparrow \pi]\chi
\wedge \pi^<(\chi)$. Then $w\in\llb [\Uparrow \pi]\chi\rrb_M$, so it remains to be shown that $\llb \pi\rrb^<_M[w]\subseteq \llb [\Uparrow \pi]\chi\rrb_M$. We have $M,w\models \pi^<(\chi)$, i.e.,
\[M,w\models \bigvee_{J\subseteq I}\big( J(\bm{\psi}) \wedge \bigwedge_{s\in S_0(I)}( s(\bm{\varphi},\bm{\psi}) \to \forall\big( \psi_{s_{\mathsf{len}(s)}} \wedge \bigwedge_{s'\in S_0(I)}\big( s'(\bm{\varphi},\bm{\psi}) \to \forall( \psi_{s'_{\mathsf{len}(s')}} \to \bigwedge_{j\in J}\neg\varphi_j)\big) \to [\Uparrow\pi]\chi\big))\big)\]
Then, there is a $J\subseteq I$ such that $ M,w\models J(\bm{\psi})$ and 
\begin{equation} \label{upg eq}
M,w\models \bigwedge_{s\in S_0(I)}( s(\bm{\varphi},\bm{\psi}) \to \forall\big( (\psi_{s_{\mathsf{len}(s)}} \wedge \bigwedge_{s'\in S_0(I)}( s'(\bm{\varphi},\bm{\psi}) \to \forall( \psi_{s'_{\mathsf{len}(s')}} \to \bigwedge_{j\in J}\neg\varphi_j))) \to [\Uparrow\pi]\chi)\big)
\end{equation}
We need to show that $\llb \pi\rrb^<_M[w]\subseteq \llb [\Uparrow \pi]\chi \rrb_M$. Take any $v$ such that $(w,v)\in\llb\pi\rrb_M^<$, i.e., $(w,v)\in\llb\pi\rrb_M$  and $(v,w)\not\in\llb\pi\rrb_M$. We will show that $v\in \llb [\Uparrow \pi]\chi \rrb_M$. First, observe that 
\small
\begin{flalign*}
& (w,v)\in\llb\pi\rrb_M\\
\text{ iff } & (w,v)\in \llb\bigcup_{s\in S_0(I)}?\big(s(\bm{\varphi},\bm{\psi}); A ; ?\psi_{s_{\mathsf{len}(s)}}\big) \cup (?\top)\rrb_M  \\ 
\text{ iff } & (w,v)\in  \llb\bigcup_{s\in S_0(I)}\big(?s(\bm{\varphi},\bm{\psi}); A ; ?\psi_{s_{\mathsf{len}(s)}}\big)\rrb_M\text{ or } (w,v)\in\llb?\top\rrb_M\\
\text{ iff } & (w,v)\in  \llb\bigcup_{s\in S_0(I)}\big(?s(\bm{\varphi},\bm{\psi}); A ; ?\psi_{s_{\mathsf{len}(s)}}\big)\rrb_M \text{ or } w=v\\
\text{ iff } & (w,v)\in  \llb\bigcup_{s\in S_0(I)}\big(?s(\bm{\varphi},\bm{\psi}); A ; ?\psi_{s_{\mathsf{len}(s)}}\big)\rrb_M \text{ or } w=v\\
\text{ iff } & (w,v)\in  \bigcup_{s\in S_0(I)} \llb?s(\bm{\varphi},\bm{\psi}); A ; ?\psi_{s_{\mathsf{len}(s)}}\rrb_M \text{ or } w=v\\
\text{ iff } & \exists s^\star\in S_0(I)( (w,v)\in  \llb?s^\star(\bm{\varphi},\bm{\psi}); A ; ?\psi_{s^\star_{\mathsf{len}(s')}}\rrb_M) \text{ or } w=v\\
\text{ iff } & \exists s^\star\in S_0(I)( w\in  \llb s^\star(\bm{\varphi},\bm{\psi})\rrb_M \text{ and } v\in\llb\psi_{s^\star_{\mathsf{len}(s^\star)}}\rrb_M) \text{ or } w=v  \ \ \text{(by Prop. } \ref{x sees y})
\end{flalign*}
\normalsize
Moreover, note that 
\small
\begin{flalign*}
& (v,w)\not\in\llb\pi\rrb_M \\
\text{ iff } & (v,w)\not\in \llb\bigcup_{s\in S_0(I)}?\big(s(\bm{\varphi},\bm{\psi}); A ; ?\psi_{s_{\mathsf{len}(s)}}\big) \cup (?\top)\rrb_M \\
\text{ iff } & (v,w)\not\in  \llb\bigcup_{s\in S_0(I)}\big(?s(\bm{\varphi},\bm{\psi}); A ; ?\psi_{s_{\mathsf{len}(s)}}\big)\rrb_M\text{ and } (w,v)\not\in\llb?\top\rrb_M\\
\text{ iff } & (v,w)\not\in \llb\bigcup_{s\in S_0(I)}\big(?s(\bm{\varphi},\bm{\psi}); A ; ?\psi_{s_{\mathsf{len}(s)}}\big)\rrb_M \text{ and } w\neq v \\
\text{ iff } & (v,w)\not\in \bigcup_{s\in S_0(I)} \llb ?s(\bm{\varphi},\bm{\psi}); A ; ?\psi_{s_{\mathsf{len}(s)}}\rrb_M \text{ and } w\neq v\\
\text{ iff } & \forall s\in S_0(I) ( (v,w)\not\in \llb?s(\bm{\varphi},\bm{\psi}); A ; ?\psi_{s_{\mathsf{len}(s)}}\rrb_M)  \text{ and } w\neq v\\
\text{ iff } & \forall s\in S_0(I) ( v\not\in \llb s(\bm{\varphi},\bm{\psi})\rrb_M \text{ or } w\not\in \llb\psi_{s_{\mathsf{len}(s)}}\rrb_M)  \text{ and } w\neq v\\
\text{ iff } & \forall s\in S_0(I) ( v\in \llb s(\bm{\varphi},\bm{\psi})\rrb_M \text{ implies } w\not\in \llb\psi_{s_{\mathsf{len}(s)}}\rrb_M)  \text{ and } w\neq v
\end{flalign*}
\normalsize
Hence we have $\exists s^\star\in S_0(I)( w\in  \llb s^\star(\bm{\varphi},\bm{\psi})\rrb_M \text{ and } v\in\llb\psi_{s^\star_{\mathsf{len}(s^\star)}}\rrb_M)$ and $\forall s\in S_0(I) ( v\in \llb s(\bm{\varphi},\bm{\psi})\rrb_M \text{ implies } w\not\in \llb\psi_{s_{\mathsf{len}(s)}}\rrb_M)$. From \ref{upg eq}, we have in particular 
\[M,w\models s^\star(\bm{\varphi},\bm{\psi}) \to \forall\big( (\psi_{s^\star_{\mathsf{len}(s^\star)}} \wedge \bigwedge_{s'\in S_0(I)}( s'(\bm{\varphi},\bm{\psi}) \to \forall( \psi_{s'_{\mathsf{len}(s')}} \to \bigwedge_{j\in J}\neg\varphi_j))) \to [\Uparrow\pi]\chi)\]
We have $w\in  \llb s^\star(\bm{\varphi},\bm{\psi})\rrb_M$, from which we get
\[M,w\models \forall\big( (\psi_{s^\star_{\mathsf{len}(s^\star)}} \wedge \bigwedge_{s'\in S_0(I)}( s'(\bm{\varphi},\bm{\psi}) \to \forall( \psi_{s'_{\mathsf{len}(s')}} \to \bigwedge_{j\in J}\neg\varphi_j))) \to [\Uparrow\pi]\chi)\]
Thus, in particular
\[M,v\models (\psi_{s^\star_{\mathsf{len}(s^\star)}} \wedge \bigwedge_{s'\in S_0(I)}( s'(\bm{\varphi},\bm{\psi}) \to \forall( \psi_{s'_{\mathsf{len}(s')}} \to \bigwedge_{j\in J}\neg\varphi_j))) \to [\Uparrow\pi]\chi\]
We already have $v\in\llb\psi_{s^\star_{\mathsf{len}(s^\star)}}\rrb_M$, so if we show that
\begin{equation} \label{conj s'}
M,v\models \bigwedge_{s'\in S_0(I)}( s'(\bm{\varphi},\bm{\psi}) \to \forall( \psi_{s'_{\mathsf{len}(s')}} \to \bigwedge_{j\in J}\neg\varphi_j))
\end{equation}
we will get $M,v\models [\Uparrow\pi]\chi$, as required. Take any $s'\in S_0(I)$ and suppose that $M,v\models s'(\bm{\varphi},\bm{\psi})$. We need to show that $M,v\models \forall( \psi_{s'_{\mathsf{len}(s')}} \to \bigwedge_{j\in J}\neg\varphi_j))$. Consider any $u\in W$ and suppose $M,u\models \psi_{s'_{\mathsf{len}(s')}}$. Towards a contradiction, suppose that $M,u\models \varphi_j$, for some $j\in J$. Consider the sequence $s''\coloneqq s'| \lan j\ran$. Note that

\[s''(\bm{\varphi},\bm{\psi}) = \varphi_{s''_1} \wedge \bigwedge^{\mathsf{len}(s'')}_{k=2}\exists(\psi_{s''_{k-1}} \wedge \varphi_{s''_k}) = s'(\bm{\varphi},\bm{\psi}) \wedge \exists(\psi_{s'_{{\mathsf{len}(s')}}} \wedge \varphi_{j})\]

Given $M,u\models \psi_{s'_{\mathsf{len}(s')}}$ and $M,u\models \varphi_j$, we have $M,v\models\exists(\psi_{s'_{{\mathsf{len}(s')}}} \wedge \varphi_{j})$. This, together with $M,v\models s'(\bm{\varphi},\bm{\psi})$ gives us $M,v\models s''(\bm{\varphi},\bm{\psi})$. From $\forall s\in S_0(I) ( v\in \llb s(\bm{\varphi},\bm{\psi})\rrb_M \text{ implies } w\not\in \llb\psi_{s_{\mathsf{len}(s)}}\rrb_M)$ we have: $v\in \llb s''(\bm{\varphi},\bm{\psi})\rrb_M \text{ implies } w\not\in \llb\psi_{{s''}_{\mathsf{len}(s'')}}\rrb_M$. Hence we have $w\not\in \llb\psi_{{s''}_{\mathsf{len}(s'')}}\rrb_M$, i.e., $M,w\not\models \psi_j$. But since $M,w\models J(\bm{\psi})$, we also have $M,w\models \psi_j$ (contradiction). Thus we have $M,v\models \forall( \psi_{s'_{\mathsf{len}(s')}} \to \bigwedge_{j\in J}\neg\varphi_j))$. As $s'$ was arbitrarily picked we get statement \ref{conj s'} above, which together with  $M,v\models \psi_{s^\star_{\mathsf{len}(s^\star)}}$ implies $M,v\models [\Uparrow\pi]\chi$. Since $v$ was picked arbitrarily, we get $\llb \pi\rrb^<_M[w]\subseteq \llb [\Uparrow \pi]\chi \rrb_M$, as required. 

$(\Leftarrow)$ Suppose that $ w \in \llb [\Uparrow \pi]\chi \rrb_M$ and $ \llb \pi\rrb^<_M[w]\subseteq \llb [\Uparrow \pi]\chi \rrb_M$. We need to show $M,w\models [\Uparrow \pi]\chi \wedge \pi^<(\chi)$. $M,w\models [\Uparrow \pi]\chi $ is immediate, so it remains to be shown that $M,w\models \pi^<(\chi)$, i.e.,
\[M,w\models \bigvee_{J\subseteq I}\big( J(\bm{\psi}) \wedge \bigwedge_{s\in S_0(I)}( s(\bm{\varphi},\bm{\psi}) \to \forall\big( (\psi_{s_{\mathsf{len}(s)}} \wedge \bigwedge_{s'\in S_0(I)}( s'(\bm{\varphi},\bm{\psi}) \to \forall( \psi_{s'_{\mathsf{len}(s')}} \to \bigwedge_{j\in J}\neg\varphi_j))) \to [\Uparrow\pi]\chi)\big)\big)\] 
Clearly, there is a $J\subseteq I$ such that $M,w\models J(\bm{\psi})$, so it remains to be shown that;
\[M,w\models \bigwedge_{s\in S_0(I)}( s(\bm{\varphi},\bm{\psi}) \to \forall\big( (\psi_{s_{\mathsf{len}(s)}} \wedge \bigwedge_{s'\in S_0(I)}( s'(\bm{\varphi},\bm{\psi}) \to \forall( \psi_{s'_{\mathsf{len}(s')}} \to \bigwedge_{j\in J}\neg\varphi_j))) \to [\Uparrow\pi]\chi)\big)\] 
Take any $s\in S_0(I)$ and suppose that $M,w\models s(\bm{\varphi},\bm{\psi})$. We need to show that 
\[M,w\models \forall\big( (\psi_{s_{\mathsf{len}(s)}} \wedge \bigwedge_{s'\in S_0(I)}( s'(\bm{\varphi},\bm{\psi}) \to \forall( \psi_{s'_{\mathsf{len}(s')}} \to \bigwedge_{j\in J}\neg\varphi_j))) \to [\Uparrow\pi]\chi\big)\] 
Take any $v\in W$ and suppose that 
\[M,v\models \psi_{s_{\mathsf{len}(s)}}\wedge \bigwedge_{s'\in S_0(I)}( s'(\bm{\varphi},\bm{\psi}) \to \forall( \psi_{s'_{\mathsf{len}(s')}} \to \bigwedge_{j\in J}\neg\varphi_j))\] 
We need to show that $M,v\models [\Uparrow\pi]\chi$. Note that, if $v=w$ we are done, so suppose that $v\neq w$.
Since $\llb \pi\rrb^<_M[w]\subseteq \llb [\Uparrow \pi]\chi \rrb_M$, if we show that $(w,v)\in\llb \pi\rrb^<_M$ we are done. We show this next, i.e., $(w,v)\in \llb \pi\rrb_M$ and $(v,w)\not \in \llb \pi\rrb_M$. Note that $M,w\models s(\bm{\varphi},\bm{\psi})$ and $M,v\models \psi_{s_{\mathsf{len}(s)}}$. Hence $(w,v)\in \llb ?s(\bm{\varphi},\bm{\psi});A;?\psi_{s_{\mathsf{len}(s)}}\rrb_M$ and thus $(w,v)\in \llb \bigcup_{s'\in S_0(I)}(?s'(\bm{\varphi},\bm{\psi}); A ; ?\psi_{s'_{\mathsf{len}(s')}}) \cup (?\top)\rrb_M$, which gives $(w,v)\in \llb \pi\rrb_M$. Towards a contradiction, suppose that $(v,w) \in \llb \pi\rrb_M$. Then 
\[(v,w)\in \llb \bigcup_{s'\in S_0(I)}(?s'(\bm{\varphi},\bm{\psi}); A ; ?\psi_{s'_{\mathsf{len}(s')}})\cup(?\top)\rrb_M \] 
Since $v\neq w$ by assumption, $(v,w)\not\in \llb?\top\rrb_M$, so we must have
\[(v,w)\in \llb \bigcup_{s'\in S_0(I)}(?s'(\bm{\varphi},\bm{\psi}); A ; ?\psi_{s'_{\mathsf{len}(s')}}) \rrb_M=\bigcup_{s'\in S_0(I)}\llb ?s'(\bm{\varphi},\bm{\psi}); A ; ?\psi_{s'_{\mathsf{len}(s')}} \rrb_M \] 

By Lemma \ref{unf}, this means that  there is a finite $vw$-path along $\llb\bigcup_{i\in I}(?\varphi_i;A;?\psi_i)\rrb_M$. That is, for some $s''\in S_0(I)$, $v\in \llb\varphi_{s''_1}\rrb_M$ and there are $z_1,z_2,\dots,z_{\mathsf{len}(s'')},z_{\mathsf{len}(s'')+1}$ such that $z_1=v$ and $z_{\mathsf{len}(s'')+1}=w$ and for each $k\in \{1,\dots,\mathsf{len}(s'')\}$, $(z_k,z_{k+1})\in \llb?\varphi_{s''_k} ; A ; ?\psi_{s''_k}\rrb_M$. Given $M,w\models J(\bm{\psi})$, we must have $\psi_{s''_{\mathsf{len}(s'')}}=\psi_j$ for some $j\in J$. Let $s^\star=\lan s''_1,\dots,s''_{\mathsf{len}(s'')-1}\ran$. 
As we have 
\[M,v\models \bigwedge_{s'\in S_0(I)}( s'(\bm{\varphi},\bm{\psi}) \to \forall( \psi_{s'_{\mathsf{len}(s')}} \to \bigwedge_{j\in J}\neg\varphi_j))\]
In particular, we also have
\[M,v\models s^\star(\bm{\varphi},\bm{\psi}) \to \forall( \psi_{s^\star_{\mathsf{len}(s^\star)}} \to \bigwedge_{j\in J}\neg\varphi_j)\]
As $M,v\models s''(\bm{\varphi},\bm{\psi})$, we also have $M,v\models s^\star(\bm{\varphi},\bm{\psi})$, which gives us
\begin{equation}\label{conj up}
M,v\models \forall( \psi_{s^\star_{\mathsf{len}(s^\star)}} \to \bigwedge_{j\in J}\neg\varphi_j)
\end{equation}
Given that $\psi_{s^\star_{\mathsf{len}(s^\star)}} = \psi_{s''_{\mathsf{len}(s'')-1}}$, we get 
\[M,z_{\mathsf{len}(s'')}\models\psi_{s^\star_{\mathsf{len}(s^\star)}}\]
and thus from statement \ref{conj up} above we get
\[M,z_{\mathsf{len}(s'')}\models\bigwedge_{j\in J} \neg \varphi_j\]
But then $(z_{\mathsf{len}(s'')},w)\not\in\llb?\varphi_{s''_{\mathsf{len}(s'')}} ; A ; ?\psi_{s''_{\mathsf{len}(s'')}}\rrb_M$ (contradiction).

\end{proof}

\end{lemma}

Having shown the lemmma above, we now consider the main soundness claim.

\begin{claim}$\mathsf{L}^\Uparrow$ is sound w.r.t $\cap$-models.
\begin{proof} Let $M=\langle W, \mathscr{R}, V, Ag_\cap\rangle$ be a $\cap$-model, $w$ a world in $M$, $\pi\in \Pi_*$ be a program with normal form $\bigcup_{s\in S_0(I)}(?s(\bm{\varphi},\bm{\psi}); A ; ?\psi_{s_{\mathsf{len}(s)}}) \cup (?\top)$. The validity of $\text{EU1}_\cap$ follows from the fact that the evidence addition transformer does not change the valuation function. The validity of the Boolean reduction axioms $\text{EU2}_\cap$ and $\text{EU3}_\cap$ can be proven by unfolding the definitions.
\begin{enumerate}
\item Axiom $\text{EU4}_\cap$:
\small
\begin{flalign*}
& M,w\models [\Uparrow\pi] \Box_0 \chi \\
\text{ iff } & M^{\Uparrow\pi},w\models \Box_0 \chi \\
\text{ iff } & \text{there is an } R\in \mathscr{R}^{+\pi} \text{ such that } R[w]\subseteq \llb \chi \rrb_{M^{\Uparrow\pi}}\\
\text{ iff } & \text{there is an } R\in \mathscr{R} \text{ such that } \big(\llb \pi\rrb_M^< \cup (\llb \pi\rrb_M \cap R)\big)[w]\subseteq \llb [\Uparrow\pi]\chi \rrb_{M}\\
\text{ iff } & \text{there is an }  R\in \mathscr{R} \text{ such that }  \llb \pi\rrb_M^<[w] \cup (\llb \pi\rrb_M \cap R)[w]\subseteq \llb [\Uparrow\pi]\chi \rrb_{M}\\
 \text{ iff } & \llb \pi\rrb_M^<[w]\subseteq \llb [\Uparrow\pi]\chi \rrb_{M} \text{ and there is an } R\in \mathscr{R} \text{ such that }  (\llb \pi\rrb_M \cap R)[w]\subseteq \llb [\Uparrow\pi]\chi \rrb_{M}\\
 \text{ iff } & \llb \pi\rrb_M^<[w]\subseteq \llb [\Uparrow\pi]\chi \rrb_{M} \text{ and there is an } R\in \mathscr{R} \text{ such that }  (\llb \pi\rrb_M \cap R)[w]\subseteq \llb [\Uparrow\pi]\chi \rrb_{M}\\
 & \text{ and } w\in\llb [\Uparrow\pi]\chi \rrb_{M}  \ \ \ (\text{as } w\in (\llb \pi\rrb_M \cap R)[w]) \\
  \text{ iff } & M,w\models [\Uparrow\pi]\chi \wedge  \pi^<(\chi) \wedge \pi^\cap(\chi) \ \ \   (\text{by Lemma } \ref{lemma upg}) 
\end{flalign*}
\normalsize
\item Axiom $\text{EU5}_\cap$: We first prove the following:
\begin{claim} $(\llb \pi \rrb_M \cap \bigcap\mathscr{R})[w]\subseteq \llb [\Uparrow \pi]\chi \rrb_{M}$ iff $M,w\models [\Uparrow  \pi]\chi \wedge \bigwedge_{s\in S_0(I)}( s(\bm{\varphi},\bm{\psi}) \to \Box (\psi_{s_{\mathsf{len}(s)}} \to [\Uparrow \pi]\chi))$.
\begin{proof} The proof is identical to the one used for the proof of the validity of `Axiom $\text{EA5}_\cap$'. 
\end{proof} 
\end{claim}
Note next that
\small
\begin{align*}
\bigcap \mathscr{R}^{\Uparrow\pi} & = \bigcap_{R\in \mathscr{R}} \Big( \llb \pi \rrb_M^< \cup (\llb \pi \rrb_M \cap R)\Big) = \llb \pi \rrb_M^< \cup \bigcap_{R\in \mathscr{R}} (\llb \pi \rrb_M \cap R)= \llb \pi \rrb_M^< \cup (\llb \pi \rrb_M \cap \bigcap\mathscr{R})
\end{align*}
\normalsize
Thus,
\small
\begin{align*}
&M,w\models [\Uparrow\pi] \Box \chi &\\
&\text{ iff } M^{\Uparrow\pi},w\models \Box \chi \\
& \text{ iff } \bigcap\mathscr{R}^{\Uparrow\pi}[w]\subseteq \llb \chi \rrb_{M^{\Uparrow\pi}}\\
& \text{ iff } \bigcap\mathscr{R}^{\Uparrow\pi}[w]\subseteq \llb [\Uparrow\pi]\chi \rrb_{M}\\
& \text{ iff } \big(\llb \pi \rrb_M^< \cup (\llb \pi \rrb_M \cap \bigcap\mathscr{R})\big)[w] \subseteq \llb [\Uparrow\pi]\chi \rrb_{M}\\
& \text{ iff } \llb \pi \rrb_M^<[w] \cup (\llb \pi \rrb_M \cap \bigcap\mathscr{R})[w] \subseteq \llb [\Uparrow\pi]\chi \rrb_{M}\\
& \text{ iff } \llb \pi \rrb_M^<[w]\subseteq \llb [\Uparrow\pi]\chi \rrb_{M} \text{ and } (\llb \pi \rrb_M \cap \bigcap\mathscr{R})[w] \subseteq \llb [\Uparrow\pi]\chi \rrb_{M}\\
& \text{ iff } w\in \llb [\Uparrow\pi]\chi \rrb_{M} \text{ and } \llb \pi\rrb_M^<[w]\subseteq \llb [\Uparrow\pi]\chi \rrb_{M} \text{ and } (\llb \pi \rrb_M \cap \bigcap\mathscr{R})[w] \subseteq \llb [\Uparrow\pi]\chi \rrb_{M}\\
& (\text{ as } w\in \big(\llb \pi\rrb_M \cap R\big)[w])\\
& M,w\models [\Uparrow\pi]\chi \wedge  \pi^<(\chi) \wedge  \bigwedge_{s\in S_0(I)}( s(\bm{\varphi},\bm{\psi}) \to \Box (\psi_{s_{\mathsf{len}(s)}} \to [\Uparrow \pi]\chi)) \\ 
& (\text{by Lemma } \ref{lemma upg} \text{ and the Claim above} )
\end{align*}
\normalsize
\item Axiom $\text{EU6}_\cap$: 
\[M,w\models  [\Uparrow \pi]\forall\chi \text{ iff }  M^{\Uparrow \pi},w\models \forall \chi \text{ iff } \llb \chi \rrb_{M^{\Uparrow \pi}}=W^{\Uparrow \pi} \\
\text{ iff } \llb  [\Uparrow \pi] \chi \rrb_{M}=W \text{ iff } M,w \models \forall [\Uparrow \pi]\chi\]
\end{enumerate}
\end{proof}
\end{claim}  

\begin{claim} $\mathsf{L}^\Uparrow$ is complete w.r.t. $\cap$-models.
\begin{proof} The follows the same steps used to prove completeness of $\mathsf{L}^!$. 
\end{proof}
\end{claim}

\subsection*{PROOF OF THEOREM 3}

We recall the theorem:

\dynamiclex*

We first introduce some lemmas. During our discussion of evidence upgrade in $\cap$-models,  we showed in Lemma \ref{lemma upg} that $[\Uparrow\pi]\chi \wedge\pi^<(\chi)$ is true at a state $w$ in a $\cap$-model $M$ if $[\Uparrow\pi]\chi$ is true at $w$ and $\llb \pi\rrb^<_M[w]\subseteq \llb [\oplus\pi]\chi\rrb_M$. It is easy to see that, since the formula $\pi^<(\chi)$ contains no occurrences of $\Box$ (only $\forall$), the result transfers to $\mathsf{lex}$ models, as the semantics of $\forall$ is the same in all \textsf{REL} models. Hence we have:

\begin{lemma} \label{lemma 1 for lex}
Let $M=\langle W, \lan \mathscr{R},\preceq\ran, V, \mathsf{lex}\rangle$, $w$ a world in $M$, $\pi\in \Pi_*$ be a program with normal form $\bigcup_{s\in S_0(I)}(?s(\bm{\varphi},\bm{\psi}); A ; ?\psi_{s_{\mathsf{len}(s)}}) \cup (?\top)$. Then 
\[M,w\models [\oplus\pi]\chi
\wedge \pi^<(\chi) \emph{ iff } w\in\llb [\oplus\pi]\chi\rrb_M \text{ and } \llb \pi\rrb^<_M[w]\subseteq \llb [\oplus\pi]\chi\rrb_M\]
\begin{proof} Follows from Lemma \ref{lemma upg}. \end{proof}
\end{lemma}

We now introduce another lemma which will be useful in the proof of the reduction axioms.

\begin{lemma} \label{lemma 2 for lex} Let $M=\langle W, \lan \mathscr{R},\preceq\ran, V, \mathsf{lex}\rangle$, $w$ a world in $M$, $\pi\in \Pi_*$ be a program with normal form $\bigcup_{s\in S_0(I)}(?s(\bm{\varphi},\bm{\psi}); A ; ?\psi_{s_{\mathsf{len}(s)}}) \cup (?\top)$. Then $$(\mathsf{lex}(\lan\mathscr{R},\preceq\ran)\cap \llb\pi\rrb_M )[w]\subseteq \llb [\oplus\pi]\chi \rrb_{M}$$
\begin{center}
\emph{iff} 
\end{center}
$$M,w\models [+  \pi]\chi \wedge \bigwedge_{s\in S_0(I)}( s(\bm{\varphi},\bm{\psi}) \to \Box (\psi_{s_{\mathsf{len}(s)}} \to [\oplus\pi]\chi))$$
\begin{proof}
Straightforward.
\end{proof}

\end{lemma}

We now show one more lemma, before presenting the proof system $\mathsf{L}^+_{\mathsf{lex}}$. With these lemmas in place, the proofs of the validity of the reduction axioms in this system will be almost immediate.

\begin{lemma} \label{lemma 3 for lex} Let $M=\langle W, \lan \mathscr{R},\preceq\ran, V, \mathsf{lex}\rangle$, $w$ a world in $M$, $\pi\in \Pi_*$ be a program and $\varphi$ a formula. Then
\[(\mathsf{lex}(\lan\mathscr{R},\preceq\ran)\cap \llb\pi\rrb_M )[w]\subseteq \llb\varphi\rrb_M \text{ and } \llb\pi\rrb^<_M[w]\subseteq \llb\varphi\rrb_M\]
\begin{center}
\emph{iff}
\end{center}
\[\mathsf{lex}(\lan \mathscr{R}^{\oplus\pi},\preceq^{\oplus\pi}\ran)[w]\subseteq \llb\varphi\rrb_M\]

\begin{proof} Straighforward.
\end{proof}

\end{lemma}

We now turn to the main soundness claim.

\begin{claim} $\mathsf{L}^\oplus$ is sound w.r.t. $\mathsf{lex}$-models.

\begin{proof} Let $M=\langle W, \lan \mathscr{R},\preceq\ran, V, \mathsf{lex}\rangle$, $w$ a world in $M$, $\pi$ be an evidence program with normal form $\bigcup_{s\in S_0(I)}(?s(\bm{\varphi},\bm{\psi}); A ; ?\psi_{s_{\mathsf{len}(s)}}) \cup (?\top)$. 
\begin{enumerate}[leftmargin=*]
\item Axiom $\oplus$EA4: Note that this same axiom was called $\text{EA4}_\cap$ in the system $\mathsf{L}^+_\cap$ for $\cap$-models. The effects on evidence possession (as expressed by $\Box_0$-formulas) of evidence addition in $\cap$-models are the same as the effects of prioritized evidence addition in $\mathsf{lex}$ models; in both cases, the piece of evidence $\llb\pi\rrb_M$ is added to the initial body of evidence $\mathscr{R}$. Thus, it is easy to see that the proof of the validity of $\text{EA4}_\cap$ in $\cap$-models can be straightforwardly adapted to show the validity of $\oplus$EA4 in $\mathsf{lex}$ models.  
\item Axiom $\oplus${EA5}:
\begin{flalign*}
& M,w\models [\oplus\pi]\Box\chi &\\
\text{ iff } & M^{+ \pi},w\models \Box \chi \\
\text{ iff } & \mathsf{lex}(\mathscr{R}^{\oplus\pi},\preceq^{\oplus\pi})[w]\subseteq \llb \chi \rrb_{M^{+^{\oplus\pi}}}\\
\text{ iff } & \mathsf{lex}(\mathscr{R}^{\oplus\pi},\preceq^{\oplus\pi})[w]\subseteq \llb [\oplus\pi]\chi \rrb_{M}\\
\text{ iff } & (\mathsf{lex}(\lan\mathscr{R},\preceq\ran)\cap \llb\pi\rrb_M )[w]\subseteq \llb[\oplus\pi]\chi \rrb_M \text{ and } \llb\pi\rrb^<_M[w]\subseteq \llb[\oplus\pi]\chi \rrb_M  \ \ ( \text{by Lemma } \ref{lemma 3 for lex})\\
\text{ iff } & M,w\models [\oplus\pi]\chi
\wedge \pi^<(\chi) \wedge \bigwedge_{s\in S_0(I)}( s(\bm{\varphi},\bm{\psi}) \to \Box (\psi_{s_{\mathsf{len}(s)}} \to [\oplus\pi]\chi))  (\text{by Lemmas } \ref{lemma 1 for lex},\ref{lemma 2 for lex} )
\end{flalign*}

\item Axiom $\text{EA6}_\cap$: 
\[M,w\models  [\oplus\pi]\forall\chi \text{ iff }  M^{\oplus\pi},w\models \forall \chi \text{ iff } \llb \chi \rrb_{M^{\oplus\pi}}=W^{\oplus\pi} \text{ iff } \llb  [\oplus\pi] \chi \rrb_{M}=W \text{ iff } M,w \models \forall [\oplus\pi]\chi \]

\end{enumerate}
\end{proof}

\end{claim}

\begin{claim} $\mathsf{L}^\oplus$ is complete w.r.t. $\mathsf{lex}$-models.
\begin{proof}
The proof follows the same steps used above to prove completeness for other proof systems with dynamic modalities.
\end{proof}

\end{claim}

\subsection*{PROOF OF THEOREM 4}

We recall the theorem:

\lc*

The soundness proof is straighforward. We focus on completeness.  The proof is based on the the completeness-via-canonicity approach. In particular, we construct of a canonical \textsf{REL} model for each $\mathsf{L}_{c}$-consistent theory $T_0$. This model is almost identical to the pre-model used in the completeness proof for $\mathsf{L}_\mathsf{lex}$.

\begin{definition}[Canonical model for $T_0$] A canonical model for $T_0$ is a structure  $M^c=\langle W^c, \lan \mathscr{R}^c, \preceq^c\ran, V^c, Ag^c\rangle$ with:
\small
\begin{itemize}
\item $W^c\coloneqq \{ T \mid T \text{ is a maximally consistent theory and } R^\forall T_0 T \}$
\item $\mathscr{R}^c\coloneqq \{ R^{\Box_0\varphi} \mid \varphi\in \mathscr{L}_{c} \text{ and } (\exists\Box_0\varphi)\in T_0 \}$ 
\item $\preceq^c$ is some preorder on $\mathscr{R}^c$
\item $V^c$ is a valuation function given by $V^c(p)\coloneqq \| p \|$
\item $Ag^c$ is an aggregator for $W^c$ with 
$$
Ag^c(\lan\mathscr{R},\preceq\ran)= 
\begin{cases}
R^{\vec{\pi}} & \text{ if } \lan\mathscr{R},\preceq\ran=\lan \mathscr{R}^{c\oplus\llb \vec{\pi}\rrb_{M^c}}, \preceq^{c\oplus\llb \vec{\pi}\rrb_{M^c}}\ran\\
W^c \times W^c & \text{ otherwise }
\end{cases}
$$
\end{itemize}
\normalsize
where:
\small
\begin{itemize}
\item $R^{\forall}$, and each $R^{\Box_0\varphi}$ are defined in the same way as in Definition \ref{canmod}.
\item for each $\varphi\in \mathscr{L}_{c}$, $\| \varphi \|\coloneqq \{T\in W^c \mid \varphi\in T \}$
\item for each sequence of evidence programs $\vec{\pi}$, $R^{\vec{\pi}}\subseteq W^c \times W^c$ is given by: $R^{\vec{\pi}}TS$ iff for all $\varphi\in \mathscr{L}_{c}$: $\Box^{\vec{\pi}}\varphi\in T \Rightarrow \varphi\in S$.
\end{itemize}
\normalsize
\end{definition}

We first show that this canonical model is indeed a \textsf{REL} model.

\begin{proposition} \label{is rel mod} $M^c$ is a \textsf{REL} model.
\begin{proof} In order to show that $M^c$ is an \textsf{REL} model, we have to show that: 
\begin{enumerate}
\item $\mathscr{R}^c$ is a family of evidence, i.e., every $R\in \mathscr{R}$ is a preorder.
\item $R^{\Box_0\top} = W^c\times W^c$, i.e., the trivial evidence order is an element of $\mathscr{R}^c$, as required.
\item $R^{\vec{\pi}}$ is a preorder for each $\vec{\pi}$, and thus $Ag^c$ is well-defined.
\end{enumerate}
The proofs of all three items are analogous to the ones given in Proposition \ref{is rel mod pre}.
\end{proof}
\end{proposition}

Having established that $M^c$ is a \textsf{REL} model, we prove now the standard lemmas to show that the canonical model works as expected.

\begin{lemma}[Existence Lemma for $\forall$] \label{ex forall} $\|\exists\varphi\|\neq\emptyset$ iff $\|\varphi\|\neq\emptyset$.
\begin{proof} Same as in Lemma \ref{ex forall lex}.
\end{proof}
\end{lemma}

\begin{lemma}[Existence Lemma for $\Box^{\vec{\pi}}$] \label{ex box} Let $\vec{\pi}=\lan \pi_1,\dots,\pi_n\ran$ be arbitrary. $T\in \|\Diamond^{\vec{\pi}}\varphi\|$ iff there is an $S\in \|\varphi\|$ such that $R^{\vec{\pi}}TS$.
\begin{proof}  Analogous to the proof of Lemma \ref{ex box lex}.
\end{proof}
\end{lemma}

\begin{lemma}[Existence Lemma for $\Box_0$] \label{ex box0} $T\in \|\Box_0\varphi\|$ iff there is an $R\in\mathscr{R}^c$ such that $R[T]\subseteq \|\varphi\|$.
\begin{proof} Same as in Lemma \ref{ex box0 lex}.
\end{proof}
\end{lemma}

\begin{lemma}[Truth Lemma] \label{truth lemma} For every formula $\varphi\in\mathscr{L}_{c}$, we have: $\llb \varphi\rrb_{M^c}=\|\varphi \|$.
\begin{proof} Analogous to Lemma \ref{truth lemma lex}.
\end{proof}
\end{lemma}

\begin{lemma} \label{lemma comp} $\mathsf{L}_c$ is strongly complete with respect to the class of \textsf{REL} models.
\begin{proof} Analogous to Lemma \ref{lemma comp lex}.
\end{proof}
\end{lemma}

\subsection{PROOF OF THEOREM 5}

We recall the theorem:

\lcplus*

We consider soundness first.

\begin{claim} The axioms \emph{PEA1-PEA6} are valid. \label{sound rel+}
\begin{proof} Let $M=\langle W, \lan \mathscr{R},\preceq\ran, V, Ag\rangle$ be an \textsf{REL} model, $w$ a world in $M$ and $\vec{\pi}=\lan \pi_1,\dots,\pi_n\ran$ be a sequence of evidence programs with normal form $\bigcup_{s\in S_0(I_i)}(?s(\bm{\varphi},\bm{\psi}); A ; ?\psi_{s_{\mathsf{len}(s)}}) \cup (?\top)$ for each $i\in\{1,\dots,n\}$.
\begin{enumerate}
\item Axiom PEA4: We first prove the following:
\begin{claim} There is an $k\in\{1,\dots,n\}$ such that $\llb \pi_k\rrb_M[w]\subseteq \llb [\oplus \vec{\pi}]\chi \rrb_{M}$ iff $M,w\models [\oplus \vec{\pi}]\chi \wedge \bigvee_{1\leq i\leq n}(\bigwedge_{s\in S_0(I_i)}( s(\bm{\varphi},\bm{\psi}) \to \forall (\psi_{s_{\mathsf{len}(s)}} \to [\oplus \vec{\pi}]\chi)))$.
\begin{proof}
($\Rightarrow$) Suppose there is an $k\in\{1,\dots,n\}$ such that $\llb \pi_k\rrb_M[w]\subseteq \llb [\oplus \vec{\pi}]\chi \rrb_{M}$. As $\pi_k$ is an evidence program, $\llb \pi_k\rrb_M$ is reflexive and thus $M,w\models [\oplus \vec{\pi}]\chi$. It remains to be shown that 
\begin{equation}\label{eq k}
M,w\models \bigvee_{1\leq i\leq n}(\bigwedge_{s\in S_0(I_i)}( s(\bm{\varphi},\bm{\psi}) \to \forall (\psi_{s_{\mathsf{len}(s)}} \to [\oplus \vec{\pi}]\chi)))
\end{equation}
To show (\ref{eq k}), it suffices to find one $i\in\{1,\dots,n\}$ such that 
\[M,w\models \bigwedge_{s\in S_0(I_i)}( s(\bm{\varphi},\bm{\psi}) \to \forall (\psi_{s_{\mathsf{len}(s)}} \to [\oplus \vec{\pi}]\chi))\]
Consider $i=k$, take any $s\in S_0(I_k)$ and suppose that $M,w\models s(\bm{\varphi},\bm{\psi})$. We need to show that $M,w\models \forall (\psi_{s_{\mathsf{len}(s)}} \to [\oplus \vec{\pi}]\chi)$. Take any $v\in W$ and suppose $M,v\models \psi_{s_{\mathsf{len}(s)}}$. If we show that $M,v\models[\oplus \vec{\pi}]\chi$, we are done. Given $M,w\models s(\bm{\varphi},\bm{\psi})$ and $M,v\models \psi_{s_{\mathsf{len}(s)}}$, by Proposition \ref{x sees y}, we have $(w,v)\in \llb ?s(\bm{\varphi},\bm{\psi}); A ; ?\psi_{s_{\mathsf{len}(s)}}\rrb_M$. Thus 
\[(w,v)\in \bigcup_{s\in S_0(I_k)}\llb?s(\bm{\varphi},\bm{\psi}); A ; ?\psi_{s_{\mathsf{len}(s)}}\rrb_M\]
Hence as 
\small
\begin{flalign*}
\llb\pi_k\rrb_M & = \llb\bigcup_{s\in S_0(I_k)}\big(?s(\bm{\varphi},\bm{\psi}); A ; ?\psi_{s_{\mathsf{len}(s)}}\big) \cup (?\top)\rrb_M \\
 & = \llb\bigcup_{s\in S_0(I_k)}\big(?s(\bm{\varphi},\bm{\psi}); A ; ?\psi_{s_{\mathsf{len}(s)}}\big)\rrb_M \cup \llb ?\top\rrb_M \\
 & = \bigcup_{s\in S_0(I_k)}\llb?s(\bm{\varphi},\bm{\psi}); A ; ?\psi_{s_{\mathsf{len}(s)}}\rrb_M \cup \llb ?\top\rrb_M 
\end{flalign*}
\normalsize
we have $(w,v)\in \llb\pi_k\rrb_M$. Hence, given $\llb \pi_k\rrb_M[w]\subseteq \llb [\oplus \vec{\pi}]\chi \rrb_{M}$ we have $M,v\models [\oplus \vec{\pi}]\chi$, as required.

$(\Leftarrow$) Suppose that $M,w\models [\oplus \vec{\pi}]\chi \wedge \bigvee_{1\leq i\leq n}(\bigwedge_{s\in S_0(I_i)}( s(\bm{\varphi},\bm{\psi}) \to \forall (\psi_{s_{\mathsf{len}(s)}} \to [\oplus \vec{\pi}]\chi)))$. Then there is some $k\in\{1,\dots,n\}$ such that 
\begin{equation}\label{eq conj}
M,w\models \bigwedge_{s\in S_0(I_k)}( s(\bm{\varphi},\bm{\psi}) \to \forall (\psi_{s_{\mathsf{len}(s)}} \to [\oplus \vec{\pi}]\chi)))
\end{equation} 
We will show that $\llb \pi_k\rrb_M[w]\subseteq \llb [\oplus \vec{\pi}]\chi \rrb_{M}$. Take any $v$ and suppose $(w,v)\in \llb \pi_k\rrb_M$. We need to show that $v\in \llb [\oplus \vec{\pi}]\chi \rrb_{M}$. If $v=w$, given $M,w\models [\oplus \vec{\pi}]\chi$ we are done. So suppose $v\neq w$. Note that 
\small
\begin{flalign*}
& (w,v) \in \llb\pi_k\rrb_M\\
\text{ iff } & (w,v)\in \llb\bigcup_{s\in S_0(I_k)}?\big(s(\bm{\varphi},\bm{\psi}); A ; ?\psi_{s_{\mathsf{len}(s)}}\big) \cup (?\top)\rrb_M\\
\text{ iff } & (w,v)\in  \llb\bigcup_{s\in S_0(I_k)}\big(?s(\bm{\varphi},\bm{\psi}); A ; ?\psi_{s_{\mathsf{len}(s)}}\big)\rrb_M\text{ or } (w,v)\in\llb?\top\rrb_M\\
\text{ iff } & (w,v)\in  \llb\bigcup_{s\in S_0(I_k)}\big(?s(\bm{\varphi},\bm{\psi}); A ; ?\psi_{s_{\mathsf{len}(s)}}\big)\rrb_M \text{ or } w=v\\
\text{ iff } & (w,v)\in  \llb\bigcup_{s\in S_0(I_k)}\big(?s(\bm{\varphi},\bm{\psi}); A ; ?\psi_{s_{\mathsf{len}(s)}}\big)\rrb_M \ \ \ (\text{as } w\neq v \text{ by assumption })\\
\text{ iff } & (w,v)\in \bigcup_{s\in S_0(I_k)} \llb?s(\bm{\varphi},\bm{\psi}); A ; ?\psi_{s_{\mathsf{len}(s)}}\rrb_M\\
\text{ iff } &\text{ for some } s'\in S_0(I_k), (w,v)\in \llb?s'(\varphi); A ; ?\psi_{s'_{\mathsf{len}(s')}}\rrb_M\\
\text{ iff } &\text{ for some } s'\in S_0(I_k), w\in \llb s'(\varphi)\rrb_M \text{ and } v\in\llb\psi_{s'_{\mathsf{len}(s')}}\rrb_M \ \ \ \text{(by Prop. } \ref{x sees y})
\end{flalign*}
\normalsize
Given (\ref{eq conj}), we have in particular
\[M,w\models s'(\varphi) \to \forall (\psi_{s'_{\mathsf{len}(s)}} \to [\oplus \vec{\pi}]\chi)\]
Thus from $w\in \llb s'(\varphi)\rrb_M$ we get $M,w\models \forall (\psi_{s'_{\mathsf{len}(s)}} \to [\oplus \vec{\pi}]\chi)$. And given $v\in\llb\psi_{s'_{\mathsf{len}(s')}}\rrb_M$ we get $M,v\models [\oplus \vec{\pi}]\chi$, as required.
\end{proof}
\end{claim}
Given the Claim, we have 
\small
\begin{flalign*}
& M,w\models [\oplus \vec{\pi}]\Box_0\chi &\\
\text{ iff } & M^{\oplus \vec{\pi}},w\models \Box_0 \chi \\
\text{ iff } & \text{there is an } R\in \mathscr{R}\cup \{\llb \pi_i \rrb_M \mid i=1,\dots, n\} \text{ such that } R[w]\subseteq \llb \chi \rrb_{M^{\oplus \vec{\pi}}}\\
\text{ iff } & \text{there is an } R\in \mathscr{R}\cup \{\llb \pi_i \rrb_M \mid i=1,\dots, n\}\text{ such that } R[w]\subseteq \llb [\oplus \vec{\pi}]\chi \rrb_{M}\\
\text{ iff } & \text{there is an } R\in \mathscr{R} \text{ such that } R[w]\subseteq \llb [\oplus \vec{\pi}]\chi \rrb_{M}\\ 
& \text{ or }\text{there is an } i\in\{1,\dots, n\} \text{ such that } \llb \pi_i\rrb_M[w]\subseteq \llb [\oplus \vec{\pi}]\chi \rrb_{M}\\
\text{ iff } & M,w\models\Box_0 [\oplus \vec{\pi}]\chi\\ 
& \text{ or } M,w\models [\oplus \vec{\pi}]\chi \wedge \bigvee_{1\leq i\leq n}(\bigwedge_{s\in S_0(I_i)}( s(\bm{\varphi},\bm{\psi}) \to \forall (\psi_{s_{\mathsf{len}(s)}} \to [\oplus \vec{\pi}]\chi))) \\
& (\text{ by the Claim above)}\\
\text{ iff } & M,w\models\Box_0 [\oplus \vec{\pi}]\chi\lor \big([\oplus \vec{\pi}]\chi \wedge \bigvee_{1\leq i\leq n}(\bigwedge_{s\in S_0(I_i)}( s(\bm{\varphi},\bm{\psi}) \to \forall (\psi_{s_{\mathsf{len}(s)}} \to [\oplus \vec{\pi}]\chi)))\big)
\end{flalign*}
\normalsize
\item Axiom EA5:
$M,w\models [\oplus \vec{\pi}]\Box^{\vec{\rho}}\chi$ iff $M^{\oplus \vec{\pi}},w\models \Box^{\vec{\rho}}\chi$ iff  $\lan W^{\oplus \llb\vec{\pi}\rrb_M}, \lan \mathscr{R}^{\oplus \llb\vec{\pi}\rrb_M}, \preceq^{\oplus \llb\vec{\pi}\rrb_M}\ran,V^{\oplus \llb\vec{\pi}\rrb_M},Ag^{\oplus \llb\vec{\pi}\rrb_M}\ran,w \models \Box^{\vec{\rho}}\chi$  iff  
\[\lan W^{\oplus \llb\vec{\pi}\rrb_M\oplus\llb\vec{\rho}\rrb_M}, \lan \mathscr{R}^{\oplus \llb\vec{\pi}\rrb_M\oplus \llb\vec{\rho}\rrb_M}, \preceq^{\oplus \llb\vec{\pi}\rrb_M\oplus \llb\vec{\pi}\rrb_M}\ran,V^{\oplus \llb\vec{\pi}\rrb_M\oplus \llb\vec{\rho}\rrb_M},Ag^{\oplus \llb\vec{\pi}\rrb_M\oplus \llb\vec{\pi}\rrb_M}\ran,w \models \chi\]
iff  $M,w \models \Box^{\vec{\pi}\oplus \vec{\rho}}\chi$
\item Axiom EA6: $M,w\models [\oplus \vec{\pi}]\forall\chi$ iff $M^{\oplus \vec{\pi}},w\models \forall \chi$ iff $\llb \chi \rrb_{M^{+\oplus \vec{\pi}}}=W^{+\oplus \vec{\pi}}$ iff $\llb [\oplus \vec{\pi}] \chi \rrb_{M}=W$ iff $M,w \models \forall [\oplus \vec{\pi}]\chi$.
\end{enumerate}
\end{proof}
\end{claim}

\begin{claim} $\mathsf{L}_c^\oplus$ is complete w.r.t. \textsf{REL} models.
\begin{proof}
Once we have established the validity of the reduction axioms, the proof is standard and follows the same steps used above to prove completeness for other proof systems with dynamic modalities.
\end{proof}
\end{claim}

\pagebreak

%
%

\bibliographystyle{splncs_srt}


\end{document}